\documentclass[onecolumn, 11 point, journal]{IEEEtran}

\usepackage{color}
\usepackage{times}
\usepackage{amssymb}
\usepackage{amsmath}
\usepackage{mathdots}
\usepackage{epsfig}
\usepackage{cite}
\usepackage{verbatim}
\usepackage{enumerate}
\usepackage[bottom]{footmisc}
\usepackage{subfigure}
\usepackage{amsmath}
\usepackage{epsf,bm}
\usepackage{epstopdf}
\usepackage{setspace}
\usepackage{algpseudocode}
\usepackage{graphicx,dblfloatfix}
\usepackage{algorithm}
\makeatletter
\renewcommand{\ALG@name}{Construction}
\makeatother
\usepackage{cases}
\usepackage{tikz}
\usepackage{url}
\usepackage{wrapfig}
\usepackage{blkarray}
\makeatletter 
\newbox\BA@first@box
\makeatother

\newtheorem{theorem}{\bf Theorem}[section]

\newtheorem{lemma}[theorem]{\bf Lemma}

\newtheorem{definition}{\bf Definition}[section]
\newtheorem{example}{\bf Example}[section]
\newtheorem{remark}{\bf Remark}[section]
\newtheorem{claim}[theorem]{\bf Claim}
\newcommand{\qedw}{\hfill \ensuremath{\Box}}
\newcommand{\qed}{\nobreak \ifvmode \relax \else
      \ifdim\lastskip<1.5em \hskip-\lastskip
      \hskip1.5em plus0em minus0.5em \fi \nobreak
      \vrule height0.75em width0.5em depth0.25em\fi}
\newenvironment{proof}[1][Proof:]{\begin{trivlist}
\item[\hskip \labelsep {\bfseries #1}]}{\qedw\end{trivlist}}

\newcommand{\vc}[1]{{\textcolor{black}{#1}}}
\newcommand{\bl}[1]{{\textcolor{black}{#1}}}
\newcommand{\bll}[1]{{\textcolor{black}{#1}}}
\newcommand{\fbl}[1]{{\textcolor{black}{#1}}}
\newcommand{\remove}[1]{}
\newcommand{\Out}{\text{Out}}
\begin{document}
\title{Numerically Stable Polynomially Coded Computing}
\fontfamily{cmr} \selectfont
\author{\large Mohammad Fahim and Viveck R. Cadambe\thanks{
M. Fahim and V. Cadambe are with the Department of Electrical Engineering, Pennsylvania State University, University Park, PA 16802.}\thanks{This work will be presented in part at the IEEE International Symposium on Information Theory (ISIT), July 2019.}
}
\date{}
\maketitle
\vspace{-0pt}
\allowdisplaybreaks{

\begin{abstract}
We study the numerical stability of polynomial based encoding methods, which has emerged to be a powerful class of techniques for providing straggler and fault tolerance in the area of coded computing. Our contributions are as follows: 
\begin{enumerate}
    \item  We construct new codes for matrix multiplication that achieve the same fault/straggler tolerance as the previously constructed \emph{MatDot Codes} \bll{and \emph{Polynomial Codes}}. Unlike previous codes that use polynomials expanded in a monomial basis, our codes use a basis of orthogonal polynomials. 
 
 \item We show that the condition number of every $m \times m$ sub-matrix of an $m \times n, n \geq m$ Chebyshev-Vandermonde matrix, evaluated on the $n$-point Chebyshev grid, grows as $O(n^{2(n-m)})$ for $n > m$. An implication of this result is that, when Chebyshev-Vandermonde matrices are used for coded computing, for a fixed number of redundant nodes $s=n-m,$ the condition number grows at most polynomially in the number of nodes $n$. 
 
 \item By specializing our orthogonal polynomial based constructions to Chebyshev polynomials, and using our condition number bound for Chebyshev-Vandermonde matrices, we construct new numerically stable techniques for coded matrix multiplication. We empirically demonstrate that our constructions have significantly lower numerical errors compared to previous approaches which involve inversion of Vandermonde matrices. We generalize our constructions to explore the trade-off between computation/communication and fault-tolerance.
 \item We propose a numerically stable specialization of Lagrange coded computing. Motivated by our condition number bound, our approach involves the choice of evaluation points and a suitable decoding procedure that involves inversion of an appropriate Chebyshev-Vandermonde matrix. Our approach is demonstrated empirically to have lower numerical errors as compared to standard methods.
\end{enumerate}
\end{abstract}
}


\section{Introduction}
The recently emerging area of ``coded computing'' focuses on incorporating redundancy based on coding-theory-inspired strategies to tackle central challenges in distributed computing, including stragglers, failures, processing errors, communication bottlenecks and security issues. 
Such ideas have been applied to different large scale distributed computations such as matrix multiplication \cite{dutta2016short,polynomialcodes, allerton17,arxiv_allerton17,genPolyDot}, gradient methods \cite{tandon2016gradient,tan17,abbegrad}, linear solvers \cite{YangGK17,crit, maity2018robust}  and multi-variate polynomial evaluation \cite{yu2018lagrange}. An important idea that has emerged from this body of the work is the use of novel, Reed-Solomon like  \emph{polynomial} based methods for encoding data. In polynomial based methods, each computation node stores a linearly encoded combination of the data partitions, where data stored at different worker nodes can be interpreted as evaluation of an appropriate polynomial at different points. The nodes then perform computation on these encoded versions of the data, and a central master/fusion node aggregates the outputs of these computations to recover the overall result via a decoding process that inevitably involves polynomial interpolation. Much like Reed Solomon Codes, if the number of nodes performing the computation is higher than the number of evaluation points required for accurate interpolation, the overall computation is tolerant to faults and stragglers. 

Perhaps the most striking application of polynomial based methods comes in the context of matrix multiplication. To multiply two $N \times N$ matrices $\mathbf{A},\mathbf{B},$ assuming that each node stores $1/m$ of each matrix, classical work in algorithm based fault tolerance \cite{Huang_TC_84} outlines a coding based method which has been analyzed in \cite{ProductCodes}. 
Reference \cite{polynomialcodes} showed through \emph{polynomial} based encoding methods that the result of just $m^2$ nodes can be used by the master node to recover the matrix-product. Remarkably, this means that polynomial based codes ensure that  \emph{the recovery threshold} - the worst case number of nodes whose computation suffices to recover the overall matrix-product - does not grow with $P$, the number of the distributed system's worker nodes, unlike the approaches of \cite{Huang_TC_84, ProductCodes}. The recovery threshold for matrix multiplication has been improved to $2m-1$ via a code construction called MatDot Codes in \cite{allerton17}, albeit at a higher communication/computation cost than codes in \cite{polynomialcodes}. 
A second prominent application of polynomial based methods is the idea of \emph{Lagrange coded computing} \cite{yu2018lagrange}, where coding is applied for multi-variate polynomial computing with guarantees of straggler resilience, security and privacy. In addition, polynomial-based methods are also useful for communication-efficient approaches for inverse problems and gradient methods \cite{abbegrad, li2018polynomially, crit}. 


 Despite the enormous success, the scalability of polynomial based methods in practice are limited by an ``inconvenient truth'', their numerical instability. The decoding methods for polynomial based methods require interpolating a degree $K-1$ polynomial using $K$ evaluation points. While this is numerically stable for classical error correcting codes for communication and storage which are implemented over finite fields, we are concerned here for data processing applications where the operations are typically real-valued. The main reason for the instability is that either implicitly or explicitly, interpolation effectively solves a linear system whose transform is characterized by a Vandermonde matrix. It is well known that the condition number of Vandermonde matrices with real-valued nodes grows exponentially in the dimension of the matrix \cite{gautschi1987lower, gautschi1990stable, gautschi1974norm, reichel1991chebyshev}. The large condition number means that small perturbations of the Vandermonde matrix due to numerical precision errors can result in singular matrices \cite{quarteroni2010numerical, trefethen2013approximation}. In practice, this can translate to large numerical errors even when the coded computation is distributed among few tens of nodes\footnote{For example, \cite{seth2018bigdata}, reports that ``In our experiments we observed large floating
point errors when inverting high degree Vandermonde matrices
for polynomial interpolation''.}.  \bll{Conventional intuition dictates that the main scalability bottlenecks in distributed computing include  computation cost per worker, communication bottlenecks, and  stragglers. However, for \emph{polynomially coded computing},  it turns out that numerically stability is also critical and constitutes  a huge bottleneck for scalability of  such codes.
Indeed, a polynomially coded computing scheme that achieves the minimum recovery threshold, and that is optimal computation/communication wise, will simply fail once implemented on a distributed system with tens of computing nodes due to the large numerical errors. 
Thus, the main contribution of our paper is a new numerically stable approach to polynomially coded computing.}


\section{Summary of Contributions}
 In this paper, we develop a new, numerically stable, approach for polynomially coded computing. A significant difference from previous polynomial coding approaches is that we depart from the monomial basis, which allows us to circumvent the inherently ill-conditioned Vandermonde-matrices. We demonstrate our approach through two important applications of polynomially coded computing: matrix multiplication, and Lagrange coded computing.

To illustrate our results, consider the coded matrix multiplication problem, where the goal is to multiply two matrices $\mathbf{A},\mathbf{B}$ over $P$ computation nodes where each node stores $1/m$ of each of the two matrices.  A master node encodes $\mathbf{A},\mathbf{B}$ into $P$ matrices each, and sends these matrices respectively to each worker node. Each worker node multiplies the received encoded matrices, and sends the product back to the fusion node\footnote{The master and fusion nodes are logical entities; in practice, they may be the same node, or may be emulated in a decentralized manner by the computation nodes.}, which aims to recover $\mathbf{A}\mathbf{B}$ from a subset of the worker nodes. The recovery threshold is defined as a number $K$ such that the computation of any set of $K$ worker nodes suffices to recover the product $\mathbf{A}\mathbf{B}.$   The MatDot scheme of \cite{allerton17} achieves the best known recovery threshold of $2m-1$. We begin with an example of MatDot Codes for $m=2.$

\textbf{Example 1: MatDot Codes \cite{allerton17}, recovery threshold = 3:}
\emph{Consider two $N \times N$ matrices} $$\mathbf{A}=\begin{bmatrix}\mathbf{{A}}_1 & \mathbf{{A}}_{2}\end{bmatrix},~~ \mathbf{B} = \begin{bmatrix}\mathbf{{B}}_1 \\ \mathbf{{B}}_{2}\end{bmatrix},$$ where $\mathbf{A}_{1},\mathbf{A}_{2}$ are $N \times N/2$ matrices and $\mathbf{B}_{1},\mathbf{B}_{2}$ are $N/2 \times N$ matrices. \emph{
 Define $p_{\mathbf{A}}(x) = \mathbf{A}_{1} + \mathbf{A}_{2}x$ and $p_{\mathbf{B}}(x) = \mathbf{B}_{1}x+\mathbf{B}_{2},$ and let $x_1, \cdots, x_P$ be distinct real values. Notice that $\mathbf{AB}=\mathbf{A}_{1}\mathbf{B}_1+\mathbf{A}_2\mathbf{B}_2$ is the coefficient of $x$ in polynomial $p_{\mathbf{A}}(x)p_{\mathbf{B}}(x)$. In MatDot Codes, as illustrated in Fig. \ref{fig:Ex1}, worker node $i$ computes $p_{\mathbf{A}}(x_i)p_{\mathbf{B}}(x_i),~ i=1,2, \ldots P,$ so that from any $3$ of the $P$ nodes, the polynomial $p(x) = \mathbf{{A}}_{1}\mathbf{{B}}_2+(\mathbf{{A}}_{1}\mathbf{{B}}_{1} + \mathbf{{A}}_{2}\mathbf{{B}}_{2}) x + \mathbf{{A}}_2\mathbf{{B}}_1 x^{2}$ can be interpolated. Having interpolated the polynomial, the product $\mathbf{A}\mathbf{B}$ is simply the coefficient of $x$.}
\begin{figure}[t]
    \centering
    \includegraphics[scale=0.450]{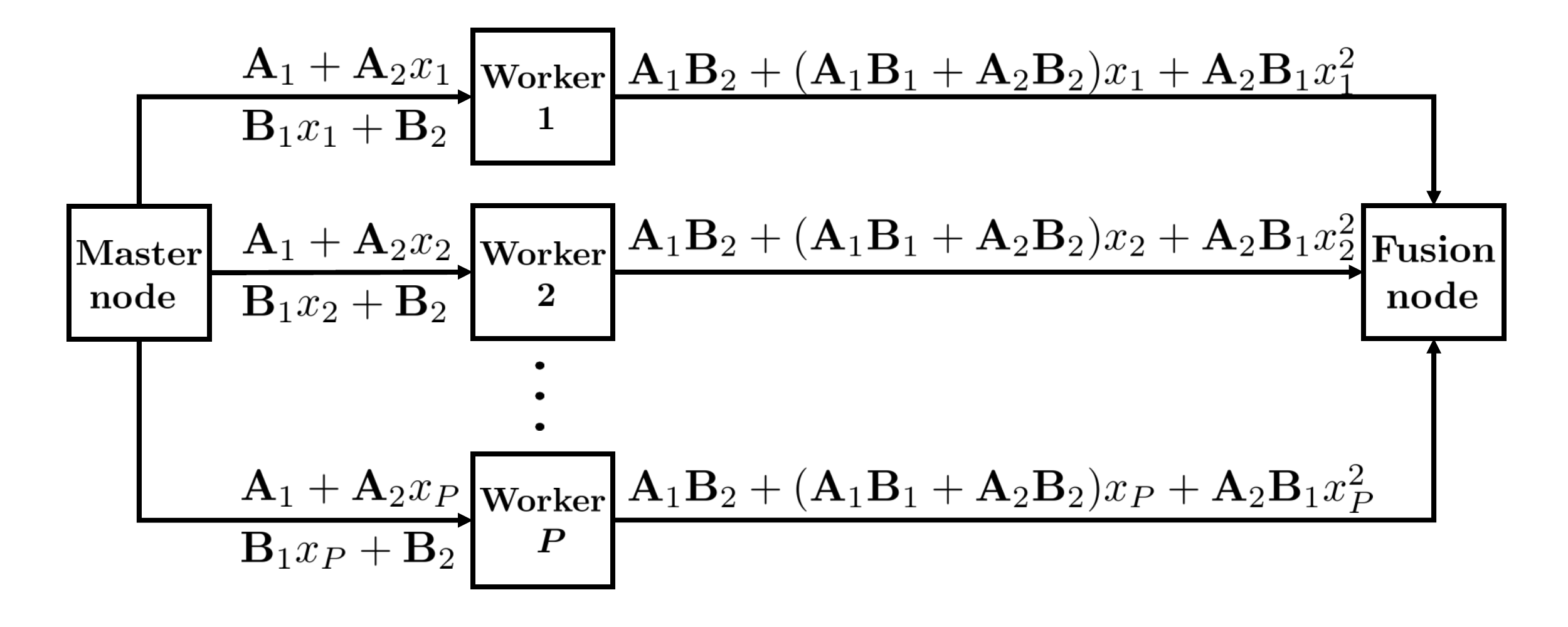}
    \caption{Example of MatDot Codes \cite{allerton17}, with a recovery threshold of $3$. The matrix product $\mathbf{A}\mathbf{B}$ is the coefficient of $x$ in $p_\mathbf{A}(x)p_\mathbf{B}(x)$, and can be recovered at the fusion node upon receiving the output of any $3$ worker nodes and interpolating $p_\mathbf{A}(x)p_\mathbf{B}(x)$. }
    \label{fig:Ex1}
\end{figure}
A generalization of the above example leads to a recovery threshold of $2m-1$, with a decoding process that involves effectively inverting a $2m-1 \times 2m-1$ Vandermonde matrix. It has been shown that the condition number of the $n \times n$ Vandermonde matrix grows exponentially in $n$ with both $\ell_{\infty}$ and $\ell_{2}$ norms  \cite{gautschi1987lower, gautschi1990stable}. The intuition behind the inherent poor conditioning of the monomial basis $\{1,x,x^{2},\ldots, x^{2m-1}\}$ is demonstrated in Fig. \ref{fig:monomialplot} and Fig. \ref{fig:monomialvec}. 

\begin{figure*}[!t]
    \centering
    \begin{minipage}[b]{0.4\textwidth}
    \includegraphics[width=\textwidth]{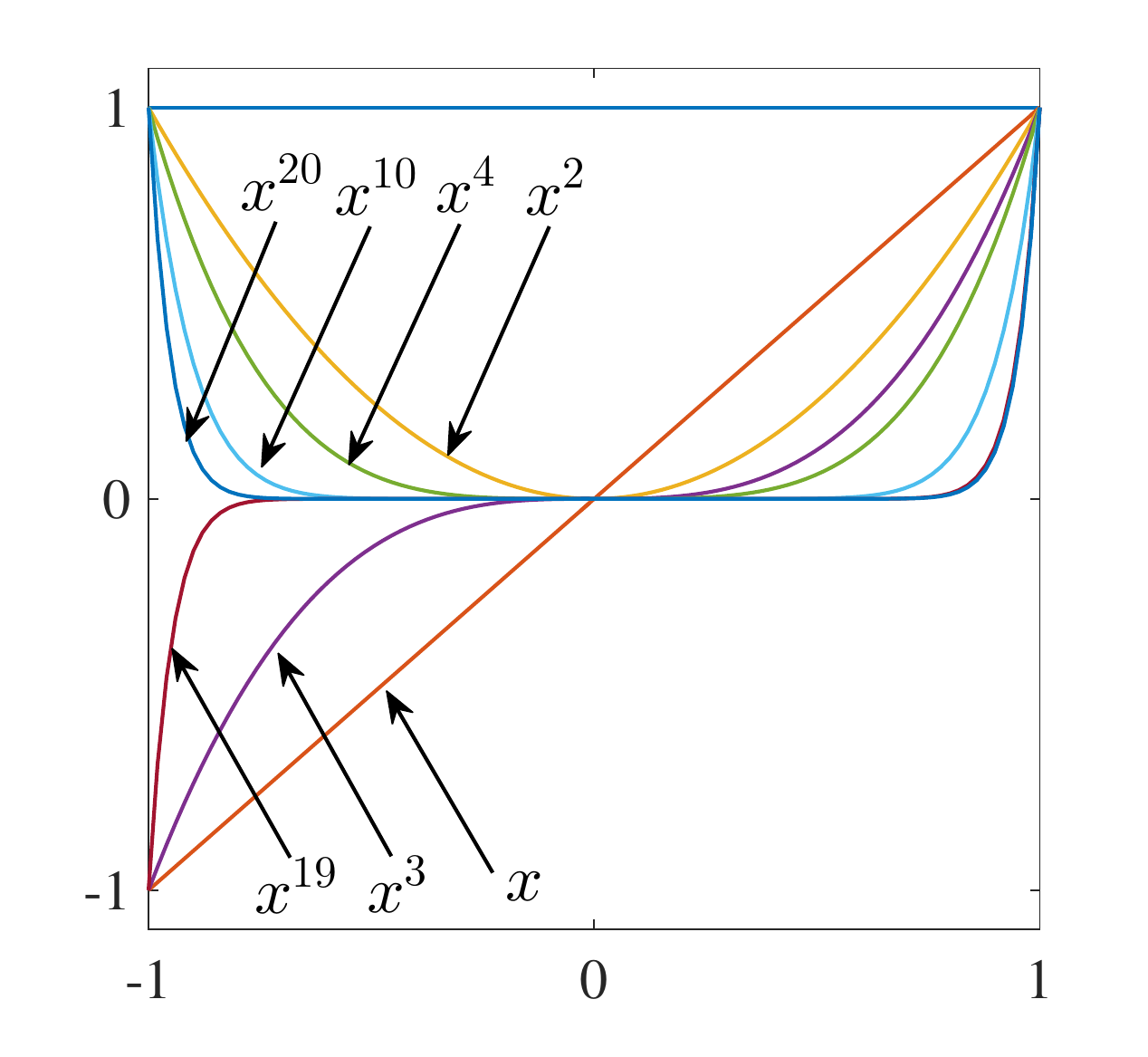}
    \caption{Plot of monomials $1,x,x^2,x^3,x^4,x^{10},x^{19},x^{20}$ versus $x$ for $x\in [-1,1]$. Note that for a large  degree $d,$ small changes in $x$ can lead to large changes in $x^{d};$ this leads to significant numerical errors when working with the monomial basis.}
    \label{fig:monomialplot}
  \end{minipage}
  \hfill
  \begin{minipage}[b]{0.42\textwidth}
    \includegraphics[width=\textwidth]{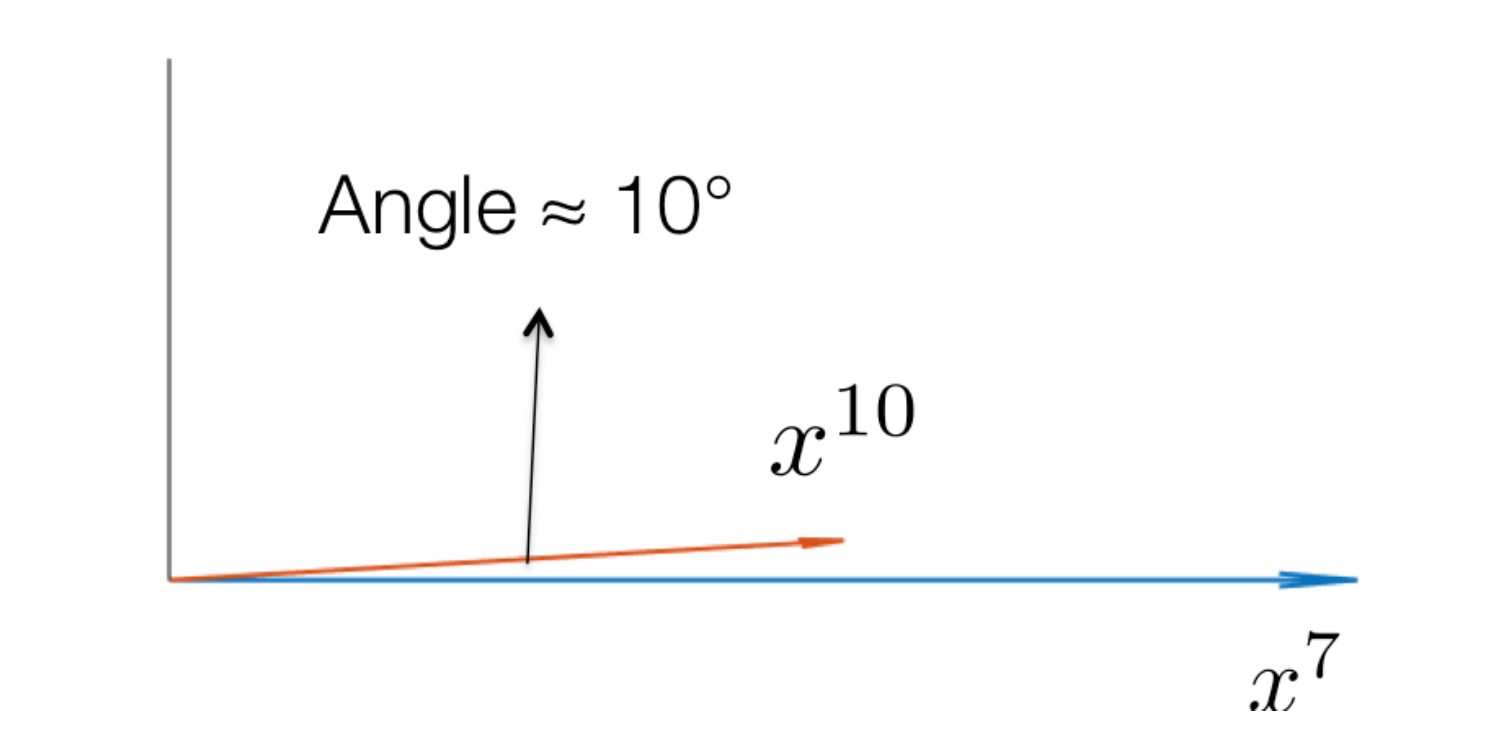}
    \caption{Note that $\{1,x,\ldots,x^{d}\}$ forms a basis for the vector space of $d$-degree polynomials, with the inner-product $\langle f,g\rangle=\int_{-1}^{1}f(x)g(x)dx.$ We have plotted the vectors $x^7$ and $x^{10}$. The small angle between the two vectors leads to numerical errors.}
    \label{fig:monomialvec}
    \hfill
  \end{minipage}
  \begin{minipage}[b]{0.44\textwidth}
  \includegraphics[width=\textwidth]{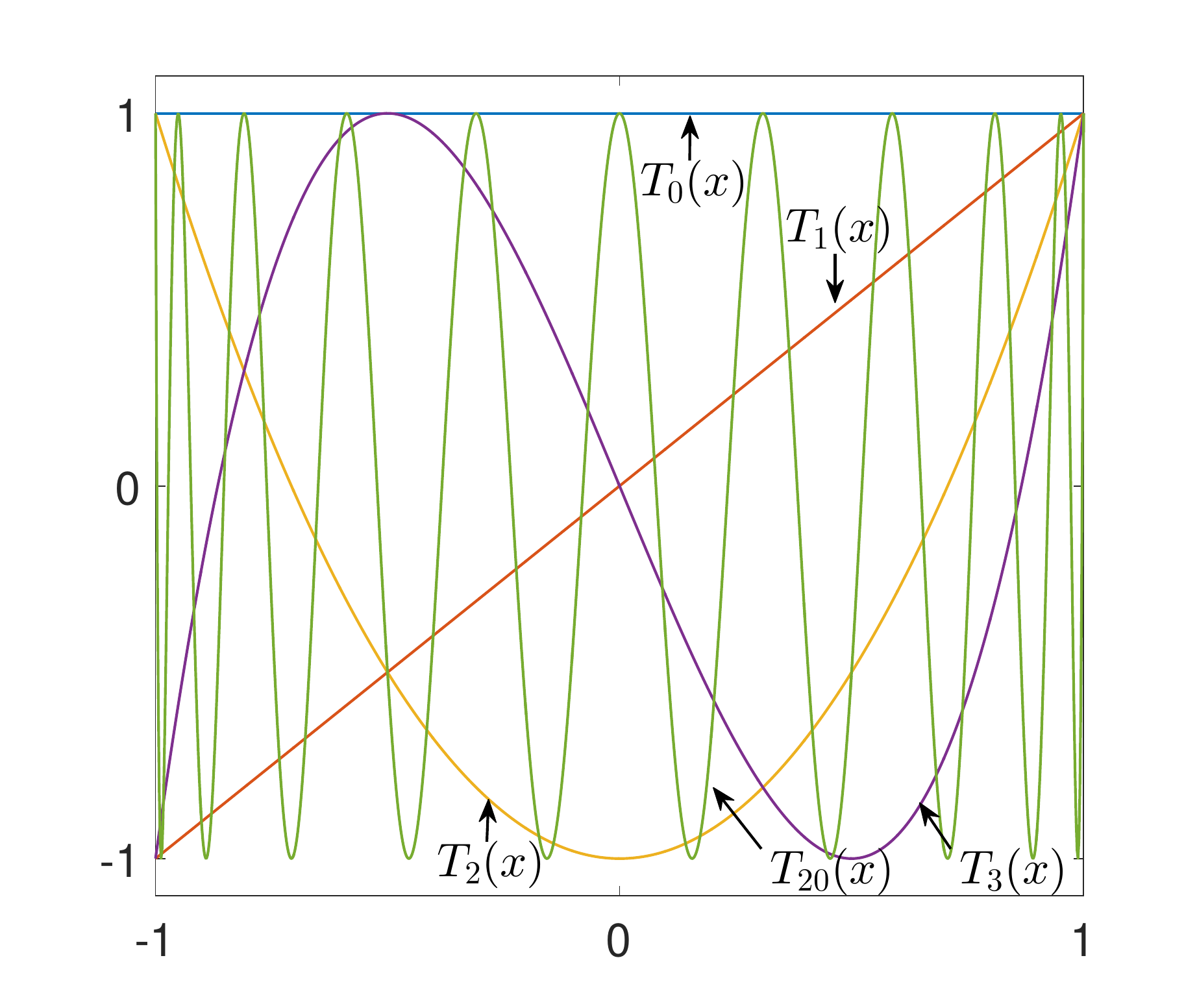}
  \caption{Plot of Chebyshev polynomials $T_0(x),T_1(x),T_2(x),T_3(x),T_{20}(x)$ versus $x$ for $x\in [-1,1]$. }
   \label{fig:chebplot}
  \end{minipage}
\end{figure*}

Motivated by Fig.\ref{fig:monomialvec}, we aim, in this paper, to choose polynomials that are orthonormal.
However, it is not immediately clear whether orthonormal polynomials are applicable for matrix multiplications. We demonstrate the applicability of orthonormal codes for matrix multiplication. For the example below, let $q_0(x), q_1(x)$ denote two orthonormal polynomials such that \begin{equation}
\int_{-1}^{1} q_i(x) q_j(x) dx = \left\{\begin{array}{cc} 0 & \textrm{if }i=j \\ 1 & \textrm{otherwise} \end{array}\right.\label{eq:ortho}
\end{equation}where $q_i(x),i=0,1$ has degree $i$.

\vspace{10mm}
\textbf{Example 2 : \fbl{OrthoMatDot Codes [This paper]}, recovery threshold = 3:}
\emph{For two $N \times N$ matrices $\mathbf{A}=\begin{bmatrix}\mathbf{{A}}_1 & \mathbf{{A}}_{2}\end{bmatrix}, \mathbf{B} = \begin{bmatrix}\mathbf{{B}}_1 \\ \mathbf{{B}}_{2}\end{bmatrix},$
let $p_{\mathbf{A}}(x) = \mathbf{A}_{1} q_0(x) + \mathbf{A}_{2}q_1(x)$ and $p_{\mathbf{B}}(x) = \mathbf{B}_{1}q_0(x)+\mathbf{B}_{2} q_1(x).$ Notice that because of (\ref{eq:ortho}), we have $$\mathbf{A}\mathbf{B} = \int_{-1}^{1} p_{\mathbf{A}}(x) p_{\mathbf{B}}(x) dx.$$  This leads to the following coded computing scheme: worker node $i$ computes $p_{\mathbf{A}}(x_i)p_{\mathbf{B}}(x_i),~ i=1,2, \ldots P,$ where $x_1,\cdots, x_P$ are distinct real values, so that from any $3$ of the $P$ nodes, the fusion node can interpolate $p(x) = p_{\mathbf{A}}(x)p_{\mathbf{B}}(x)$. Having interpolated the polynomial, the fusion node obtains the product $\mathbf{A}\mathbf{B}$ by performing $\int_{-1}^{1}p_{\mathbf{A}}(x)p_{\mathbf{B}}(x)dx$. This example is illustrated in Fig. \ref{fig:Ex2}}.
\begin{figure}[t]
    \centering
    \includegraphics[scale=0.450]{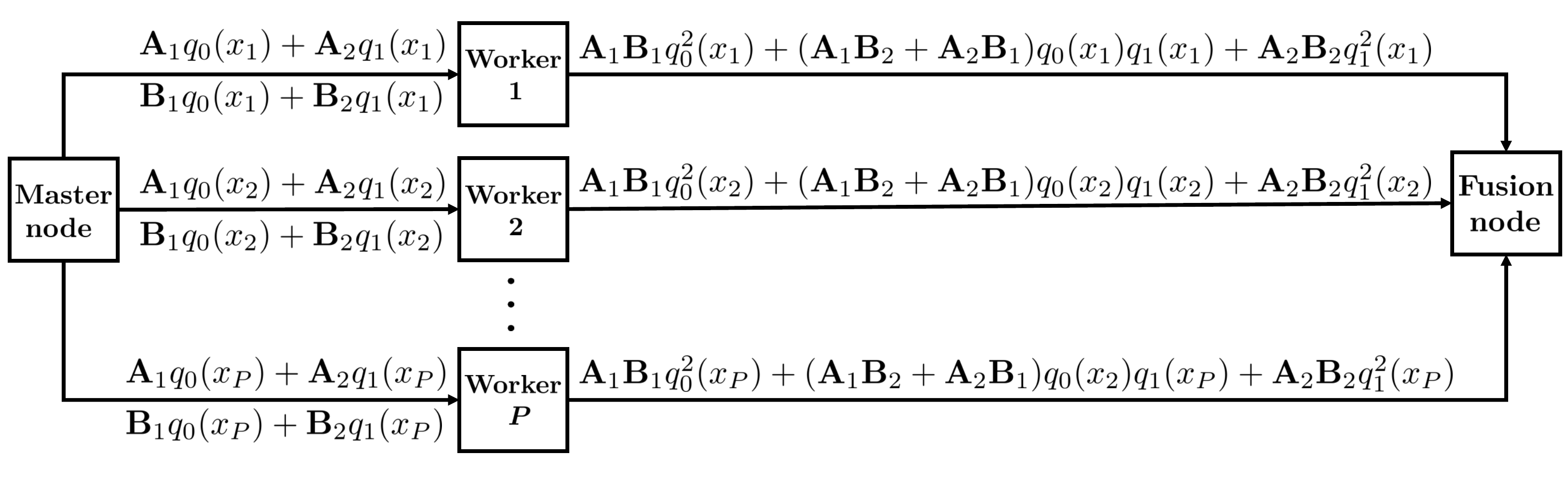}
    \caption{Example of our proposed orthonormal polynomials based codes, with a recovery threshold of $3$. The matrix product $\mathbf{A}\mathbf{B}$ is $\int_{-1}^{1} p_\mathbf{A}(x)p_\mathbf{B}(x)dx $, and can be recovered at the fusion node upon receiving the output of any $3$ worker nodes, then interpolating $p_\mathbf{A}(x)p_\mathbf{B}(x)$, and computing the integral $\int_{-1}^{1} p_\mathbf{A}(x)p_\mathbf{B}(x)dx.$}
    \label{fig:Ex2}
\end{figure}

A simple generalization of the above example, described in Construction \ref{con:orthpoly} in Section \ref{sec:ortho}, leads to a class of codes, \fbl{we refer to it as \emph{OrthoMatDot Codes,}} with recovery threshold of $2m-1$, the same recovery threshold as MatDot Codes. In general, orthonormal polynomials are defined over arbitrary weight measure $\int_{-1}^{1} {\bf \cdot}~ w(x) dx;$ some well known classes of polynomials corresponding to different weight measures $w(x)$ include Legendre, Chebyshev, Jacobi and Laguerre Polynomials \cite{quarteroni2010numerical, trefethen2013approximation} (See Section \ref{sec:prem} for definitions). Our \fbl{OrthoMatDot} Codes in Section \ref{sec:ortho} can use any weight measure, and therefore can be used with different classes of orthonormal polynomials. Of particular interest to our paper are the Chebyshev polynomials (Fig. \ref{fig:chebplot}).

With our basic template, the task of developing numerically stable codes boils down to (A) interpolating $p_{\mathbf{A}}(x)p_{\mathbf{B}}(x)$ in a numerically stable manner, and (B) integrating this polynomial in a numerically stable manner. For task (B), we use a decoding procedure via \emph{Gauss Quadrature} \cite{quarteroni2010numerical, num97, trefethen2013approximation} to recover the integral. Task (A) is particularly challenging in the coding setting, because our goal is to interpolate the coefficients of $p_\mathbf{A}(x)p_{\mathbf{B}}(x)$ - expanded over a series of orthonormal polynomials - from any $2m-1$ points among a set of $P$ points. 

In Section \ref{sec:Chebyshev}, we provide a \fbl{specialization to the class of OrthoMatDot Codes}, a numerically stable matrix multiplication code construction \bll{that has the same recovery threshold and communication/computation cost per worker as MatDot codes}. The construction specializes \fbl{the class of OrthoMatDot Codes}  via the use of Chebyshev polynomials, which are a class of orthogonal polynomials that are ubiquitous in numerical methods and approximation theory \cite{trefethen2013approximation}. Construction \ref{con:chebpoly} also specifies the choice of evaluation points  $x_1, x_2, \ldots, x_P.$ 

The decoding procedure outlined \fbl{for the specialization of OrthoMatDot Codes in Section \ref{sec:Chebyshev}} involves the effective inversion of some $2m-1 \times 2m-1$ sub-matrix of a $2m-1 \times P$ \emph{Chebyshev-Vandermonde} matrix \cite{reichel1991chebyshev}, where each of the $i$-th  column contains evaluations of the first $2m-1$ Chebyshev polynomials at $x_i,i=1,2,\ldots,P$. A key technical result of our paper shows that, with our choice of evaluation points $x_1, x_2, \ldots, x_P,$ every $2m-1 \times 2m-1$ square sub-matrix of the $2m-1 \times P$ Chebyshev-Vandermonde matrix is well-conditioned. More precisely, we show that, with our choice of $x_1, x_2, \ldots, x_P$, the condition number of \emph{any} $2m-1 \times 2m-1$ sub-matrix of the Chebyshev-Vandermonde matrix grows at most polynomially in $P$ when the number of redundant parity nodes $\Delta=P-(2m-1)$ is fixed. Our condition number bound may be viewed as result of independent interest in the area of numerical methods, and requires non-trivial use of techniques from numerical approximation theory. This result is in contrast with the well known exponential growth for Vandermonde systems. We also show the significant improvement in  stability via numerical experiments in Section \ref{sec:numerical}. 
We also provide a preview of the results here in Table \ref{table:errors}, whose results demonstrate that remarkably, our Chebyhev-Vandermonde construction with even $P=150$ nodes has a smaller relative error than the Vandermonde-based MatDot Codes\footnote{We note that the numerical error depends not only on the condition number of the matrix, but also the algorithm used for solving the linear system. However, we are not aware of any approach that can accurately solve, say, a $150 \times 150$ linear system with a Vandermonde matrix (See e.g., \cite{demmel2005accurate, ramamoorthy2019universally})} with $P=30$ nodes.
 \begin{table*}[!htb]
	\centering
    	\caption{A table depicting the relative errors of various schemes for $\Delta=P-(2m-1)=3$ redundant nodes. The error is measured via the Frobenius norm, i.e., $\frac{||\mathbf{A}\mathbf{B}-\hat{\mathbf{C}}||_F}{||\mathbf{A}\mathbf{B}||_F}$. The matrices $\mathbf{A},\mathbf{B}$ are chosen with  entries $\mathcal{N}(0,1).$ The average relative error averages over all possible $3$ node failures, i.e., over every set of $2m-1$ nodes among the $P=2m+2$ nodes; the worst case relative error involves the worst set of $2m-1$ nodes. See Section \ref{sec:numerical} for more details.\newline}  
	\label{table:errors}
	\begin{tabular}{|c|c|c|c|c|}
		\hline
		Number   &{MatDot}  & {{\fbl{OrthoMatDot}}}& {{MatDot }}   &{\fbl{OrthoMatDot}}  \\
	of Workers	& worst case& worst case&average&average
		\\
	$(P)$	&relative error&relative error&relative error&relative error
		\\
		\hline
		30 & $1.54\times 10^{-6}$ & $5.14\times 10^{-11}$ & $1.36 \times 10^{-7}$ & $1.36 \times 10^{-13}$\\
		50 & $8.6 \times 10^{3}$ & $1.27 \times 10^{-9}$ & $2.00 \times 10^2$ & $2.04 \times 10^{-13}$\\
		80 & $2.45 \times 10^{6}$ & $1.98 \times 10^{-8}$ &$ 2.19 \times 10^{2} $& $3.08 \times 10^{-12}$\\
		150 & $3.87 \times 10^{7}$ & $7.84 \times 10^{-7}$ & $8.73 \times 10^{2}$ & $2.03 \times 10^{-11}$\\
		\hline
	\end{tabular}
\end{table*}

While MatDot Codes \cite{allerton17} have an optimal recovery threshold of $2m-1$, they have relatively higher computation cost per worker ($O(N^3/m)$) and worker node to fusion node communication cost ($O(N^2)$) as compared to Polynomial Codes \cite{polynomialcodes} which have a computation cost per worker of $O(N^3/m^2)$ and worker node to fusion node communication cost of $O(N^2/m^2)$. In particular, each worker in MatDot Codes performs an ``outer'' product of an $N \times N/m$ matrix with a $N/m \times N$ matrix, whereas each worker in Polynomial Codes performs an ``inner'' product of a $N/m \times N$ matrix with a $N \times N/m$ matrix. The reduced computation/communication comes at the cost of weaker fault-tolerance - Polynomial Codes have a higher recovery threshold of $m^2$ as compared with MatDot Codes ($2m-1$).  In Section \ref{sec:chebpolycod}, we develop numerically stable codes for matrix multiplication, again via orthogonal polynomials, that achieve the same low computation/communication costs as Polynomial Codes as well as the same recovery threshold, \fbl{we refer to these codes as \emph{OrthoPoly Codes}}.

The trade-off between computation/communication cost and recovery threshold imposed by MatDot Codes and Polynomial Codes has motivated general code constructions that interpolates both of them \cite{allerton17, genPolyDot, entPoly}, albeit using the monomial basis.   In Section \ref{sec:polydot-chebyshev}, we extend our approach to a general matrix multiplication code construction, \fbl{referred to as \emph{Generalized OrthoMatDot},} that offers a computation/communication cost vs recovery threshold trade-off, following the research thread for the monomial basis \cite{allerton17, genPolyDot, entPoly}, however we also target  numerical stability in our proposed construction. While our \fbl{Generalized OrthoMatDot Codes specialize to OrthoMatDot Codes}, i.e.,  they achieve the same optimal recovery threshold as \fbl{OrthoMatDot} Codes when allowing for the same computation/communication cost as \fbl{OrthoMatDot} Codes, they do not specialize to \fbl{OrthoPoly} Codes. Specifically, \fbl{Generalized OrthoMatDot} codes have higher recovery threshold than  \fbl{OrthoPoly} Codes when allowing for the same computation/communication cost as OrthoPoly Codes.  In Section \ref{sec:lcc}, we exploit the result obtained in Theorem \ref{thm:bound} on the condition number of the square $K \times K$ sub-matrices of the $K\times P$ Chebyshev-Vandermonde matrices to propose a numerically stable algorithm for Lagrange coded computing.  In Section \ref{sec:conc}, we conclude with a discussion on other related problems such as matrix-vector multiplication \cite{Huang_TC_84,lee2016speeding}, and describe some related open questions.

\section{Preliminaries on Numerical Analysis and Notations}\label{sec:prem}

 {
We discuss, in this section, the problem of finite precision in representing real numbers on digital machines and how it may horribly affect the output of computation problems performed on these machines. In addition, we also introduce some basic definitions and results from the area of numerical approximation theory that will be used in this paper\cite{num97}, \cite{shortApp}. At the end of this section, we provide most of the common notations that will be used in this paper.
}
\subsection{Preliminaries on Numerical Analysis}
 {
Since digital machines have finite memory, real numbers are digitally stored using a finite number of bits, i.e., finite precision. However, storing  real numbers using a finite number of bits leads to  inevitable errors since a finite number of bits can only represent a finite number of real numbers with no errors. On the other hand, real numbers that cannot be directly represented using the specified finite number of bits have to be either truncated or rounded-off in order to fit in the memory. Although such perturbation  (e.g., truncation/round-off error) of real numbers due to the finite precision of digital machines can be negligibly small, the perturbation of the output of any computation that uses such ``small" perturbed stored real numbers as input does not necessarily be small as well. In fact, a very small perturbation to  the input of some computation may lead to an output that is totally wrong and irrelevant to the correct output. The condition number of a computation problem captures/measures this observation. 
\begin{definition}[Condition Number]
Let $f$ be a function representing a computation problem with input $x$, and let $\delta x$ be a small perturbation of $x$, and define $\delta f=f(x+\delta x)-f(x)$ to be the perturbation of $f$ due to $\delta x$, the condition number of the problem at $x$ with respect to some norm $||\cdot||$ is 
\begin{align}
\kappa(x)=\sup_{\delta x} \left(\frac{||\delta f||}{||f(x)||}\middle/\frac{||\delta x||}{||x||}\right).
\end{align}
\end{definition}
Given the above definition of condition number, a problem is said to be ``ill-conditioned" if small perturbations in the input lead to large perturbation in the output (i.e., the condition number is large). On the other hand, a problem is said to be ``well-conditioned" if small perturbations in the input lead to small perturbations in the output (i.e., the condition number is small).
}
 {
In what follows,  we discuss the condition number of two computation problems: the matrix-vector multiplication and solving a system of linear equations. For both problems, consider the system of linear equations represented in the matrix form $\mathbf{A}\mathbf{x}=\mathbf{y}$, where $\mathbf{A}\in \mathbb{R}^{n,n}$ and non-singular, and $\mathbf{x},\mathbf{y}\in\mathbb{R}^n$, and let $||\cdot||$ be some matrix norm. Then, let $\mathbf{A}$ be fixed, the condition number of this matrix-vector multiplication problem with $\mathbf{y}$ as its output given small perturbations in the input  $\mathbf{x}$ is $\kappa(\mathbf{x})\leq ||\mathbf{A}||||\mathbf{A}^{-1}||$, for any $\mathbf{x} \in \mathbb{R}^n$. Also, for the problem of solving the system of linear equations $\mathbf{A}\mathbf{x}=\mathbf{y}$, with $\mathbf{A}$ still fixed, the condition number of the problem of solving this system of linear equations, given small perturbations in the input $\mathbf{y}$, where $\mathbf{x}$ is the output, is  $\kappa(\mathbf{y}) \leq ||\mathbf{A}||||\mathbf{A}^{-1}||$, for any $\mathbf{y}  \in \mathbb{R}^n$.
}

Since we focus on polynomially coded computing, next, we introduce some basic tools of  numerical approximation theory  that will be used throughout this paper. Notice that, in the following, $C[a,b]$ denotes the vector space of continuous integrable functions defined on the interval $[a,b]$.  
\begin{definition}[Inner Products on \text{$C[a,b]$}]
For any $f,g \in C[a,b]$, and given a non-negative integrable weight function $w$,  
\begin{align}
\langle f,g \rangle=\int_{a}^{b} f(x)g(x)w(x) dx\notag
\end{align}
defines an inner product on $C[a,b]$ relative to $w$.
\end{definition}
\vspace{2mm}
\begin{definition}[Orthogonal Polynomials]
Consider a non-negative integrable weight function $w$, the polynomials $\{q_i\}_{i\geq 0}$ in $C[a,b]$ \vc{where $q_i(x)$ has degree $i$} and 
\begin{align}
    \langle q_i, q_j \rangle= \left\{ \begin{array}{ll}
       c_{i} & \text{if $i=j$,} \\
        0 &\text{otherwise,}
        \end{array}\right.
\end{align}
for some non-zero values $c_{i}$, where the inner product is relative to $w$, are called orthogonal polynomials relative to $w$, .
\end{definition}
\begin{definition}[Orthonormal Polynomials]
Consider a non-negative integrable weight function $w$, the  polynomials $\{q_i\}_{i\geq 0}$, \vc{where $q_i(x)$ has degree $i$}, in $C[a,b]$ such that 
\vspace{-0mm}
\begin{align}
    \langle q_i, q_j \rangle= \left\{ \begin{array}{ll}
       1 & \text{if $i=j$,} \\
        0 &\text{otherwise,}
        \end{array}\right.
\end{align}
where the inner product is relative to $w$, are called orthonormal polynomials relative to $w$.
\end{definition}

\vc{Note that based on the above definitions, if the polynomials $\{q_i\}_{i\geq 0}$ are orthogonal (or orthonormal), then $q_n(x)$ is orthogonal to all polynomials of degree $\leq n-1$, i.e., $\langle p_{n-1}(x), q_{n}(x)\rangle=0$, for any polynomial $p_{n-1}\in C[a,b]$ with degree strictly less than $n$. It's also worth noting that for $w(x)=1, a=-1,b=1$, the orthogonal polynomials are Legendre polynomials, which are derived via Gram-Schmidt procedure applied to $\{1,x,x^2, \ldots,\}$ sequentially. In addition, the following is an important class of orthogonal polynomials in our paper.}

\begin{example}[Chebyshev polynomials of the first kind]\label{ex:cheb}
The following recurrence relation  defines the Chebyshev polynomials of the first kind: 
$$T_n(x)=2xT_{n-1}(x)-T_{n-2}(x),$$
where, $T_0(x)=1, T_1(x)=x$. These Chebyshev polynomials are the corner stone of modern numerical approximation theory and practice with applications to numerical integration, and least-square approximations of continuous functions \cite{num97},\cite{shortApp}.
  $\frac{1}{\sqrt{2}} T_0, T_1, T_2, \cdots$ are orthonormal relative to the weight function $\frac{2}{\pi \sqrt{1-x^2}}$. In general, Chebyshev polynomials are defined over $x\in \mathbb{R}$. However, for $x\in [-1,1]$, $T_n(x) = \cos(n\arccos(x))$, for any   $n \in \mathbb{N}$.  For the rest of this paper, unless otherwise is stated, whenever  Chebyshev polynomials are used, they are restricted only to  the range $[-1,1]$.
\end{example}

\vc{We state, next, two results from \cite{shortApp} in Theorems  \ref{thm:integration} and \ref{thm:GQ}.}
\begin{theorem}
Let $w$ be a weight function on the range $[a,b]$, \bl{i.e., $w$ is a non-negative integrable function on $[a,b]$}, and let $x_1, \cdots, x_n$ be distinct real numbers such that $a < x_1 < \cdots < x_n < b$, there exist unique weights $a_1, \cdots, a_n$ such that 
$$\int_a^b f(x) w(x) dx =\sum_{i=1}^n a_i f(x_i),$$
for all polynomials $f$ with degree less than $n$.
\label{thm:integration}
\end{theorem}

\vc{Theorem \ref{thm:integration} is not surprising - the left hand side of the equation stated in the theorem is a linear operator on the vector space of $n-1$-degree polynomials. Because of Lagrange-interpolation, the space of $n-1$-degree polynomials is itself a linear transformation on its evaluation at $n$ points. Therefore, the left hand side can be expressed as an inner product of the functions evaluations at $n$ points. We next state a remarkable result by Gauss which states conditions under which the expression of Theorem \ref{thm:integration} is exact for polynomials of degree up to $2n-1,$ even though the number of evaluation points is just $n$.}

\begin{theorem}[Gauss Quadrature]\label{thm:GQ}
Fix a weight function $w$, and let $\{q_i\}_{i\geq 0}$ be a set of orthonormal polynomials in $C[a,b]$ relative to $w$. Given $n$, let $\eta_1, \cdots, \eta_n$ be  the roots of $q_n$ such that $a\leq \eta_1< \eta_2 < \cdots <\eta_n\leq b$, and choose  real values $a_1, \cdots, a_n$ such that $\sum_{i=1}^{n} a_i f(\eta_i) = \int_{a}^{b} f(x) w(x) dx$, for any $f \in C[a,b]$ with degree less than $n$. Then, $\sum_{i=1}^{n} a_i f(\eta_i) = \int_{a}^{b} f(x) w(x) dx$, for any polynomial $f$ with degree less than $2n$. 
\end{theorem}
\vspace{8mm}
\begin{remark}\label{rmk:orth}
\begin{enumerate}
\item Consider any orthonormal polynomials $\{q_i\}_{i>0}$. For any $n \in \mathbb{N}$, the set  $\{q_0, q_1, \cdots, q_{n-1}\}$ forms a basis for the vector space of polynomials with degree less than $n$.

\item In Theorem \ref{thm:GQ},
 $a_1, \cdots, a_n$  can be chosen as 
\begin{align}\label{eq:ais}
a_i= \int_{a}^{b} \bigg(\prod_{j\in[n]-i} \frac{x-\eta_j}{\eta_i-\eta_j}\bigg) w(x)dx ,~ i\in[n].
\end{align}

\item In Theorem \ref{thm:GQ}, the roots of $q_n$, i.e., $\eta_1, \cdots ,\eta_n$ are, in fact,  real and distinct. Moreover, the Chebyshev polynomial of the first kind $T_n$ has the following roots
\begin{align}\label{eq:Chebnodes}
 \rho^{(n)}_i =  \cos\left(\frac{2i-1}{2n} \pi\right),~ i\in [n].
\end{align}
The  set  $\{\rho^{(n)}_1, \cdots, \rho^{(n)}_n\}$ is often called the $n$-point Chebyshev grid, and its elements $\rho^{(n)}_1, \cdots, \rho^{(n)}_n$ are called ``Chebyshev nodes'' of degree $n$. We here discard the term ``node''  and use the term ``Chebyshev points''  to avoid confusion with computation nodes. We also denote by $\boldsymbol{\rho}^{(n)}$ the vector $(\rho_1^{(n)}, \cdots, \rho_n^{(n)})$. It is useful to note that $T_n(x)$ can be written as
\begin{align}
T_n(x)=2^{n-1}\prod_{i=1}^{n} (x-\rho_i^{(n)}),    
\end{align} and for $T_n(x)$, the weights $a_i$ in (\ref{eq:ais}) are all equal to ${2}/{n}$ when $w(x) =\frac{2}{\pi \sqrt{1-x^2}}.$
\end{enumerate}
\end{remark}
\subsection{Notations}\label{sec:notation}
Throughout this paper, we  use lowercase bold letters to denote vectors and uppercase bold letters to denote matrices. In addition,  for any positive integers $k,n$, and given a  set of orthogonal polynomials $q_0,q_1, \cdots, q_{k-1}$ on the interval $[a,b]$, let $\mathbf{x}=(x_1, \cdots, x_n)$ be a vector with entries in $[a,b]$, we define the $k\times n$ matrix $\mathbf{Q}^{(k,n)}(\mathbf{x})$ as:
\begin{align}
  \mathbf{Q}^{(k,n)}(\mathbf{x})= \left(\hspace{-2mm}\begin{array}{ccc} q_0(x_1)&  \cdots & q_{0}(x_n)\\\vdots & \ddots &\vdots \\ 
   q_{k-1}(x_1)&\cdots& q_{k-1}(x_n)\end{array}\hspace{-2mm}\right).
\end{align}
 For any subset $\mathcal{S}=\{s_1, \cdots, s_r\}\subset [n]$, we denote by $\mathbf{Q}^{(k,n)}_\mathcal{S}(\mathbf{x})$ the sub-matrix of $\mathbf{Q}^{(k,n)}(\mathbf{x})$ formed by concatenating columns with indices in $\mathcal{S}$, i.e.,    
\begin{align}
  \mathbf{Q}_\mathcal{S}^{(k,n)}(\mathbf{x})= \left(\hspace{-2mm}\begin{array}{ccc} q_0(x_{s_1})&  \cdots & q_{0}(x_{s_r})\\\vdots & \ddots &\vdots \\ 
   q_{k-1}(x_{s_1})&\cdots& q_{k-1}(x_{s_r})\end{array}\hspace{-2mm}\right).
\end{align}

For the special case where the orthogonal polynomials are the  Chebyshev polynomials of the first kind $T_0, T_1, \cdots, T_{k-1}$, we define the $k\times n$ matrix $\mathbf{G}^{(k,n)}(\mathbf{x})$ as:
\begin{align}\label{eq:GTT}
  \mathbf{G}^{(k,n)}(\mathbf{x})= \left(\hspace{-2mm}\begin{array}{ccc} T_0(x_1)&  \cdots & T_{0}(x_n) \\ 
  \vdots & \ddots &\vdots \\ 
   T_{k-1}(x_1)&\cdots& T_{k-1}(x_n)\end{array}\hspace{-2mm}\right),
\end{align}
we denote by $\mathbf{G}^{(k,n)}_\mathcal{S}(\mathbf{x})$    the sub-matrix of $\mathbf{G}^{(k,n)}(\mathbf{x})$ formed by concatenating columns with indices in $\mathcal{S}$, i.e.,    
\begin{align}
\mathbf{G}_\mathcal{S}^{(k,n)}(\mathbf{x})= \left(\hspace{-2mm}\begin{array}{ccc}
  T_0(x_{s_1})&  \cdots & T_{0}(x_{s_r})\\\vdots & \ddots &\vdots \\ 
   T_{k-1}(x_{s_1})&\cdots& T_{k-1}(x_{s_r})\end{array}\hspace{-2mm}\right).
\end{align}

Also, for the case where the orthogonal polynomials are the ``orthonormal'' Chebyshev polynomials $\frac{1}{\sqrt{2}}T_0, T_1, \cdots, T_{k-1}$, we define the $k\times n$ matrix $\tilde{\mathbf{G}}^{(k,n)}(\mathbf{x})$ as:
\begin{align}\label{eq:GT}
  \tilde{\mathbf{G}}^{(k,n)}(\mathbf{x})= \left(\hspace{-2mm}\begin{array}{ccc} T_0(x_1)/\sqrt{2}&  \cdots & T_{0}(x_n)/\sqrt{2} \\ 
   T_1(x_1)&  \cdots & T_{1}(x_n)\\\vdots & \ddots &\vdots \\ 
   T_{k-1}(x_1)&\cdots& T_{k-1}(x_n)\end{array}\hspace{-2mm}\right),
\end{align}
and we denote by  $\tilde{\mathbf{G}}^{(k,n)}_\mathcal{S}(\mathbf{x})$  the sub-matrix of  $\tilde{\mathbf{G}}^{(k,n)}(\mathbf{x})$ formed by concatenating columns with indices in $\mathcal{S}$, i.e.,    
\begin{align}
  \tilde{\mathbf{G}}_\mathcal{S}^{(k,n)}(\mathbf{x})= \left(\hspace{-2mm}\begin{array}{ccc} T_0(x_{s_1})/\sqrt{2}&  \cdots & T_{0}(x_{s_r})/\sqrt{2}\\
  T_1(x_{s_1})&  \cdots & T_{1}(x_{s_r})\\\vdots & \ddots &\vdots \\ 
   T_{k-1}(x_{s_1})&\cdots& T_{k-1}(x_{s_r})\end{array}\hspace{-2mm}\right).
\end{align}
Wherever there is no ambiguity on  $\mathbf{x}$, it may be dropped from the notation.


In the next section, we show that orthonormal polynomials can be used for designing codes for the distributed  large scale matrix multiplication problem. 

\section{\fbl{OrthoMatDot:} Orthonormal Polynomials based Codes for Distributed Matrix Multiplication}
\label{sec:ortho}

In this section,  we present a new orthonormal polynomials based class of codes for  matrix-multiplication \fbl{called OrthoMatDot}. These codes  achieve the same  recovery threshold as MatDot Codes, and have similar computational complexity as MatDot. The main advantage of the proposed codes is that they avoid dealing with the ill-conditioned monomial basis used in previous work (e.g., in \cite{allerton17,polynomialcodes,genPolyDot,entPoly}). In Section \ref{sec:Chebyshev}, \fbl{OrthoMatDot} Codes will be specialized and demonstrated to have higher numerical stability as compared with state of the art.  We begin with a formal problem formulation in Section \ref{sec:sysprob}, and describe our codes in Section \ref{sec:orthosub}. 

\subsection{System Model and Problem Formulation}\label{sec:sysprob}
\subsubsection{System Model}\label{sec:sysmod}
\begin{figure}[t]
    \centering
    \includegraphics[scale=0.450]{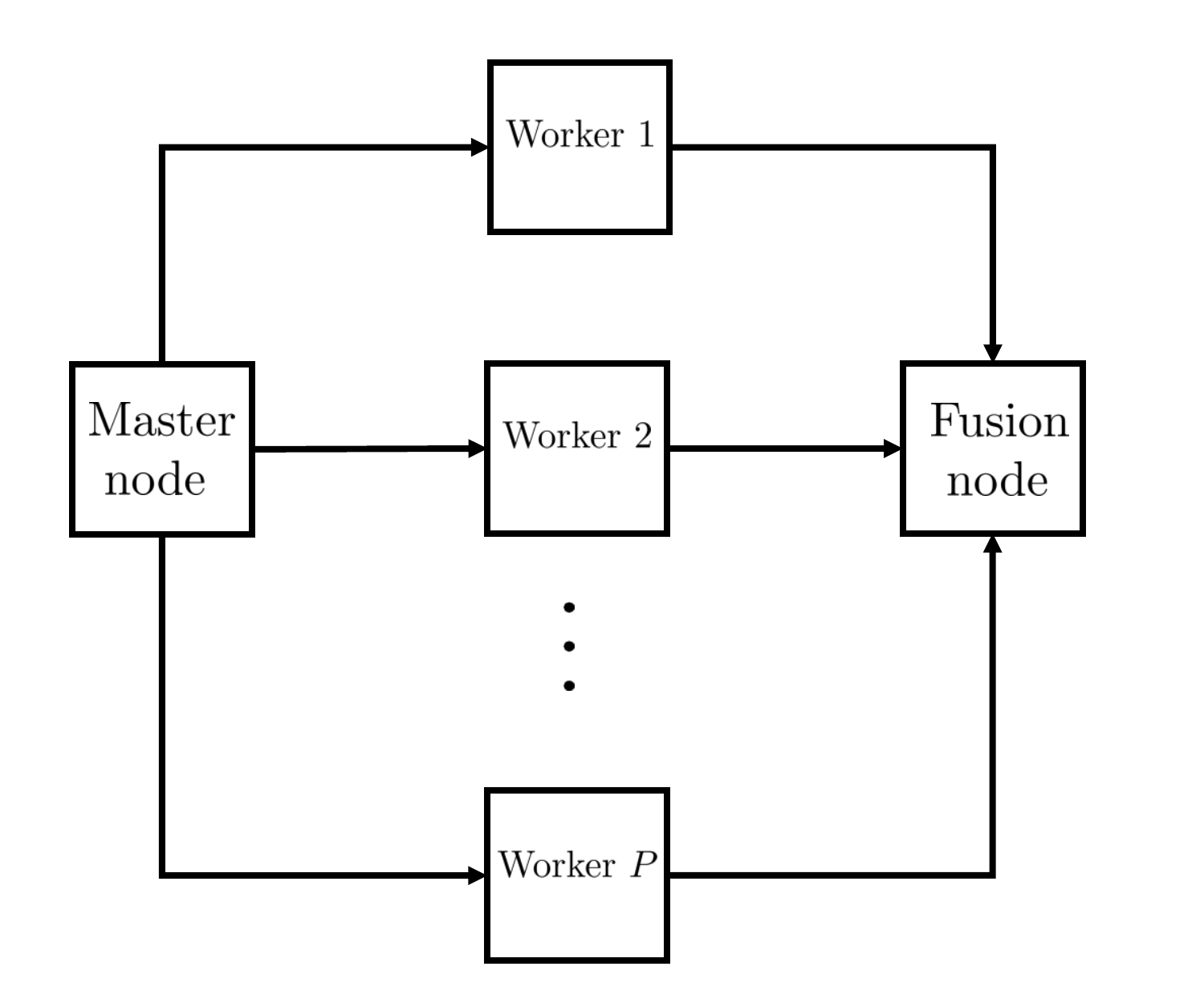}
    \caption{The distributed system framework}
    \label{fig:Model}
\end{figure}
We consider the distributed framework depicted in Fig. \ref{fig:Model} that consists of a master node, $P$ worker nodes, and a fusion node where the only communication allowed  is from the master node to the different worker nodes and from the worker nodes to the fusion node. It can happen that the fusion node and the master node be represented by the same node. In this case, the only communication allowed is the communication between the master node and every worker node. 

\subsubsection{Problem Formulation}\label{sec:probfor}
The master node possesses two real-valued input matrices $\mathbf{A}$, $\mathbf{B}$ with dimensions $N_1\times N_2$, $N_2\times N_3$, respectively. Every worker node receives from the master node an encoded matrix of $\mathbf{A}$ of dimension $N_1 \times N_2/m$ and  an encoded matrix of $\mathbf{B}$ of dimension $N_2/m \times N_3$, and performs matrix multiplication of these two received inputs. Upon performing the matrix multiplication, each worker node sends the result to the fusion node. 
The fusion node needs to recover the matrix multiplication $\mathbf{A}\mathbf{B}$ once it receives the results of any $K$ worker nodes, where $K \leq P$. In this case, $K$ is denoted by \emph{the recovery threshold} of the distributed computing scheme. 

\subsection{\fbl{OrthoMatDot} Code Construction}\label{sec:orthosub}
Our result regarding the existence of achievable codes solving the distributed matrix multiplication problem using orthonormal polynomials is stated in the following theorem. 
\vspace{0.0mm}

\begin{theorem}
\label{thm:orthpoly}
For the matrix multiplication problem described in Section \ref{sec:probfor} computed on the system defined in Section \ref{sec:sysmod},  a recovery threshold of 
$2m-1$ is achievable using any set of orthonormal polynomials $\{q_i\}_{i\geq 0}$  relative to some weight polynomial $w$ and defined on a range $[a,b]$.
\end{theorem}
\vspace{2.0mm}

Before proving this theorem, we first present \fbl{OrthoMatDot}, a code construction that achieves the recovery threshold of $2m-1$ given any set $\{q_i\}_{i\geq 0}$ of orthonormal polynomials relative to a weight polynomial $w(x)$ and defined on a range $[a,b]$. In our code construction, we assume that matrix $\mathbf{A}$  is split vertically into $m$ equal sub-matrices, of dimension $N_1\times N_2/m$ each, and matrix $\mathbf{B}$ is split horizontally into $m$ equal sub-matrices, of dimension $N_2/m \times N_3 $ each, as follows:
\begin{equation}\label{eq:spltA}
\mathbf{A} = \left(\mathbf{A}_0 \ \mathbf{A}_1 \ \ldots \ \mathbf{A}_{m-1}\right),\;\;\; \mathbf{B}=\left(\begin{array}c \mathbf{B}_0\\\mathbf{B}_1\\\vdots \\\mathbf{B}_{m-1}\end{array}\right),
\end{equation}
we also define a set of $P$ distinct real numbers $x_1, \cdots, x_P$ in the range $[a,b]$, and define two encoding polynomials $p_\mathbf{A}(x)=\sum_{i=0}^{m-1} \mathbf{A}_i q_{i}(x)$ and $p_\mathbf{B}(x)=\sum_{i=0}^{m-1} \mathbf{B}_i q_{i}(x),$ and let $p_\mathbf{C}(x)=p_{\mathbf{A}}(x)p_\mathbf{B}(x)$. 

In the following, we briefly describe the \fbl{OrthoMatDot}  construction. First, for every $r \in[P]$, the master node sends to the $r$-th worker node evaluations of $p_\mathbf{A}(x),p_\mathbf{B}(x)$ at $x=x_r$, that is, it sends $p_\mathbf{A}(x_r)$ and $p_\mathbf{B}(x_r)$ to the $r$-th worker node. Next,  for every $r \in [P]$, the $r$-th worker node computes the matrix product $p_{\mathbf{C}}(x_r)=p_\mathbf{A}(x_r) p_\mathbf{B}(x_r)$ and sends the result to the fusion node. Once the fusion node receives the output of any $2m-1$ worker nodes, it interpolates the polynomial $p_{\mathbf{C}}(x)=p_{\mathbf{A}}(x)p_{\mathbf{B}}(x)$, and evaluates  $p_{\mathbf{C}}(x)$ at $\eta_1, \cdots, \eta_{m}$, where $\eta_1, \cdots, \eta_m$ are the roots of $q_m$. Then, it performs the summation $\sum_{r=1}^{m} a_r\hspace{0.5mm} p_{\mathbf{C}}(\eta_r)$,  where $a_1, \cdots, a_m$ are as in (\ref{eq:ais}). 

We formally present \fbl{OrthoMatDot code} in Construction \ref{con:orthpoly}. Construction \ref{con:orthpoly} uses the following notation. The output of the algorithm is the  $N_1\times N_3$ matrix $\hat{\mathbf{C}}.$ The $(i,j)$-th entries of  the matrix polynomial $p_{\mathbf{C}}(x)$ and the matrix $\hat{\mathbf{C}}$ are respectively denoted as  $p^{(i,j)}_{\mathbf{C}}(x)$ and $\hat{C}(i,j).$
The reader may also recall the definition of matrices $\mathbf{Q}^{(2m-1,P)}(\mathbf{x})$ and $\mathbf{Q}_{\mathcal{R}}^{(2m-1,P)}(\mathbf{x}),$ for any subset $\mathcal{R}=\{{r_1}, \cdots, {r_{2m-1}}\}  \subset [P]$. $\boldsymbol{\eta}=(\eta_1, \cdots, \eta_m)$ is the vector of the roots of $q_m$. Based on  Construction \ref{con:orthpoly}, we state the following claim.

\begin{algorithm}
\caption{\fbl{OrthoMatDot}: \textbf{Inputs}: $\mathbf{A}, \mathbf{B}$,~~\textbf{Output}: $\hat{\mathbf{C}}$} \label{con:orthpoly}
\begin{algorithmic}[1]
\Procedure{MasterNode}{$\mathbf{A},\mathbf{B}$}\Comment{The master node's procedure}
\State $r \gets 1$
\While{$r\not=P+1$}
\State $p_\mathbf{A}(x_r) \gets \sum_{i=0}^{m-1} \mathbf{A}_i q_{i}(x_r)$ 
\State $p_\mathbf{B}(x_r) \gets \sum_{i=0}^{m-1} \mathbf{B}_i q_{i}(x_r)$
\State \textbf{send} $p_{\mathbf{A}}(x_r), p_{\mathbf{B}}(x_r)$ \textbf{to worker node} $r$
\State $r \gets r+1$
\EndWhile
\EndProcedure
\State
\Procedure{WorkerNode}{$p_{\mathbf{A}}(x_r), p_{\mathbf{B}}(x_r)$}\Comment{The procedure of worker node $r$}
\State $p_\mathbf{C}(x_r) \gets p_{\mathbf{A}}(x_r) p_{\mathbf{B}}(x_r)$
\State \textbf{send} $p_{\mathbf{C}}(x_r)$ \textbf{to the fusion node}
\EndProcedure
\State
\Procedure{FusionNode}{$\{p_{\mathbf{C}}(x_{r_1}), \cdots, p_{\mathbf{C}}(x_{r_{2m-1}})\}$}\Comment{The fusion node's procedure, $r_i$'s are distinct}
\State  $\mathbf{Q}_{\operatorname{inv}}\gets\left(\mathbf{Q}_{\mathcal{R}}^{(2m-1,P)}\right)^{-1}$
\For{$i\in [N_1]$}
\For{$j\in [N_3]$}
\State $(c^{(i,j)}_0, \cdots, c^{(i,j)}_{2m-2}) \gets (p_{\mathbf{C}}^{(i,j)}(x_{r_1}), \cdots,p_{\mathbf{C}}^{(i,j)}(x_{r_{2m-1}}))\mathbf{Q}_{\operatorname{inv}}$
\State $(p_\mathbf{C}^{(i,j)}(\eta_1), \cdots, p_{\mathbf{C}}^{(i,j)}(\eta_m))\gets (c^{(i,j)}_0, \cdots, c^{(i,j)}_{2m-2}) \mathbf{Q}^{(2m-1,m)}(\boldsymbol{\eta})$ 
\State $\hat{C}(i,j)\gets (p_\mathbf{C}^{(i,j)}(\eta_1), \cdots, p_{\mathbf{C}}^{(i,j)}(\eta_m)) (a_1, \cdots, a_m)^T $\Comment{$a_i$'s are as defined in (\ref{eq:ais})}
\EndFor
\EndFor
\State \textbf{return} $\hat{\mathbf{C}}$
\EndProcedure
\end{algorithmic}
\end{algorithm}
\begin{claim}\label{cl:orthAB}
 $\mathbf{A}\mathbf{B}=\sum_{r=1}^{m} a_r\hspace{0.5mm} p_{\mathbf{C}}(\eta_r).$
\end{claim}
The proof of Claim \ref{cl:orthAB} is provided in Appendix \ref{app:ortho}. 

Now, we can prove Theorem \ref{thm:orthpoly}.
\begin{proof}[Proof of Theorem \ref{thm:orthpoly}:]
In order to prove the theorem, it suffices to show that Construction \ref{con:orthpoly} is a valid construction with a recovery threshold of $2m-1$. Therefore, in the following, we prove that Construction \ref{con:orthpoly} can recover $\mathbf{A}\mathbf{B}$ after the fusion node receives the output of at most $2m-1$ worker nodes. 
Assume that the fusion node has already received the results of any  $2m-1$ worker nodes. Now, because the polynomial $p_{\mathbf{C}}(x)$ has degree $2m-2$,  the evaluations of $p_{\mathbf{C}}(x)$ at any $2m-1$ distinct points is sufficient to interpolate the polynomial, and since $x_1, \cdots, x_P$ are distinct, the fusion node can interpolate  $p_{\mathbf{C}}(x)$ once it receives  the output of any $2m-1$ worker nodes. Afterwards, given that $\mathbf{A}\mathbf{B}=\sum_{r=1}^{m} a_r\hspace{0.5mm} p_{\mathbf{C}}(\eta_r)$ (Claim \ref{cl:orthAB}), the fusion node can evaluate $p_{\mathbf{C}}(\eta_1), \cdots, p_{\mathbf{C}}(\eta_m)$ and perform the scaled summation $\sum_{r=1}^{m} a_r\hspace{0.5mm} p_{\mathbf{C}}(\eta_r)$  to recover $\mathbf{A}\mathbf{B}$.
\end{proof}
\vspace{02mm}
\begin{remark}
In Construction \ref{con:orthpoly}, setting $x_1, \cdots, x_m$ to be the roots of $q_m$ leads to a faster decoding for the scenarios in which the first $m$ worker nodes send their results but only  less than $2m-1$ workers succeed to send their outputs. For such scenarios, we have $\sum_{r=1}^{m} a_r\hspace{0.5mm} p_{\mathbf{C}}(x_r)=$ $\sum_{r=1}^{m} a_r\hspace{0.5mm} p_{\mathbf{C}}(\eta_r)=\mathbf{A}\mathbf{B}$, where the last equality follows from Claim \ref{cl:orthAB}.
\end{remark}

Next, we study the computational and communication costs of \fbl{OrthoMatDot}. 
 
\subsubsection{Complexity Analyses of \fbl{OrthoMatDot}}\label{sec:con1complexity}
\hspace{2mm}

\textbf{Encoding Complexity}: 
Encoding for each worker requires performing two additions, each adding $m$ scaled matrices of size $N_1 N_2/m$ and $N_2 N_3 /m$,  for an overall encoding complexity for each worker of $O( N_1 N_2+N_2N_3)$. Therefore, the overall computational complexity of encoding for $P$ workers is $O(N_1N_2P+N_2N_3P)$.

\textbf{Computational Cost per Worker}: Each worker multiplies two matrices of  dimensions $N_1\times N_2/m$ and $N_2/m\times N_3$, requiring $O(N_1N_2N_3/m)$ operations.

\textbf{Decoding Complexity}:
Since $p_\mathbf{C}(x)$ has degree $2m-2$, the interpolation of $p_\mathbf{C}(x)$ requires the inversion of a $2m-1 \times 2m-1$ matrix, with complexity $O(m^3)$, and performing $N_1N_3$ matrix-vector multiplications, each of them is between the inverted  matrix and a  column vector of length $2m-1$ of the received evaluations of the matrix polynomial $p_\mathbf{C}(x)$ at some position $(i,j) \in [N_1] \times [N_3]$, with complexity $O(N_1N_3m^2)$. Next, the evaluation of the polynomial $p_\mathbf{C}(x)$ at $\eta_1, \cdots, \eta_m$ requires a complexity of $O(N_1N_3m^2)$. Finally,  performing the summation $\sum_{r=1}^m a_r p_\mathbf{C}(\eta_r)$ requires a complexity of $O(N_1N_3m)$. Thus, assuming that $m \ll N_1,N_3$, the overall decoding complexity is $O(m^3+2N_1N_3m^2+N_1N_3m)=O(N_1N_3m^2)$.

\textbf{Communication Cost}:
The master node sends $O( N_1N_2P/m+N_2N_3P/m)$ symbols, and the fusion node receives $O(N_1 N_3 m)$ symbols from the successful worker nodes. 

\vspace{3mm}
 
\begin{remark}
With the reasonable assumption that the dimensions of the input matrices $\mathbf{A},\mathbf{B}$ are large enough such that $N_1,N_2,N_3 \gg m,P$, we can conclude that the encoding and decoding costs at the master and fusion nodes, respectively, are negligible compared to the computation cost at each worker node.
\end{remark}

\section{Numerically Stable Codes for Matrix Multiplication via \fbl{OrthoMatDot Codes with Chebyshev Polynomials}}
\label{sec:Chebyshev}

In this section, we specialize \fbl{OrthoMatDot} Codes by restricting the orthonormal polynomials to be Chebyshev polynomials of the first kind $\{T_i\}_{i\geq 0}$ with the evaluation points chosen to be the $P$-dimensional Chebyshev grid, i.e., $x_i=\rho_i^{(P)}, i\in [P]$.  Our \fbl{specialized  OrthoMatDot}, described in Construction \ref{con:chebpoly} in Section \ref{subsec:Chebyshev}, develops a decoding that involves inversion of a $2m-1 \times 2m-1$ sub-matrix of a $2m-1 \times P$ Chebyshev-Vandermonde matrix. One of the main technical results of this section (and paper), presented in Theorem \ref{thm:bound} in Section \ref{subsec:condbound},  is an upper bound to the worst case condition number over all possible $2m-1 \times 2m-1$ sub-matrices of the $2m-1 \times P$ Chebeshev-Vandermonde matrix for the case where the distinct evaluation points $x_1, \cdots, x_P$ are chosen as the Chebyshev points of degree $P$, i.e., $x_i=\rho_i^{(P)}, i\in [P]$. In fact, the derived bound shows that the worst case condition number grows at most polynomially in $P$ at a fixed number of straggler/parity worker nodes. This is in contrast with the monomial basis codes where the  condition number grows exponentially in $P$, even when there is no redundancy \cite{gautschi1987lower, gautschi1990stable, gautschi1974norm, reichel1991chebyshev}. We show through numerical experiments in Section \ref{sec:numerical} that our proposed codes provide significantly lower numerical errors as  compared to  MatDot Codes in \cite{allerton17}.

\subsection{Chebyshev Polynomials based \fbl{OrthoMatDot} Code Construction}
\label{subsec:Chebyshev}
Recalling from Example \ref{ex:cheb} that  $\frac{1}{\sqrt{2}} T_0, T_1, T_2, \cdots$ form an orthonormal polynomial set relative to the weight function $w(x)= \frac{2}{\pi \sqrt{1-x^2}}$, in Construction \ref{con:chebpoly}, we  explain the application of Chebyshev polynomials of the first kind to Construction \ref{con:orthpoly}. Note that, in Construction \ref{con:chebpoly}, we assume that the input matrices $\mathbf{A}$ and $\mathbf{B}$ are also split as in (\ref{eq:spltA}), and let $x_1, x_2, \ldots, x_P$ be  distinct real numbers  in the range $[-1,1]$, and define the encoding functions $p_\mathbf{A}(x), p_{\mathbf{B}}(x)$ as $p_\mathbf{A}(x)=\frac{1}{\sqrt{2}}\mathbf{A}_0 T_0(x)+\sum_{i=1}^{m-1} \mathbf{A}_i T_{i}(x)$ and $p_\mathbf{B}(x)=\frac{1}{\sqrt{2}}\mathbf{B}_0 T_0(x)+\sum_{i=1}^{m-1} \mathbf{B}_i T_{i}(x),$ and let $p_\mathbf{C}(x)=p_\mathbf{A}(x)p_\mathbf{B}(x)$.

The idea of our Chebyshev polynomials based \fbl{OrthoMatDot} code is as follows: First, for every $r \in [P]$, the master node sends to the $r$-th worker node  $p_\mathbf{A}(\rho^{(P)}_r)$ and $p_\mathbf{B}(\rho^{(P)}_r)$. Next, for every $r \in [P]$, the $r$-th worker node computes the matrix product  $p_\mathbf{C}(\rho^{(P)}_r)=p_\mathbf{A}(\rho^{(P)}_r)p_\mathbf{B}(\rho^{(P)}_r)$ and sends the result to the fusion node. Once the fusion node receives the output of any $2m-1$ worker nodes, it  interpolates $p_{\mathbf{C}}(x)$. Then,  it evaluates  $p_{\mathbf{C}}(x)$ at $\rho^{(m)}_1, \cdots, \rho^{(m)}_m,$ \vc{ where $\rho^{(m)}_i$'s are as defined in (\ref{eq:Chebnodes}),} and computes $\sum_{i=1}^{m} a_i\hspace{0.5mm} p_{\mathbf{C}}(\rho^{(m)}_i)$, \vc{where $a_i = {2}/{m}, i\in[m]$ based on 3) in Remark \ref{rmk:orth}.} 
 
A formal description of our Chebyshev polynomials based \fbl{OrthoMatDot} code is provided in Construction \ref{con:chebpoly}. \bl{Construction \ref{con:chebpoly} uses the following notation. We let the $(i,j)$-th entry of the matrix polynomial $p_\mathbf{C}(x)$ be denoted $p^{(i,j)}_{\mathbf{C}}(x)$ and written as $p^{(i,j)}_{\mathbf{C}}(x)=\frac{1}{\sqrt{2}} c_0^{(i,j)} T_{0}(x)+\sum_{l=1}^{2m-2} c_l^{(i,j)} T_{l}(x)$. Also, following the notation in Section \ref{sec:notation}, we define the   {Chebyshev-Vandermonde} matrices  $\tilde{\mathbf{G}}^{(2m-1,P)}(\boldsymbol{\rho}^{(P)}),$ 
and $\tilde{\mathbf{G}}_{\mathcal{R}}^{(2m-1,P)}(\boldsymbol{\rho}^{(P)})$, for any subset $\mathcal{R}=\{{r_1}, \cdots, {r_{2m-1}}\}  \subset [P]$, we also define the matrix $\tilde{\mathbf{G}}^{(2m-1,m)}(\boldsymbol{\rho}^{(m)})$.} Finally, we assume that our construction returns an $N_1\times N_3$ matrix $\hat{\mathbf{C}}$ representing the result of the product $\mathbf{A}\mathbf{B}$, where the $(i,j)$-th entry of $\hat{\mathbf{C}}$ is $\hat{C}(i,j)$.

\begin{algorithm}
\caption{Chebyshev Polynomials based \fbl{OrthoMatDot}: \textbf{Inputs}: $\mathbf{A}, \mathbf{B}$,~~\textbf{Output}: $\hat{\mathbf{C}}$} \label{con:chebpoly}
\begin{algorithmic}[1]
\Procedure{MasterNode}{$\mathbf{A},\mathbf{B}$}\Comment{The master node's procedure}
\State $r \gets 1$
\While{$r\not=P+1$}
\State $p_\mathbf{A}(\rho^{(P)}_r) \gets \frac{1}{\sqrt{2}}\mathbf{A}_0+ \sum_{i=1}^{m-1} \mathbf{A}_i T_{i}(\rho^{(P)}_r)$ 
\State $p_\mathbf{B}(\rho^{(P)}_r) \gets \frac{1}{\sqrt{2}}\mathbf{B}_0+\sum_{i=1}^{m-1} \mathbf{B}_i T_{i}(\rho^{(P)}_r)$
\State \textbf{send} $p_{\mathbf{A}}(\rho^{(P)}_r), p_{\mathbf{B}}(\rho^{(P)}_r)$ \textbf{to worker node} $r$
\State $r \gets r+1$
\EndWhile
\EndProcedure
\State
\Procedure{WorkerNode}{$p_{\mathbf{A}}(\rho^{(P)}_r), p_{\mathbf{B}}(\rho^{(P)}_r)$}\Comment{The procedure of worker node $r$}
\State $p_\mathbf{C}(\rho^{(P)}_r) \gets p_{\mathbf{A}}(\rho^{(P)}_r) p_{\mathbf{B}}(\rho^{(P)}_r)$
\State \textbf{send} $p_{\mathbf{C}}(\rho^{(P)}_r)$ \textbf{to the fusion node}
\EndProcedure
\State
\Procedure{FusionNode}{$\{p_{\mathbf{C}}(\rho^{(P)}_{r_1}), \cdots, p_{\mathbf{C}}(\rho^{(P)}_{r_{2m-1}})\}$}\Comment{The fusion node's procedure, $r_i$'s are distinct}
\State  $\mathbf{G}_{\operatorname{inv}}\gets\left(\tilde{\mathbf{G}}_{\mathcal{R}}^{(2m-1,P)}\right)^{-1}$
\For{$i\in [N_1]$}
\For{$j\in [N_3]$}
\State $(c^{(i,j)}_0, \cdots, c^{(i,j)}_{2m-2}) \gets (p_{\mathbf{C}}^{(i,j)}(\rho^{(P)}_{r_1}), \cdots,p_{\mathbf{C}}^{(i,j)}(\rho^{(P)}_{r_{2m-1}}))\mathbf{G}_{\operatorname{inv}}$
\State $(p_\mathbf{C}^{(i,j)}(\rho_1^{(m)}),\hspace{-1mm}\cdots\hspace{-1mm}, p_{\mathbf{C}}^{(i,j)}(\rho_m^{(m)}))\gets (c^{(i,j)}_0, \cdots, c^{(i,j)}_{2m-2}) \tilde{\mathbf{G}}^{(2m-1,m)}(\boldsymbol{\rho}^{(m)})$ 
\State $\hat{C}(i,j)\gets \frac{2}{m} (p_\mathbf{C}^{(i,j)}(\rho^{(m)}_1), \cdots, p_{\mathbf{C}}^{(i,j)}(\rho^{(m)}_m)) (1, \cdots, 1)^T $\Comment{$a_i$'s are all ${2}/{m}$}
\EndFor
\EndFor
\State \textbf{return} $\hat{\mathbf{C}}$
\EndProcedure
\end{algorithmic}
\end{algorithm}

\begin{figure*}[!t]
    \centering
    \includegraphics[scale=0.50]{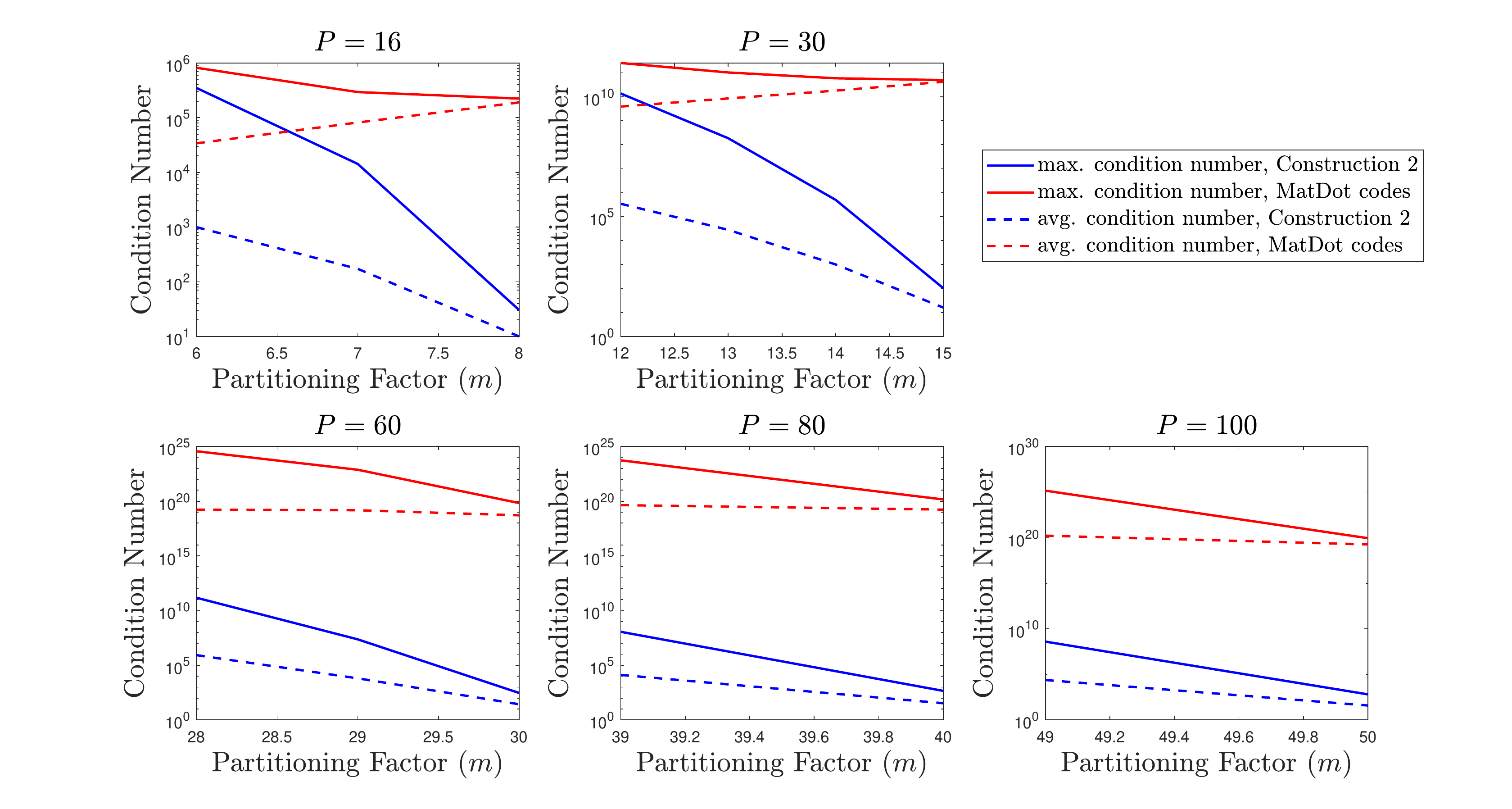}
    \caption{Comparison between the condition number of the interpolating matrix of the Chebyshev polynomials based \fbl{OrthoMatDot} Codes and MatDot Codes in five different distributed systems with $16,30,60,80$, and $100$ worker nodes, respectively.}
    \label{fig:Pall}
  \rule{\textwidth}{0.7pt}
\end{figure*}

\subsubsection{Complexity Analyses:}
The different encoding complexity, computational complexity per worker, decoding complexity and communication cost for \fbl{Chebyshev polynomials based OrthoMatDot are the same as their counterparts of OrthoMatDot} stated in Section \ref{sec:con1complexity}.

\subsection{Evaluation Points and Condition Number Bound}
\label{subsec:condbound}
\vc{When there is no redundancy, i.e., $n=2m-1,$ it is well known that the $n \times n$ decoding matrix $\mathbf{G}^{(n,n)}$ has condition number $n$ with the $\ell_2$ as well as the Frobenius norms \cite{gautschi1990stable}. Note the remarkable contrast with the Vandermonde matrix, whose condition number for real-valued evaluation points grows exponentially in $n$, no matter how the nodes are chosen \cite{gautschi1987lower, gautschi1990stable}. Our problem differs from the standard problem in numerical methods, since we have to choose a rectangular ``generator'' matrix where every square sub-matrix is well-conditioned. In particular,  even for Chebyshev-Vandermonde matrix, if the evaluation points are not chosen carefully, they are poorly conditioned \cite{reichel1991chebyshev} (also see Fig. \ref{fig:fixpar}). Here, we show that choosing $x_i=\rho_i^{(n)}$ leads to a well-conditioned system with $s$ redundant nodes.}
\vc{Our goal is to choose vector $\mathbf{x}$ such that $\kappa^{max}(\mathbf{G}^{(n-s,n)}(\mathbf{x}))$ is sufficiently small, where $\kappa^{max}(\mathbf{G}^{(n-s,n)}(\mathbf{x}))$ denotes the worst case condition number over all possible $n-s \times n-s$ sub-matrices of $\mathbf{G}^{(n-s,n)}(\mathbf{x})$.}

 {\begin{theorem}\label{thm:bound}
For any $s \in [n-1]$, $$\kappa^{max}_F(\mathbf{G}^{(n-s,n)}(\boldsymbol{\rho}^{(n)})) = O\left( (n-s)\sqrt{ns(n-s)}\left({2}n^2\right)^{s-1}\right),$$
where $\kappa^{max}_F$ denotes the worst case condition number over all possible $n-s \times n-s$ sub-matrices of $\mathbf{G}^{(n-s,n)}(\mathbf{x})$ with respect to the Frobenius norm,  $\boldsymbol{\rho}^{(n)}=(\rho_1^{(n)},\rho_2^{(n)},\ldots,\rho_n^{(n)})$ are the roots of the Chebyshev polynomial $T_n$, i.e., $\rho^{(n)}_i =  \cos\left(\frac{2i-1}{2n} \pi\right), i\in [n]$.
\end{theorem}}
Since $||.||_2 \leq ||.||_F,$ the above bound applies to the standard $\ell_2$ matrix norm as well. The proof uses techniques from numerical methods, and is provided in Appendix \ref{app:proofgen}.
\begin{remark}
Although the bound in Theorem \ref{thm:bound} is derived for $\mathbf{G}^{(n-s,n)}(\boldsymbol{\rho}^{(n)})$, the theorem also applies for $\tilde{\mathbf{G}}^{(n-s,n)}(\boldsymbol{\rho}^{(n)})$. This is because it can be shown using simple matrix operations that for any $\tilde{\mathbf{G}}^{(n-s,n)}_{\mathcal{R}}$, for a subset $\mathcal{R}\subset [n]$ such that $|\mathcal{R}|=n-s$,
$\kappa_F(\tilde{\mathbf{G}}^{(n-s,n)}_{\mathcal{R}})<\sqrt{2}\hspace{1mm}\kappa_F(\mathbf{G}^{(n-s,n)}_{\mathcal{R}}).$
\end{remark} 
\subsection{Numerical Results}
\label{sec:numerical}
\begin{figure*}[!t]
    \centering
    \includegraphics[scale=0.50]{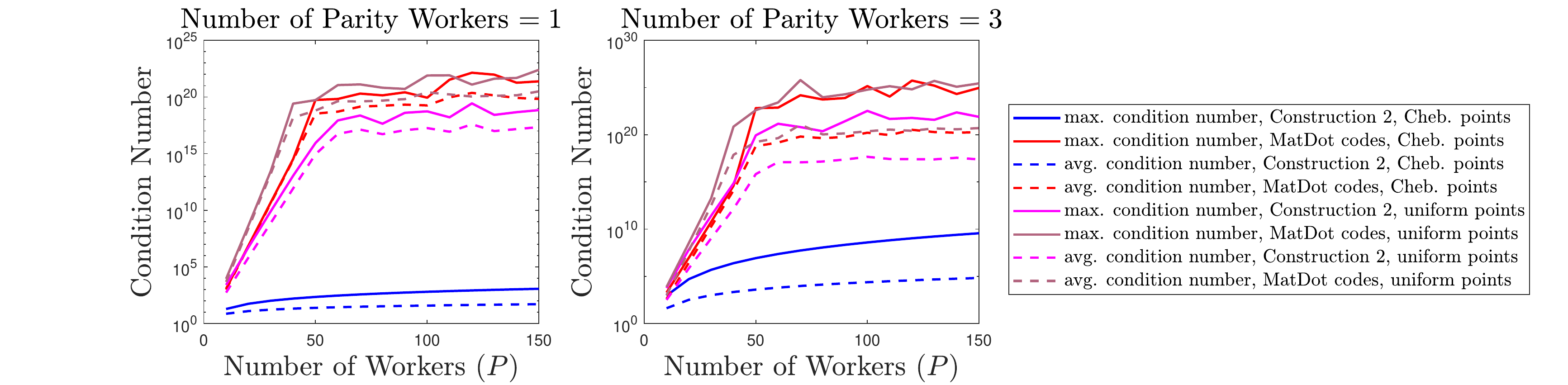}
    \caption{The growth of the condition number, for both Chebyshev polynomials based \fbl{OrthoMatDot} and MatDot Codes, with the system size given a fixed number of redundant worker nodes.}
    \label{fig:fixpar}
  \rule{\textwidth}{0.55pt}
\end{figure*}

\vc{The numerical stability of our codes is determined by the condition number of $2m-1 \times 2m-1$ sub-matrices of $\mathbf{G}^{(2m-1,P)}.$ The natural comparison is with MatDot Codes where the decoding depends on effectively inverting $2m-1\times 2m-1$ square sub-matrices of}
\begin{align}
  \mathbf{M}= \left(\hspace{-2mm}\begin{array}{cccc} 1& 1 & \cdots & 1\\ x_1 & x_2 & \cdots & x_P\\\vdots & \vdots &\ddots &\vdots \\ 
  x_1^{2m-2}& x_2^{2m-2}&\cdots& x_P^{2m-2}\end{array}\hspace{-2mm}\right).
\end{align}

\vc{Based on the result of Theorem \ref{thm:bound}, we choose $x_i=\rho_i^{(P)}.$ In our experiments, we consider systems with various number of worker nodes, namely, $P=16,30,60,80,100$. We compare $\kappa^{max}_2(\mathbf{G}^{(2m-1,P)})$ with $\kappa^{max}_2(\mathbf{M})$. We also compare the average $\ell_2$ condition number of all $2m-1\times 2m-1$ sub-matrices of $\mathbf{G}^{(2m-1,P)}$ and all $2m-1\times 2m-1$ sub-matrices of $\mathbf{M}$. The results, in Fig. \ref{fig:Pall}, show that, for every examined system, the maximum and average condition numbers of the  $2m-1\times 2m-1$ sub-matrices of $\mathbf{G}^{(2m-1,P)}$ are less than its MatDot Codes counterparts, especially for larger systems with $60, 80,$ and $100$ worker nodes. In fact, for these specific systems, the improvement in the condition number is around a scaling of $10^{15}$.}   

\vc{Fig. \ref{fig:fixpar} shows how the maximum/average condition number of the $2m-1\times2m-1$ sub-matrices of $\mathbf{G}^{(2m-1,P)}$ grows with the size of the distributed system given a fixed number of redundant worker nodes, namely 1 and 3, and compares with MatDot Codes.  The figure shows that while MatDot Codes provide a reasonable condition number $(\sim 10^{10})$ to distributed systems with size up to only $25$ worker nodes, Construction 2 can afford distributed systems with size up to $150$ worker nodes for the same condition number bound $\sim 10^{10}$.} 

As a reflection to the significant higher stability of \fbl{Chebyshev polynomials based OrthoMatDot} compared to MatDot Codes, Fig. \ref{fig:fixparerr} shows that \fbl{Chebyshev polynomials based OrthoMatDot}  provides much more accurate outputs compared to MatDot Codes. For the experiments whose results are shown in Fig. \ref{fig:fixparerr}, the entries of the input matrices $\mathbf{A},\mathbf{B}$ are chosen independently according to the standard Gaussian distribution $\mathcal{N}(0,1)$. In addition, for any two input matrices $\mathbf{A}, \mathbf{B}$, let $\hat{\mathbf{C}}$ be the output of the distributed system (which is not necessarily equal to the correct answer $\mathbf{A}\mathbf{B}$), we define the  relative error between $\mathbf{A}\mathbf{B}$ and $\hat{\mathbf{C}}$  to be 
\begin{align}
    E_r(\mathbf{A}\mathbf{B},\hat{\mathbf{C}})=\frac{||\mathbf{A}\mathbf{B}-\hat{\mathbf{C}}||_F}{||\mathbf{A}\mathbf{B}||_F}.
\end{align}
\begin{figure*}[!t]
    \centering
    \includegraphics[scale=0.50]{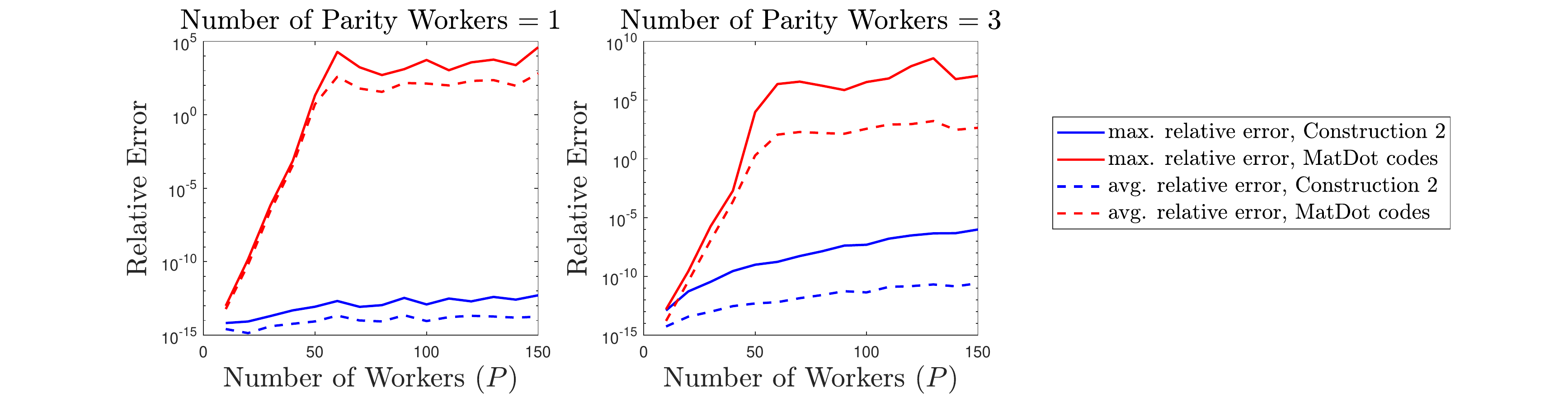}
    \caption{The growth of the relative error, for both \fbl{Chebyshev polynomials based OrthoMatDot} and MatDot Codes, both using Chebyshev points, with the system size given a fixed number of redundant worker nodes.}
    \label{fig:fixparerr}
  \rule{\textwidth}{0.7pt}
\end{figure*}
Fig. \ref{fig:fixparerr} shows how the maximum relative error (the worst case relative error given a fixed number of parity workers $s$ among all the $P-s$ successful nodes scenarios) grows with the size of the distributed system. In Fig. \ref{fig:fixparerr}, we plot the average result of five different realizations of the system at each system size $P$. The figure shows that  MatDot Codes crushes after the size of the system exceeds $50$ workers, providing a relative error of around $10^5$. On the other hand, \fbl{our OrthoMatDot construction} can support systems with sizes up to $150$ worker nodes only allowing for a relative error $< 10^{-5}$. It is also worth mentioning that in our experiments, we use the MATLAB command $inv()$  \cite{matlab}  for matrix inversion. We have also tried matrices inversion through the Bjork-Pereyra algorithm \cite{bjorck1970solution}, however, its results were much less accurate than $inv()$, especially for large systems with a number of worker nodes $> 50$.   
\begin{remark}
A main challenge in this work is that we assume operations over the real field. For  finite fields, one can always perform arithmetic operations with no errors. Although this fact may motivate a simple solution to the numerical stability of real-valued computations by rounding the computation's inputs to a finite field's elements and performing computations over this finite field, such solution has limited applicability, especially for inputs with wide range, due to the following reason. Since performing arithmetic operations over a finite field $\mathbb{F}_{2^n}$ requires representing each element of $\mathbb{F}_{2^n}$ as an element in $\mathbb{F}_{2}^{n}$ through a bit representation, this solution is  applicable in machines with fixed point operations and word sizes of at least $n$. However, the solution is not applicable in machines with floating point operations since in floating point representation not all the intermediate values between the minimum and the maximum representable values can be represented, this is a drawback of the floating point representation over the fixed point representation, though floating point representation can represent a wider range of values than fixed point representation for the same word size.
\end{remark}

\section{\fbl{OrthoPoly: Low Communication/Computation Numerically Stable Codes for Distributed Marix Multiplication}}\label{sec:chebpolycod}
While MatDot Codes \cite{allerton17} have an optimal recovery threshold of $2m-1$, they have relatively higher computation cost per worker and worker node to fusion node communication cost  as compared to Polynomial Codes \cite{polynomialcodes}. 
In this section, motivated by the condition number bound in Theorem \ref{thm:bound}, we use the idea of using Chebyshev polynomials to provide a numerically stable code construction for matrix multiplication that has the same low communication/computation costs as Polynomial Codes, as well as the same recovery threshold. However, as will be shown in this section, our proposed codes, \fbl{denoted by \emph{OrthoPoly}}, provides lower numerical errors than Polynomial Codes.  In this section, we follow the same system model as in Section \ref{sec:sysmod}, and solve  the problem statement formulated    in Section  \ref{sec:chebprobfor}. We  provide a motivating example in Section \ref{sec:chebpolyex}, then we provide the general code construction in Section \ref{sec:chebpolygen}. Finally, in Section \ref{sec:chebpolyexp}, we show experimentally that \fbl{OrthoPoly Codes}  achieve lower numerical errors as compared to Polynomial Codes.

\subsection{Problem Formulation}\label{sec:chebprobfor}
The master node possesses two real-valued input matrices $\mathbf{A}$, $\mathbf{B}$ with dimensions $N_1\times N_2$, $N_2\times N_3$, respectively. Every worker node receives from the master node an encoded matrix of $\mathbf{A}$ of dimension $N_1/m \times N_2$ and  an encoded matrix of $\mathbf{B}$ of dimension $N_2 \times N_3/n$, and performs matrix multiplication of these two received inputs. Upon performing the matrix multiplication, each worker node sends the result to the fusion node. 
The fusion node needs to recover the matrix multiplication $\mathbf{A}\mathbf{B}$ once it receives the results of any $mn$ worker nodes.
\subsection{Example $(m=n=3)$}
\label{sec:chebpolyex}
Consider computing the matrix multiplication $\mathbf{A}\mathbf{B}$, for some two real matrices $\mathbf{A}, \mathbf{B}$ of dimensions $N_1\times N_2$ and $N_2 \times N_3$, respectively, over a distributed system of $P \geq 9$ workers such that:
\begin{enumerate}
    \item Each worker receives an encoded matrix of $\mathbf{A}$ of dimension  $N_1 /3\times N_2$, and an encoded matrix of $\mathbf{B}$ of dimension $N_2 \times N_3/3$.
    \item The product $\mathbf{A}\mathbf{B}$ can be recovered by the fusion node given the results of any $9$ worker nodes. 
\end{enumerate}
A solution can be as follows: First, matrices $\mathbf{A}, \mathbf{B}$ can be partitioned as
\begin{align}
\label{eq:polsplt}
 \mathbf{A}&=\left(\begin{array}{ccc}\mathbf{A}_{0}\\\mathbf{A}_{1}\\\mathbf{A}_{2}\end{array}\right),~~
 \mathbf{B}=\left(\begin{array}{ccc}\mathbf{B}_{0}~~  \mathbf{B}_{1}~~\mathbf{B}_{2}\end{array}\right), 
\end{align}
where, for any $i\in\{0,1,2\}$, $\mathbf{A}_{i}$ has dimension $N_1/3 \times N_2$, and $\mathbf{B}_{i}$ has dimension $N_2 \times N_3/3$.  Next, let 
\begin{align}
p_\mathbf{A}(x)=\mathbf{A}_{0}T_{0}(x)+\mathbf{A}_{1}T_{1}(x)+\mathbf{A}_{2}T_{2}(x),\notag\\
p_\mathbf{B}(x)=\mathbf{B}_{0}T_{0}(x)+\mathbf{B}_{1} T_{3}(x)+\mathbf{B}_{2}T_{6}(x).\notag
\end{align}
Now, $p_\mathbf{A}(x)p_\mathbf{B}(x)$ can be written as 
\begin{align}
p_\mathbf{A}(x)p_\mathbf{B}(x)&=\big(\mathbf{A}_{0}T_{0}(x)+\mathbf{A}_{1}T_{1}(x)+\mathbf{A}_{2}T_{2}(x)\big)\big(\mathbf{B}_{0}T_{0}(x)+\mathbf{B}_{1} T_{3}(x)+\mathbf{B}_{2}T_{6}(x)\big)\notag\\
&= \mathbf{A}_0\mathbf{B}_0+\big(\mathbf{A}_1\mathbf{B}_0+\frac{1}{2}\mathbf{A}_2\mathbf{B}_1\big)T_1(x)+\big(\mathbf{A}_2\mathbf{B}_0+\frac{1}{2}\mathbf{A}_1\mathbf{B}_1\big)T_2(x)\notag\\
&+  \mathbf{A}_0\mathbf{B}_1 T_3(x)+\frac{1}{2}\big(\mathbf{A}_1\mathbf{B}_1+\mathbf{A}_2\mathbf{B}_2\big)T_4(x)+\frac{1}{2}\big(\mathbf{A}_1\mathbf{B}_2+\mathbf{A}_2\mathbf{B}_1\big)T_5(x)\notag\\
&+\mathbf{A}_0\mathbf{B}_2 T_6(x)+\frac{1}{2}\mathbf{A}_1\mathbf{B}_2T_7(x)+\frac{1}{2}\mathbf{A}_2\mathbf{B}_2T_8(x)
\end{align}
Since $p_\mathbf{A}(x)p_\mathbf{B}(x)$ is a degree $8$ polynomial, once the fusion node receives the output of any $9$ workers, it can interpolate  $p_\mathbf{A}(x)p_\mathbf{B}(x)$, i.e., obtain its matrix coefficients, let such matrix coefficients be $\mathbf{C}_{T_0}, \cdots, \mathbf{C}_{T_8}$.
Specifically, for any $i\in \{0, \cdots, 8\}$, let $\mathbf{C}_{T_i}$ be the matrix coefficient of $T_i$ in $p_\mathbf{A}(x)p_\mathbf{B}(x)$. 
Now, recalling (\ref{eq:polsplt}), the product $\mathbf{A}\mathbf{B}$ can be written as 
\begin{align}
\mathbf{A}\mathbf{B}=    \left(\begin{array}{ccc}\mathbf{A}_0\mathbf{B}_0&  \mathbf{A}_0\mathbf{B}_1& \mathbf{A}_0\mathbf{B}_2\\\mathbf{A}_1\mathbf{B}_0& \mathbf{A}_1\mathbf{B}_1 &\mathbf{A}_1\mathbf{B}_2\\\mathbf{A}_2\mathbf{B}_0&\mathbf{A}_2\mathbf{B}_1 &\mathbf{A}_2\mathbf{B}_2\end{array}\right).
\end{align}
While the obtained set of matrix coefficients $\{\mathbf{C}_{T_i}: i \in \{0, \cdots, 8\}\}$ is not equal to $\{\mathbf{A}_i\mathbf{B}_j : i,j\in\{0,1,2\}\}$, $\mathbf{C}_{T_{i}}$'s are linear combinations of $\mathbf{A}_i\mathbf{B}_j$'s. Specifically, for any $\mathbf{C}_{T_i}$, $i \in \{0, \cdots, 8\}$, let $\mathbf{C}_{T_i}^{(k,l)}$ be its $(k,l)$-th entry, and, for any $i,j \in\{0,1,2\}$, let $(\mathbf{A}_i\mathbf{B}_j)^{(k,l)}$ be the $(k,l)$-th entry of the product $\mathbf{A}_i\mathbf{B}_j$, we can write
\begin{align}\label{eq:tcoef}
\left(\begin{array}{c}\mathbf{C}_{T_0}^{(k,l)}\\  \mathbf{C}_{T_1}^{(k,l)}\\ \mathbf{C}_{T_2}^{(k,l)}\\\mathbf{C}_{T_3}^{(k,l)}\\ \mathbf{C}_{T_4}^{(k,l)}\\\mathbf{C}_{T_5}^{(k,l)}\\\mathbf{C}_{T_6}^{(k,l)}\\\mathbf{C}_{T_7}^{(k,l)}\\\mathbf{C}_{T_8}^{(k,l)}\end{array}\right)= \left(\begin{array}{ccccccccc}
1&0&0&0&0&0&0&0&0\\0&1&0&0&0&1/2&0&0&0\\0&0&1&0&1/2&0&0&0&0
\\0&0&0&1&0&0&0&0&0\\0&0&0&0&1/2&0&0&0&1/2\\0&0&0&0&0&1/2&0&1/2&0\\0&0&0&0&0&0&1&0&0\\0&0&0&0&0&0&0&1/2&0\\0&0&0&0&0&0&0&0&1/2
\end{array}\right)\left(\begin{array}{c}(\mathbf{A}_0\mathbf{B}_0)^{(k,l)}\\  (\mathbf{A}_1\mathbf{B}_0)^{(k,l)}\\ (\mathbf{A}_2\mathbf{B}_0)^{(k,l)}\\(\mathbf{A}_0\mathbf{B}_1)^{(k,l)}\\ (\mathbf{A}_1\mathbf{B}_1)^{(k,l)}\\(\mathbf{A}_2\mathbf{B}_1)^{(k,l)}\\(\mathbf{A}_0\mathbf{B}_2)^{(k,l)}\\(\mathbf{A}_1\mathbf{B}_2)^{(k,l)}\\(\mathbf{A}_2\mathbf{B}_2)^{(k,l)}\end{array}\right),
\end{align}
for any $(k,l) \in [N_1/3]\times[N_3/3]$.
Thus, the products $\mathbf{A}_i\mathbf{B}_j,  i,j\in\{0,1,2\}$ can be obtained by computing 
\begin{align}\label{eq:ABcoef}
\left(\begin{array}{c}(\mathbf{A}_0\mathbf{B}_0)^{(k,l)}\\  (\mathbf{A}_1\mathbf{B}_0)^{(k,l)}\\ (\mathbf{A}_2\mathbf{B}_0)^{(k,l)}\\(\mathbf{A}_0\mathbf{B}_1)^{(k,l)}\\ (\mathbf{A}_1\mathbf{B}_1)^{(k,l)}\\(\mathbf{A}_2\mathbf{B}_1)^{(k,l)}\\(\mathbf{A}_0\mathbf{B}_2)^{(k,l)}\\(\mathbf{A}_1\mathbf{B}_2)^{(k,l)}\\(\mathbf{A}_2\mathbf{B}_2)^{(k,l)}\end{array}\right)=\left(\begin{array}{ccccccccc}
1&0&0&0&0&0&0&0&0\\0&1&0&0&0&1/2&0&0&0\\0&0&1&0&1/2&0&0&0&0
\\0&0&0&1&0&0&0&0&0\\0&0&0&0&1/2&0&0&0&1/2\\0&0&0&0&0&1/2&0&1/2&0\\0&0&0&0&0&0&1&0&0\\0&0&0&0&0&0&0&1/2&0\\0&0&0&0&0&0&0&0&1/2
\end{array}\right)^{-1} \left(\begin{array}{c}\mathbf{C}_{T_0}^{(k,l)}\\  \mathbf{C}_{T_1}^{(k,l)}\\ \mathbf{C}_{T_2}^{(k,l)}\\\mathbf{C}_{T_3}^{(k,l)}\\ \mathbf{C}_{T_4}^{(k,l)}\\\mathbf{C}_{T_5}^{(k,l)}\\\mathbf{C}_{T_6}^{(k,l)}\\\mathbf{C}_{T_7}^{(k,l)}\\\mathbf{C}_{T_8}^{(k,l)}\end{array}\right),
\end{align}
for all $(k,l) \in [N_1/3]\times[N_3/3]$.
In the following, we provide the general code construction.

\subsection{\fbl{OrthoPoly} Code Construction}\label{sec:chebpolygen}
We assume that matrix $\mathbf{A}$  is split horizontally into $m$ equal sub-matrices, of dimension $N_1/m\times N_2$ each, and matrix $\mathbf{B}$ is split vertically into $n$ equal sub-matrices, of dimension $N_2 \times N_3/n $ each, as follows:
\begin{equation}\label{eq:spltApol}
\mathbf{A}=\left(\begin{array}c \mathbf{A}_0\\\mathbf{A}_1\\\vdots \\\mathbf{A}_{m-1}\end{array}\right),\;\;\;\mathbf{B} = \left(\mathbf{B}_0 \ \mathbf{B}_1 \ \ldots \ \mathbf{B}_{n-1}\right),
\end{equation}
and define two encoding polynomials $p_\mathbf{A}(x)=\sum_{i=0}^{m-1} \mathbf{A}_i T_{i}(x)$ and $p_\mathbf{B}(x)=\sum_{i=0}^{n-1} \mathbf{B}_i T_{im}(x),$ and let $p_\mathbf{C}(x)=p_{\mathbf{A}}(x)p_\mathbf{B}(x)$. We describe, next, the idea of the general code construction. First, for all $r \in [P]$, the master node sends to the $r$-th worker evaluations of $p_\mathbf{A}(x)$ and $p_\mathbf{B}(x)$ at $x=\rho^{(P)}_r$, that is, it sends $p_\mathbf{A}(\rho^{(P)}_r)$ and $p_\mathbf{B}(\rho^{(P)}_r)$ to the $r$-th worker. Next, for every $r \in [P]$, the $r$-th worker node computes the matrix product $p_\mathbf{C}(\rho^{(P)}_r)=p_{\mathbf{A}}(\rho^{(P)}_r)p_{\mathbf{B}}(\rho^{(P)}_r)$ and sends the result to the fusion node. Once the fusion node receives the output of any $mn$ worker nodes, it interpolates $p_\mathbf{C}(x)$. Next, the fusion node recovers the products $\mathbf{A}_i\mathbf{B}_j, i \in \{0, \cdots, m-1\}, j \in \{0, \cdots, n-1\}$, from the matrix coefficients of $p_\mathbf{C}(x)$ using a low complexity matrix-vector multiplication, specified later in Construction \ref{con:polycheb2}. 
We formally present our \fbl{OrthoPoly} Codes in Construction \ref{con:polycheb2}.  In the following, we explain the notation used in Construction \ref{con:polycheb2}.
The output of the algorithm is the  $N_1\times N_3$ matrix $\hat{\mathbf{C}},$ where the $(k,l)$-th block  of $\hat{\mathbf{C}}$ is the $N_1/m \times N_3/n$ matrix $\hat{\mathbf{C}}_{k,l}$, and   the $(i,j)$-th entry of any matrix $\hat{\mathbf{C}}_{k,l}$ is $\hat{{c}}_{k,l}^{(i,j)}$.  The $(i,j)$-th entry of  the matrix polynomial $p_{\mathbf{C}}(x)$ is denoted as $p^{(i,j)}_{\mathbf{C}}(x)$, and Section \ref{sec:notation} defines matrices $\mathbf{G}^{(mn,P)}(\boldsymbol{\rho}^{(P)})$ and $\mathbf{G}^{(mn,P)}_{\mathcal{R}}(\boldsymbol{\rho}^{(P)})$, for any subset $\mathcal{R}=\{{r_1}, \cdots, {r_{mn}}\}  \subset [P]$.
In addition, $\mathbf{H}$ is an $mn \times mn$ matrix of the following form $\mathbf{H}=\left(\begin{array}{ccc}\mathbf{H}_{0}~~  \mathbf{H}_{1}~~\cdots~~\mathbf{H}_{n-1}\end{array}\right),$ where $\mathbf{H}_0$ is an $mn \times m$ matrix with ones on the main diagonal and zeros elsewhere, and for any $i\in \{1, \cdots, n-1\}$, $\mathbf{H}_i$ is an $mn\times m$ matrix of the following structure 
$$\mathbf{H}_i=\left(\begin{array}{ccccccccc}
0&0&0&\cdots&0\\\vdots&\vdots&\vdots&\vdots&\vdots\\0&0&0&\cdots&0\\\vdots&\vdots&\vdots&\iddots&1/2\\0&0&0&\iddots&0\\0&0&1/2&\iddots&\vdots\\0&1/2&0&\cdots&0\\1&0&0&\cdots&0\\0&1/2&0&\cdots&0\\0&0&1/2&\ddots&\vdots\\0&0&0&\ddots&0\\\vdots&\vdots&\vdots&\ddots&1/2\\0&0&0&\cdots&0\\\vdots&\vdots&\vdots&\vdots&\vdots\\0&0&0&\cdots&0
\end{array}\right),$$
where the value $1$ in the first column is at the $(im+1)$-th row of $\mathbf{H}_i$.

\begin{algorithm}
\caption{\fbl{OrthoPoly}: \textbf{Inputs}: $\mathbf{A}, \mathbf{B}$,~~\textbf{Output}: $\hat{\mathbf{C}}$} \label{con:polycheb2}
\begin{algorithmic}[1]
\Procedure{MasterNode}{$\mathbf{A},\mathbf{B}$}\Comment{The master node's procedure}
\State $r \gets 1$
\While{$r\not=P+1$}
\State $p_\mathbf{A}(\rho^{(P)}_r) \gets \sum_{i=0}^{m-1}\mathbf{A}_{i} T_{i}(\rho^{(P)}_r)$
\State $p_\mathbf{B}(\rho^{(P)}_r)\gets\sum_{i=0}^{n-1} \mathbf{B}_{i} T_{im}(\rho^{(P)}_r)$ 
\State \textbf{send} $p_{\mathbf{A}}(\rho^{(P)}_r), p_{\mathbf{B}}(\rho^{(P)}_r)$ \textbf{to worker node} $r$
\State $r \gets r+1$
\EndWhile
\EndProcedure
\State
\Procedure{WorkerNode}{$p_{\mathbf{A}}(\rho^{(P)}_r), p_{\mathbf{B}}(\rho^{(P)}_r)$}\Comment{The procedure of worker node $r$}
\State $p_\mathbf{C}(\rho^{(P)}_r) \gets p_{\mathbf{A}}(\rho^{(P)}_r) p_{\mathbf{B}}(\rho^{(P)}_r)$
\State \textbf{send} $p_{\mathbf{C}}(\rho^{(P)}_r)$ \textbf{to the fusion node}
\EndProcedure
\State
\Procedure{FusionNode}{$\{p_{\mathbf{C}}(\rho^{(P)}_{r_1}), \cdots, p_{\mathbf{C}}(\rho^{(P)}_{r_{mn}})\}$}\Comment{The fusion node's procedure, $r_i$'s are distinct}
\State  $\mathbf{G}_{\operatorname{inv}}\gets \left(\mathbf{G}^{(mn,P)}_{\mathcal{R}}\right)^{-1}$
\For{$i\in [N_1/m]$}
\For{$j\in [N_3/n]$}
\State $(c^{(i,j)}_0, \cdots, c^{(i,j)}_{mn-1}) \gets (p_{\mathbf{C}}^{(i,j)}(\rho^{(P)}_{r_1}), \cdots,p_{\mathbf{C}}^{(i,j)}(\rho^{(P)}_{r_{mn}}))\mathbf{G}_{\operatorname{inv}}$
\State $(\hat{c}_{0,0}^{(i,j)} \cdots \hat{c}_{m-1,0}^{(i,j)} \cdots\cdots \hat{c}_{0,n-1}^{(i,j)} \cdots \hat{c}_{m-1,n-1}^{(i,j)}) \gets (c^{(i,j)}_0, \cdots, c^{(i,j)}_{mn-1}) (\mathbf{H}^{-1})^T$
\EndFor
\EndFor
\State \textbf{return} $\hat{\mathbf{C}}$
\EndProcedure
\end{algorithmic}
\end{algorithm}
\subsubsection{Complexity Analyses of \fbl{OrthoPoly}}\label{sec:con4complexity}
\hspace{2mm}

\textbf{Encoding Complexity}: 
Encoding for each worker requires performing two additions, the first one adds $m$ scaled matrices of size $N_1 N_2/ m$ and the other adds $n$ scaled matrices of size $N_2 N_3 /n$,  for an overall encoding complexity for each worker of $O(N_1 N_2+N_2N_3)$. Therefore, the overall computational complexity of encoding for $P$ workers is $O(N_1 N_2P+N_2N_3P)$.

\textbf{Computational Cost per Worker}: Each worker multiplies two matrices of  dimensions $N_1  /m\times N_2$ and $N_2\times N_3/n$, requiring $O(N_1N_2N_3/mn)$ operations.

\textbf{Decoding Complexity}:
Since $p_\mathbf{A}(x)p_\mathbf{B}(x)$ has degree $mn-1$, the interpolation of $p_\mathbf{C}(x)$ requires the inversion of a $mn \times mn$ matrix, with complexity $O(m^3n^3)$, and performing $N_1N_3/mn$ matrix-vector multiplications, each of them is between the inverted  matrix and a  column vector of length $mn$ of the received evaluations of the matrix polynomial $p_\mathbf{C}(x)$ at some position $(i,j) \in [N_1/m] \times [N_3/n]$, with complexity $O(N_1N_3m^2n^2/(mn))=O(N_1N_3mn)$.  Thus, assuming that $mn \ll N_1,N_3$, the overall decoding complexity is $O(N_1N_3 m n)$.

\textbf{Communication Cost}:
The master node sends $O( N_1N_2P/m+N_2N_3P/n)$ symbols, and the fusion node receives $O(N_1 N_3)$ symbols from the successful worker nodes. 

\vspace{2mm}
\begin{remark}
With the reasonable assumption that the dimensions of the input matrices $\mathbf{A},\mathbf{B}$ are large enough such that $N_1,N_2,N_3 \gg m,n,P$, we can conclude that the encoding and decoding costs at the master and fusion nodes, respectively, are negligible compared to the computation cost at each worker node.
\end{remark}
\subsection{Numerical Results}\label{sec:chebpolyexp}
In our experiments, the entries of the input matrices $\mathbf{A},\mathbf{B}$ are chosen independently according to the standard Gaussian distribution $\mathcal{N}(0,1)$. In addition, for any two input matrices $\mathbf{A}, \mathbf{B}$, let $\hat{\mathbf{C}}$ be the output of the distributed system, we define the  relative error between $\mathbf{A}\mathbf{B}$ and $\hat{\mathbf{C}}$  to be 
\begin{align*}
    E_r(\mathbf{A}\mathbf{B},\hat{\mathbf{C}})=\frac{||\mathbf{A}\mathbf{B}-\hat{\mathbf{C}}||_F}{||\mathbf{A}\mathbf{B}||_F}.
\end{align*}
Fig. \ref{fig:fixparerrpoly} shows how the maximum relative error (the worst case relative error given a fixed number of parity workers $s$ among all the $P-s$ successful nodes scenarios) grows with the size of the distributed system for both Construction \ref{con:polycheb2} and Polynomial Codes. In Fig. \ref{fig:fixparerrpoly}, we plot the average result of five different realizations of the system at each system size $P$. The figure shows that  Polynomial Codes have unacceptable relative errors after the size of the system exceeds $50$ workers, providing a relative error of around $10^5$. On the other hand, \fbl{OrthoPoly} can support systems with sizes up to $170$ worker nodes only allowing for a relative error $< 10^{-5}$. 
\begin{figure*}[!t]
    \centering
    \includegraphics[scale=0.50]{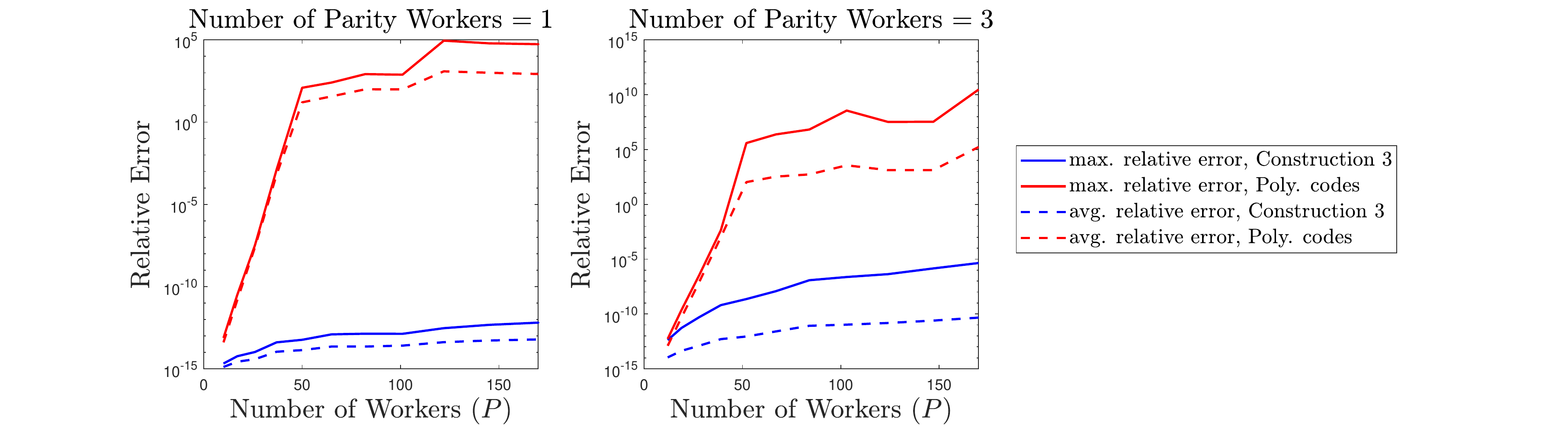}
    \caption{The growth of the relative error, for both \fbl{OrthoPoly} and Polynomial Codes, both using Chebyshev points, with the system size given a fixed number of redundant worker nodes.}
    \label{fig:fixparerrpoly}
  \rule{\textwidth}{0.7pt}
\end{figure*}

\section{\fbl{Generalized OrthoMatDot: Numerically stable  Codes for  Matrix Multiplication with Communication/Computation-Recovery Threshold Trade-off}}\label{sec:polydot-chebyshev}
Although MatDot Codes \cite{allerton17} have a low recovery threshold of $2m-1$ as compared with Polynomial Codes \cite{polynomialcodes} which have a recovery threshold of $mn$, MatDot Codes' worker to fusion nodes communication cost and  computation cost per worker are higher than Polynomial Codes. 
 Codes proposed in  \cite{arxiv_allerton17,genPolyDot,entPoly} offer a trade-off between the communication/computation cost and the recovery threshold. However, all of these codes are based on the ``ill-conditioned'' monomial basis.  In this section, we offer a numerically stable code construction, \fbl{denoted by \emph{Generalized OrthoMatDot},}  that offers  a trade-off between communication/computation costs and recovery threshold. Our construction incurs a higher recovery threshold than the codes of \cite{genPolyDot, entPoly} by a factor of at most $4$ for the same communication/computation cost. We provide in Section \ref{sec:probfor2} the formal problem statement considered in this section. We describe an example of our construction in Section \ref{sec:polydot_example}, provide the general code construction in Section \ref{sec:Chebyshev_polydot}, and describe our numerical experiments in Section \ref{sec:numerical_polydot}.

\subsection{System Model and Problem Formulation}\label{sec:probfor2}
We consider the same system model and problem formulation as in Section \ref{sec:sysprob} with the following change:  We assume that the master node is allowed to send an encoded ${1}/{m}$ fraction of matrix $\mathbf{A}$, and an encoded ${1}/{n}$ fraction of matrix $\mathbf{B}$, where $m$ and $n$ are not necessarily equal, and $\mathbf{A}$ and $\mathbf{B}$ are split as follows
\begin{align}
\label{eq:spltA2}
 \mathbf{A}&=\left(\begin{array}{ccc}\mathbf{A}_{0,0}&  \cdots& \mathbf{A}_{0,m_2-1}\\\vdots& \ddots &\vdots\\\mathbf{A}_{m_1-1,0} &\cdots &\mathbf{A}_{m_1-1,m_2-1}\end{array}\right),\notag\\
 \mathbf{B}&=\left(\begin{array}{ccc}\mathbf{B}_{0,0}&  \cdots& \mathbf{B}_{0,m_3-1}\\\vdots& \ddots &\vdots\\\mathbf{B}_{m_2-1,0} &\cdots &\mathbf{B}_{m_2-1,m_3-1}\end{array}\right), 
\end{align}
 where $m_1,m_2,m_3$ divide $N_1,N_2,N_3$, respectively, and $m=m_1m_2, n=m_2m_3$.
In addition, we assume that each worker node receives a linear combination of sub-matrices $\mathbf{A}_{i,j}$, and another linear combination of sub-matrices $\mathbf{B}_{i,j}$. 

\begin{remark}
Although, in this section, we offer \fbl{Generalized OrthoMatDot,} a code construction  with lower condition numbers than codes in \cite{genPolyDot, entPoly}, the recovery threshold of our codes are higher by a factor of at most $4$ than the codes of these references. Specifically,  \fbl{Generalized OrthoMatDot} codes have a recovery threshold of $4m_1m_2m_3-2(m_1m_2+m_2m_3+m_3m_1)$ $+m_1+2m_2+m_3-1$ while both codes in \cite{genPolyDot,entPoly} have a recovery threshold of $m_1m_2m_3+m_2-1$. This increased recovery threshold is due to the fact that \fbl{Generalized OrthoMatDot} Codes are based on Chebyshev polynomials which have the following property: For any $i,j \in \mathbb{N}$, $T_{i}(x)T_{j}(x)=1/2~(T_{i+j}(x)+T_{|i-j|}(x))$. This property allows for a higher number of \emph{undesired} terms in the multiplication of the encoding polynomials $p_{\mathbf{A}}(x), p_{\mathbf{B}}(x)$. In order to avoid combining undesired and desired terms at the same degree, higher degree Chebyshev polynomials have to be used in $p_{\mathbf{B}}(x)$, yielding a higher recovery threshold. It is still an open question  whether the recovery threshold in \cite{genPolyDot, entPoly} can be achieved using orthonormal polynomials.     
\end{remark}

\subsection{Example $(m_1=m_2=m_3=2)$}
\label{sec:polydot_example}
Consider computing the matrix multiplication $\mathbf{A}\mathbf{B}$, for some two real matrices $\mathbf{A}, \mathbf{B}$ of dimensions $N_1\times N_2$ and $N_2 \times N_3$, respectively, over a distributed system of $P \geq 15$ workers such that:
\begin{enumerate}
    \item Each worker receives an encoded matrix of $\mathbf{A}$ of dimension  $N_1/2 \times N_2/2$, and an encoded matrix of $\mathbf{B}$ of dimension $N_2/2 \times N_3/2$.
    \item The product $\mathbf{A}\mathbf{B}$ can be recovered by the fusion node given the results of any $15$ worker nodes. 
\end{enumerate}
A solution can be as follows: First, matrices $\mathbf{A}, \mathbf{B}$ can be partitioned as
\begin{align}
\label{eq:exsplt}
 \mathbf{A}&=\left(\begin{array}{ccc}\mathbf{A}_{0,0}&  \mathbf{A}_{0,1}\\\mathbf{A}_{1,0}  &\mathbf{A}_{1,1}\end{array}\right),
 \mathbf{B}=\left(\begin{array}{ccc}\mathbf{B}_{0,0}& \mathbf{B}_{0,1}\\\mathbf{B}_{1,0} &\mathbf{A}_{1,1}\end{array}\right), 
\end{align}
where, for $i,j \in\{0,1\}$, $\mathbf{A}_{i,j}$ has dimension $N_1/2 \times N_2/2$, and $\mathbf{B}_{i,j}$ has dimension $N_2/2 \times N_3/2$.  Next, let 
\begin{align}
p_\mathbf{A}(x)=\mathbf{A}_{0,0}T_{1}(x)+\frac{1}{2}\mathbf{A}_{0,1}T_{0}(x)+\mathbf{A}_{1,0}T_{\alpha+1}(x)+\mathbf{A}_{1,1}T_{\alpha}(x),\notag\\
p_\mathbf{B}(x)=\frac{1}{2}\mathbf{B}_{0,0}T_{0}(x)+\mathbf{B}_{1,0} T_{1}(x)+\mathbf{B}_{0,1}T_{\beta}(x)+\mathbf{B}_{1,1}T_{\beta+1}(x),\notag
\end{align}
where $\alpha, \beta$ to be specified next, and define $P$ distinct real numbers $x_1,x_2, \cdots, x_P$ in the range $[-1,1]$. For each worker node $r \in [P]$, the master node sends $p_{\mathbf{A}}(x_r)p_{\mathbf{B}}(x_r)$.

 {Now, in order to specify the best values for $\alpha,\beta$, we expand the polynomial $p_\mathbf{A}(x)p_\mathbf{B}(x)$ in the Chebyshev basis, and then point out some observations. 
\begin{align}\label{eq:expandab}
    p_\mathbf{A}(x)p_\mathbf{B}(x)&=\frac{1}{4} \mathbf{A}_{0,1}\mathbf{B}_{0,0}+\frac{1}{2} \mathbf{A}_{0,0}\mathbf{B}_{0,0}T_{1}(x)+\frac{1}{2}\mathbf{A}_{0,1}\mathbf{B}_{1,0}T_{1}(x)+ \mathbf{A}_{0,0}\mathbf{B}_{1,0}T_1(x)T_{1}(x)\notag\\
    &+ \frac{1}{2} \mathbf{A}_{1,1}\mathbf{B}_{0,0}T_{\alpha}(x)+  \mathbf{A}_{1,1}\mathbf{B}_{1,0}T_{1}(x)T_{\alpha}(x)+ \frac{1}{2} \mathbf{A}_{1,0}\mathbf{B}_{0,0}T_{\alpha+1}(x)+  \mathbf{A}_{1,0}\mathbf{B}_{1,0}T_{1}(x)T_{\alpha+1}(x)\notag\\
    &+ \frac{1}{2} \mathbf{A}_{0,1}\mathbf{B}_{0,1}T_{\beta}(x)+  \mathbf{A}_{0,0}\mathbf{B}_{0,1}T_{1}(x)T_{\beta}(x)+ \frac{1}{2} \mathbf{A}_{0,1}\mathbf{B}_{1,1}T_{\beta+1}(x)+  \mathbf{A}_{0,0}\mathbf{B}_{1,1}T_{1}(x)T_{\beta+1}(x)\notag\\
     &+ \mathbf{A}_{0,1}\mathbf{B}_{1,1}T_{1}(x)T_{\beta+1}(x)+  \mathbf{A}_{1,0}\mathbf{B}_{0,1}T_{\alpha+1}(x)T_{\beta}(x)+\mathbf{A}_{1,1}\mathbf{B}_{1,1}T_{\alpha}(x)T_{\beta+1}(x) \notag\\
     &+\mathbf{A}_{1,0}\mathbf{B}_{1,1}T_{\alpha+1}(x)T_{\beta+1}(x).
\end{align}
Using the property of the Chebyshev polynomials that  for any $i,j \in \mathbb{N}$, $T_{i}(x)T_{j}(x)=1/2~(T_{i+j}(x)+T_{|i-j|}(x))$, (\ref{eq:expandab}) can be rewritten as
\begin{align}\label{eq:expandab2}
    p_\mathbf{A}(x)p_\mathbf{B}(x)&=\frac{1}{4} \mathbf{A}_{0,1}\mathbf{B}_{0,0}+\frac{1}{2}\mathbf{A}_{0,0}\mathbf{B}_{1,0}+\frac{1}{2} \left(\mathbf{A}_{0,0}\mathbf{B}_{0,0}+\mathbf{A}_{0,1}\mathbf{B}_{1,0}\right)T_{1}(x)+\frac{1}{2} \mathbf{A}_{0,0}\mathbf{B}_{1,0}T_{2}(x)\notag\\
    &+\frac{1}{2}\mathbf{A}_{1,1}\mathbf{B}_{1,0}T_{\alpha-1}(x)+\frac{1}{2} \left(\mathbf{A}_{1,1}\mathbf{B}_{0,0}+\mathbf{A}_{1,0}\mathbf{B}_{1,0}\right) T_{\alpha}(x)+\frac{1}{2} \left(\mathbf{A}_{1,0}\mathbf{B}_{0,0}+\mathbf{A}_{1,1}\mathbf{B}_{1,0}\right) T_{\alpha+1}(x)\notag\\
    &+\frac{1}{2}\mathbf{A}_{1,0}\mathbf{B}_{1,0}T_{\alpha+2}(x)+\frac{1}{2}\mathbf{A}_{1,0}\mathbf{B}_{0,1}T_{\beta-\alpha-1}(x)+
    \frac{1}{2}\left(\mathbf{A}_{1,1}\mathbf{B}_{0,1}+\mathbf{A}_{1,0}\mathbf{B}_{1,1}\right)T_{\beta-\alpha}(x)\notag\\
    &+\frac{1}{2}\mathbf{A}_{1,1}\mathbf{B}_{1,1}T_{\beta-\alpha+1}(x)+\frac{1}{2}\mathbf{A}_{0,0}\mathbf{B}_{0,1}T_{\beta-1}(x)+\frac{1}{2}\left(\mathbf{A}_{0,1}\mathbf{B}_{0,1}+\mathbf{A}_{0,0}\mathbf{B}_{1,1}\right)T_{\beta}(x)\notag\\
    &+ \frac{1}{2}\left(\mathbf{A}_{0,0}\mathbf{B}_{0,1}+\mathbf{A}_{0,1}\mathbf{B}_{1,1}\right)T_{\beta+1}(x)+\frac{1}{2}\mathbf{A}_{0,0}\mathbf{B}_{1,1}T_{\beta+2}(x)+\frac{1}{2}\mathbf{A}_{1,1}\mathbf{B}_{0,1}T_{\beta+\alpha}(x)\notag\\
    &+ \frac{1}{2} \left(\mathbf{A}_{1,0}\mathbf{B}_{0,1}+\mathbf{A}_{1,1}\mathbf{B}_{1,1} \right)T_{\beta+\alpha+1}(x)+\frac{1}{2}\mathbf{A}_{1,0}\mathbf{B}_{1,1}T_{\beta+\alpha+2}(x).
\end{align}
Now, note the following regrading $p_{\mathbf{A}}(x)p_{\mathbf{B}}(x)$ in (\ref{eq:expandab2}):
\begin{enumerate}[(i)]
    \item $\frac{1}{2}\left(\mathbf{A}_{0,0}\mathbf{B}_{0,0}+\mathbf{A}_{0,1}\mathbf{B}_{1,0}\right)$ is the coefficient of $T_{1}(x)$,
    \item $\frac{1}{2}\left(\mathbf{A}_{1,0}\mathbf{B}_{0,0}+\mathbf{A}_{1,1}\mathbf{B}_{1,0}\right)$ is the coefficient of $T_{\alpha+1}(x)$,
    \item $\frac{1}{2}\left(\mathbf{A}_{0,0}\mathbf{B}_{0,1}+\mathbf{A}_{0,1}\mathbf{B}_{1,1}\right)$ is the coefficient of $T_{\beta+1}(x)$,
    \item $\frac{1}{2}\left(\mathbf{A}_{1,0}\mathbf{B}_{0,1}+\mathbf{A}_{1,1}\mathbf{B}_{1,1}\right)$ is the coefficient of $T_{\beta+\alpha+1}(x)$.
\end{enumerate}
}
 {
Since $p_{\mathbf{A}}(x)p_{\mathbf{B}}(x)$ has degree $\beta+\alpha+2$, and this polynomial is evaluated at distinct value at each worker node, once the fusion node receives the output of any $\beta+\alpha+3$ worker nodes, it can interpolate $p_{\mathbf{A}}(x)p_{\mathbf{B}}(x)$ and extract the product $\mathbf{A}\mathbf{B}$ (i.e., the matrix coefficients of $T_1(x),T_{\alpha+1}(x)$,
$ T_{\beta+1}(x), T_{\beta+\alpha+1}(x)$). 
Now, we aim for picking values for $\alpha,\beta$ such that the degree of $p_\mathbf{A}(x)p_\mathbf{B}(x)$ is minimal; and hence, the recovery threshold is minimal as well. These minimal values for $\alpha,\beta$ must be chosen such that the desired coefficients in (i)-(iv) are separate. That is, each of them is neither combined with another desired nor undesired term. This constraint leads to the following two inequalities:
}
$$\alpha-1>1,\text{~and }
\alpha+1<\beta-\alpha-1,$$
which implies that $\alpha=3,\beta=9$. Next, we provide our general code construction for the \fbl{Generalized OrthoMatDot Codes.}

\subsection{\fbl{Generalized OrthoMatDot Code Construction}}
\label{sec:Chebyshev_polydot}

\begin{theorem}\label{thm:genCheb}
For the matrix multiplication problem described in Section \ref{sec:probfor2} computed on the system defined in Section \ref{sec:sysmod}, there exists a coding strategy with recovery threshold 
\begin{align}\label{eq:recgen}
\hspace{-0mm}4m_1m_2&m_3-2(m_1m_2+m_2m_3+m_3m_1)\notag\\
&+m_1+2m_2+m_3-1.
\end{align}
\end{theorem}

Notice that the problem specified in Section \ref{sec:probfor2} restricts the output matrix of each worker node to be of dimension $N_1/m_1 \times N_3/m_3$, for some positive integers $m_1,m_3$ that divide $N_1,N_3$, respectively.  This is smaller than the dimensions of the output matrix of each worker node according to the problem specified in Section \ref{sec:probfor} (i.e., $N_1\times N_3$)  by a factor of $m_1m_3$. However, according to Theorem \ref{thm:genCheb}, \fbl{this communication advantage, when $m_1>1$ or $m_2>1$,  comes at the expense of a higher recovery threshold compared to OrthoMatDot Codes.} 

\begin{remark}[Notation]
For ease of exposition in the remaining of this section, we use $T^{'}_0,T^{'}_1,T^{'}_2, \cdots$ to denote $\frac{1}{2}T_0,T_1,T_2, \cdots$, respectively.
\end{remark}

In order to prove  Theorem \ref{thm:genCheb}, we first present a  code construction that achieves the recovery threshold in (\ref{eq:recgen}), then we prove that the presented code construction is valid. First, note that in \fbl{the Generalized OrthoMatDot} code construction, we assume that the two input matrices $\mathbf{A}, \mathbf{B}$ are split as in (\ref{eq:spltA2}).  Also, note that given this partitioning of input matrices, we can write $\mathbf{C}=\mathbf{A}\mathbf{B}$, where $\mathbf{C}$ is written as
\begin{align}
\mathbf{C}=    \left(\begin{array}{ccc}\mathbf{C}_{0,0}&  \cdots& \mathbf{C}_{0,m_3-1}\\\vdots& \ddots &\vdots\\\mathbf{C}_{m_1-1,0} &\cdots &\mathbf{C}_{m_1-1,m_3-1}\end{array}\right),
\end{align}
and each of $\mathbf{C}_{i,l}$ has dimension ${N_1}/{m_1}\times{N_3}/{m_3}$ and can be expressed as
$\mathbf{C}_{i,l}=\sum_{j=0}^{m_2-1} \mathbf{A}_{i,j}\mathbf{B}_{j,l},$
for any $i\in \{0, 1, \cdots, m_1-1\},$ and $l\in \{0,1, \cdots, m_3-1\}$.
Also, let $x_1, \cdots, x_P$ be distinct real numbers in the range $[-1,1]$, and define encoding polynomials \begin{align}\label{eq:genpolys}
p_\mathbf{A}(x)&=\sum_{i=0}^{m_1-1}\sum_{j=0}^{m_2-1}\mathbf{A}_{i,j} T^{'}_{m_2-1-j+i(2m_2-1)}(x), \notag\\ p_\mathbf{B}(x)&=\sum_{k=0}^{m_2-1}\sum_{l=0}^{m_3-1} \mathbf{B}_{k,l} T^{'}_{k+l(2m_1-1)(2m_2-1)}(x),
\end{align}
and let $p_\mathbf{C}(x)=p_\mathbf{A}(x)p_\mathbf{B}(x)$. Notice that $p_{\mathbf{C}}(x)$ is a polynomial matrix of  degree equals $\deg_\mathbf{C}:=4m_1m_2m_3-2(m_1m_2+m_2m_3+m_3m_1)$  $+m_1+2m_2+m_3-2$.

\begin{claim}\label{cl:chebcoef}
 For any $i\in \{0, 1, \cdots, m_1-1\}$ and $l\in \{0,1, \cdots, m_3-1\}$, $\frac{1}{2}\mathbf{C}_{i,l}$ is the matrix coefficient of $T_{m_2-1+i(2m_2-1)+l(2m_1-1)(2m_2-1)}$ in  $p_{\mathbf{C}}(x)$, 
\end{claim}
The proof of this claim is in Appendix \ref{app:proof}.

 We describe, next, the idea of our proposed \fbl{Generalized OrthoMatDot code construction}. First, for all $r \in [P]$, the master node sends to the $r$-th worker evaluations of $p_\mathbf{A}(x)$ and $p_\mathbf{B}(x)$ at $x=\rho^{(P)}_r$, that is, it sends $p_\mathbf{A}(\rho^{(P)}_r)$ and $p_\mathbf{B}(\rho^{(P)}_r)$ to the $r$-th worker. Next, for every $r \in [P]$, the $r$-th worker node computes the matrix product $p_\mathbf{C}(\rho^{(P)}_r)=p_{\mathbf{A}}(\rho^{(P)}_r)p_{\mathbf{B}}(\rho^{(P)}_r)$ and sends the result to the fusion node. Once the fusion node receives the output of any $\deg_\mathbf{C}+1$ worker nodes, it interpolates $p_\mathbf{C}(x)$.

 We formally present our \fbl{Generalized OrthoMatDot code construction} in Construction \ref{con:gencheb}.  \bl{In the following, we explain the notation used in Construction \ref{con:gencheb}.
 The output of the algorithm is the  $N_1\times N_3$ matrix $\hat{\mathbf{C}},$ where the $(k,l)$-th block  of $\hat{\mathbf{C}}$ is the $N_1/m_1 \times N_3/m_3$ matrix $\hat{\mathbf{C}}_{k,l}$, and   the $(i,j)$-th entry of any matrix $\hat{\mathbf{C}}_{k,l}$ is $\hat{{c}}_{k,l}^{(i,j)}$.  
 The $(i,j)$-th entry of  the matrix polynomial $p_{\mathbf{C}}(x)$ is denoted as  $p^{(i,j)}_{\mathbf{C}}(x)$, and 
 Section \ref{sec:notation} defines matrices $\mathbf{G}^{(\deg_\mathbf{C}+1,P)}(\boldsymbol{\rho}^{(P)})$
and $\mathbf{G}^{(\deg_\mathbf{C}+1,P)}_{\mathcal{R}}(\boldsymbol{\rho}^{(P)})$, for any subset $\mathcal{R}=\{{r_1}, \cdots, {r_{\deg_\mathbf{C}+1}}\}  \subset [P]$.
}



\begin{algorithm}
\caption{\fbl{Generalized OrthoMatDot}: \textbf{Inputs}: $\mathbf{A}, \mathbf{B}$,~~\textbf{Output}: $\hat{\mathbf{C}}$} \label{con:gencheb}
\begin{algorithmic}[1]
\Procedure{MasterNode}{$\mathbf{A},\mathbf{B}$}\Comment{The master node's procedure}
\State $r \gets 1$
\While{$r\not=P+1$}
\State $p_\mathbf{A}(\rho^{(P)}_r) \gets \sum_{i=0}^{m_1-1}\sum_{j=0}^{m_2-1}\mathbf{A}_{i,j} T^{'}_{m_2-1-j+i(2m_2-1)}(\rho^{(P)}_r)$
\State $p_\mathbf{B}(\rho^{(P)}_r)\gets\sum_{k=0}^{m_2-1}\sum_{l=0}^{m_3-1} \mathbf{B}_{k,l} T^{'}_{k+l(2m_1-1)(2m_2-1)}(\rho^{(P)}_r)$ 
\State \textbf{send} $p_{\mathbf{A}}(\rho^{(P)}_r), p_{\mathbf{B}}(\rho^{(P)}_r)$ \textbf{to worker node} $r$
\State $r \gets r+1$
\EndWhile
\EndProcedure
\State
\Procedure{WorkerNode}{$p_{\mathbf{A}}(\rho^{(P)}_r), p_{\mathbf{B}}(\rho^{(P)}_r)$}\Comment{The procedure of worker node $r$}
\State $p_\mathbf{C}(\rho^{(P)}_r) \gets p_{\mathbf{A}}(\rho^{(P)}_r) p_{\mathbf{B}}(\rho^{(P)}_r)$
\State \textbf{send} $p_{\mathbf{C}}(\rho^{(P)}_r)$ \textbf{to the fusion node}
\EndProcedure
\State
\Procedure{FusionNode}{$\{p_{\mathbf{C}}(\rho^{(P)}_{r_1}), \cdots, p_{\mathbf{C}}(\rho^{(P)}_{r_{\deg_\mathbf{C}+1}})\}$}\Comment{The fusion node's procedure, $r_i$'s are distinct}

\State  $\mathbf{G}_{\operatorname{inv}}\gets \left(\mathbf{G}^{(\deg_\mathbf{C}+1,P)}_{\mathcal{R}}\right)^{-1}$
\For{$i\in [N_1/m_1]$}
\For{$j\in [N_3/m_3]$}
\State $(c^{(i,j)}_0, \cdots, c^{(i,j)}_{\deg_\mathbf{C}}) \gets (p_{\mathbf{C}}^{(i,j)}(\rho^{(P)}_{r_1}), \cdots,p_{\mathbf{C}}^{(i,j)}(\rho^{(P)}_{r_{\deg_\mathbf{C}+1}}))\mathbf{G}_{\operatorname{inv}}$
\For{$k\in [m_1]$}
\For{$l\in [m_3]$}
\State $\hat{c}_{k,l}^{(i,j)} \gets 2 c^{(i,j)}_{m_2-1+(k-1)(2m_2-1)+(l-1)(2m_1-1)(2m_2-1)}$ 
\EndFor
\EndFor
\EndFor
\EndFor
\State \textbf{return} $\hat{\mathbf{C}}$
\EndProcedure
\end{algorithmic}
\end{algorithm}






\vspace{3mm}
Now, we  prove Theorem \ref{thm:genCheb}.
\begin{proof}[Proof of Theorem \ref{thm:genCheb}:]
To prove the theorem, it suffices to prove that Construction \ref{con:gencheb} is valid. Noting that $p_\mathbf{A}(x)p_\mathbf{B}(x)$ has degree $4m_1m_2m_3-2(m_1m_2+m_2m_3+m_3m_1)+m_1+2m_2+m_3-2$ and every worker node sends an evaluation of $p_\mathbf{A}(x)p_\mathbf{B}(x)$ at  a distinct point, once the fusion node receives the output of any $4m_1m_2m_3-2(m_1m_2+m_2m_3+m_3m_1)+m_1+2m_2+m_3-1$ worker node, it can interpolate $p_\mathbf{A}(x)p_\mathbf{B}(x)$ (i.e., obtain all its matrix coefficients). This includes the coefficients of $T_{m_2-1+i(2m_2-1)+l(2m_1-1)(2m_2-1)}$ for all  $i\in \{0, 1, \cdots, m_1-1\},$ and $l\in \{0,1, \cdots, m_3-1\}$, i.e., $\mathbf{C}_{i,l}$, for all  $i\in \{0, 1, \cdots, m_1-1\},$ and $l\in \{0,1, \cdots, m_3-1\}$ (Claim \ref{cl:chebcoef}), which completes the proof. 
\end{proof}

Next, we provide the different complexity analyses of \fbl{the Generalized OrthoMatDot Codes.}

\subsubsection{Complexity Analyses of \fbl{Generalized OrthoMatDot}}\label{sec:con3complexity}
\hspace{2mm}

\textbf{Encoding Complexity}: 
Encoding for each worker requires performing two additions, the first one adds $m_1m_2$ scaled matrices of size $N_1 N_2/ (m_1m_2)$ and the other adds $m_2 m_3$ scaled matrices of size $N_2 N_3 /(m_2m_3)$,  for an overall encoding complexity for each worker of $O(N_1 N_2+N_2N_3)$. Therefore, the overall computational complexity of encoding for $P$ workers is $O(N_1 N_2P+N_2N_3P)$.

\textbf{Computational Cost per Worker}: Each worker multiplies two matrices of  dimensions $N_1  /m_1\times N_2/m_2$ and $N_2/m_2\times N_3/m_3$, requiring $O(N_1N_2N_3/(m_1m_2m_3))$ operations.

\textbf{Decoding Complexity}:
Since $p_\mathbf{A}(x)p_\mathbf{B}(x)$ has degree $k-1:=4m_1m_2m_3- 2(m_1m_2+m_2m_3+m_3m_1)+m_1+2m_2+m_3-2$, the interpolation of $p_\mathbf{C}(x)$ requires the inversion of a $k \times k$ matrix, with complexity $O(k^3)=O(m_1^3m_2^3m_3^3)$, and performing $N_1N_3/(m_1m_3)$ matrix-vector multiplications, each of them is between the inverted  matrix and a  column vector of length $k$ of the received evaluations of the matrix polynomial $p_\mathbf{C}(x)$ at some position $(i,j) \in [N_1/m_1] \times [N_3/m_3]$, with complexity $O(N_1N_3k^2/(m_1m_3))=O(N_1N_3m_1m_2^2m_3)$.  Thus, assuming that $m_1,m_3 \ll N_1,N_3$, the overall decoding complexity is $O(N_1N_3m_1m_2^2m_3)=O(N_1N_3 m n)$.

\textbf{Communication Cost}:
The master node sends $O( N_1N_2P/(m_1m_2)+N_2N_3P/(m_2m_3))$ symbols, and the fusion node receives $O(N_1 N_3 m_2)$ symbols from the successful worker nodes. 

\vspace{2mm}
 
\begin{remark}
With the reasonable assumption that the dimensions of the input matrices $\mathbf{A},\mathbf{B}$ are large enough such that $N_1,N_2,N_3 \gg m_1,m_2,m_3,P$, we can conclude that the encoding and decoding costs at the master and fusion nodes, respectively, are negligible compared to the computation cost at each worker node.
\end{remark}

\begin{figure*}[!t]
    \centering
    \includegraphics[scale=0.45]{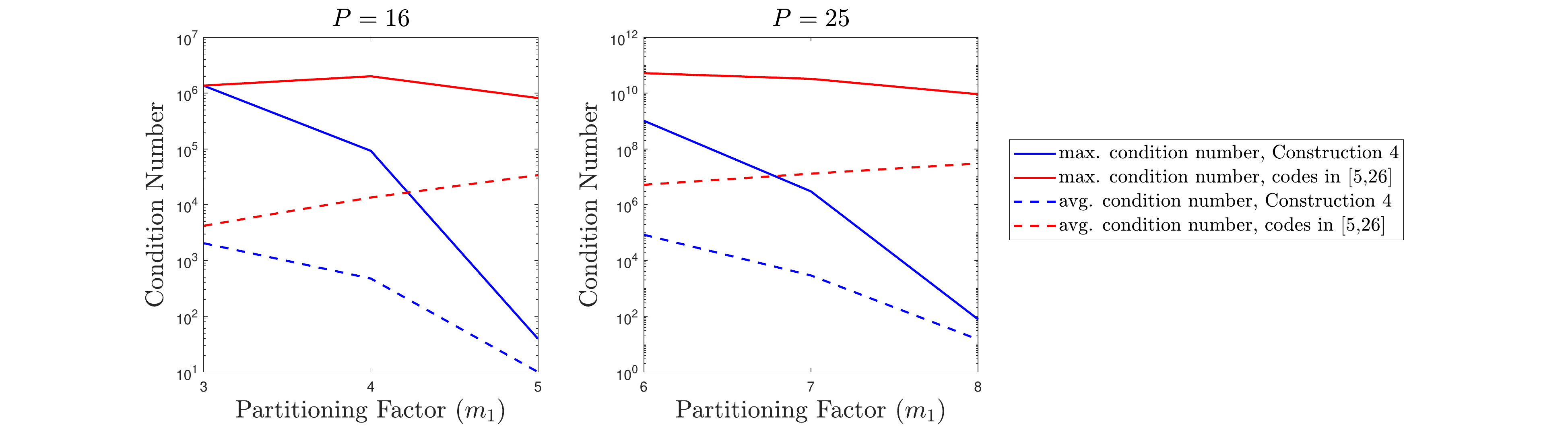}
    \caption{Comparison between the condition number of the interpolating matrix of \fbl{the Generalized OrthoMatDot Codes} and the monomial-based codes \cite{genPolyDot, entPoly} in two distributed systems, one with $16$ worker nodes and the other with $25$ worker nodes, at different partitioning factors $m_1$.}
    \label{fig:Pallgen}
  \noindent\rule{\textwidth}{0.7pt}
\end{figure*}
\subsection{Numerical Results}
\label{sec:numerical_polydot}
In our experiments on Construction \ref{con:gencheb}, we considered distributed systems with $P=16,25$ worker nodes. 
Fig. \ref{fig:Pallgen} shows that, for every examined system, the condition number of the interpolation matrix using the \fbl{Generalized OrthoMatDot} Codes is less than its   counterpart codes in \cite{genPolyDot,entPoly}. The results in Fig. \ref{fig:Pallgen} also show that, for the same system, as the partitioning factor $m_1$ decreases (i.e., as the redundancy in worker nodes increases), the stability of \fbl{the Generalized OrthoMatDot code construction} decreases; however, it is still better than the monomial-basis based codes in any cases.




\section{Numerically Stable Lagrange Coded Computing}\label{sec:lcc}
In this section, we study the numerical stability of Lagrange coded computing \cite{yu2018lagrange} that  lifts coded computing beyond matrix-vector and matrix-matrix multiplications, to multi-variate polynomial computations. As shown in \cite{yu2018lagrange}, Lagrange coded computing has applications in gradient coding, privacy and secrecy. Our main contribution here is to develop a numerically stable approach towards Lagrange coded computing inspired by our result of Theorem \ref{thm:bound}. In particular, our contribution involves (a) careful choice of evaluation points, and (b) a careful decoding algorithm that involves inversion of the appropriate Chebyshev Vandermonde matrix. We describe the system model in Section \ref{sec:lccsysprob}. We overview the Lagrange coded computing technique of \cite{yu2018lagrange} in Section  \ref{sec:Lagrange}. We describe our numerically stable approach in Section \ref{sec:Lagrange_stable}, and present the results of our numerical experiments  in Section \ref{sec:Lagrange_numerical}.

\subsection{System Model and Problem Formulation}\label{sec:lccsysprob}

We consider, for this section, the distributed computing framework depicted in Fig. \ref{fig:lccsys}, that is used in \cite{yu2018lagrange} and consists of a master node, $P$ worker nodes, and a fusion node where the only communication allowed  is from the master node to the different worker nodes and from the worker nodes to the fusion node.
The worker nodes have a prior knowledge of a polynomial function of interest $f: \mathbb{R}^d \rightarrow \mathbb{R}^v$ of degree $\operatorname{deg}(f)$, where $d,v\in\mathbb{N}^{+}$. In addition, the master node possesses a set of data points $\mathcal{X}=\{X_1, \cdots, X_m\}$, where $X_i \in \mathbb{R}^d$, $i \in [m]$. For every worker node $i \in [P]$, the master node is allowed to send some encoded vector $\tilde{X}_i(X_1,\cdots, X_m) \in \mathbb{R}^d$. Once a worker node receives the encoded vector on its input, it evaluates $f$ at this encoded vector and sends the evaluation to the fusion node. That is, for $i\in [P]$, worker node $i$ receives $\tilde{X}_i$ on its input, evaluates $f(\tilde{X}_i)$, then it sends the result to the fusion node. Finally, the fusion node is expected to numerically stably  decode the set of evaluations $\mathcal{F}=\{f(X_1), \cdots, f(X_m)\}$ after it receives the output of any $K$ worker nodes. 
\begin{figure}[t]
    \centering
    \includegraphics[scale=0.50]{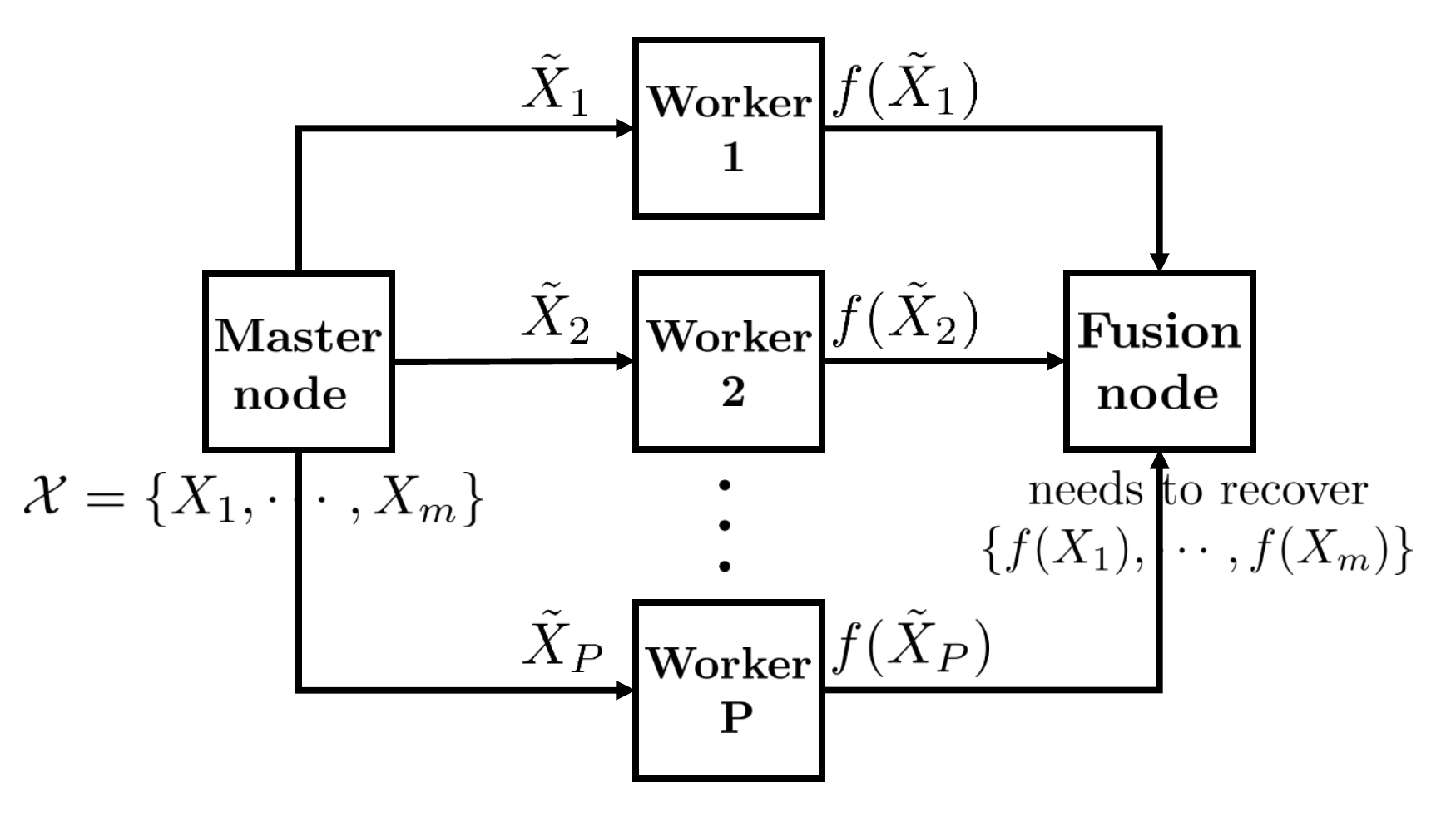}
    \caption{The Lagrange coded computing system framework}
    \label{fig:lccsys}
\end{figure}

\subsection{Background on Lagrange Coded Computing}
\label{sec:Lagrange}

In this section, we review the baseline Lagrange coded computing method introduced in \cite{yu2018lagrange} considering the framework in Section \ref{sec:lccsysprob}. Notice that although the method in \cite{yu2018lagrange} is more general, here, for simplicity, we limit our discussion to the \emph{systematic} Lagrange coded computing. That is, we assume that for $i\in [m]$, worker node $i$ receives the $i$-th data point from the master node. In other words, we assume that $\tilde{X}_i=X_i, i\in[m]$. Now, the encoding procedure goes as follows: First, let $x_1, \cdots, x_P$ be distinct real values, an encoding function $g(x)$ is defined as:
\begin{align}\label{eq:lagenc}
g(x)=\sum_{i=1}^{m} X_i \prod_{j\in[m]-i} \frac{x-x_j}{x_i-x_j}.
\end{align}
Given this encoding function, the master node sends the encoded vector $\tilde{X}_i=g(x_i)$ to the worker node $i$, for every $i \in [P]$. Notice that the encoding function $g(x)$ indeed leads to a systematic encoding since $\tilde{X}_i=g(x_i)=X_i,$ for all $i\in [m]$.  Every worker node $i$ computes $f(\tilde{X}_i)$ upon the reception of $\tilde{X}_i$, and sends the result to the fusion node. The fusion node waits till receiving the output of any $K:=(m-1)\deg(f)+1$. \bl{ Since $f(g(x))$ has degree $(m-1)\deg(f)$  in $x$, the fusion node is able to interpolate $f(g(x))$ after receiving the outputs of any $(m-1)\deg(f)+1$, i.e., $K$, worker nodes.  Since $g(x_i)=X_i, i\in [m]$, the fusion nodes evaluates $\{f(g(x_1)), \cdots, f(g(x_m))\}$ to obtain $\{f(X_1), \cdots, f(X_m)\}.$ }
\subsection{Numerically Stable Lagrange Coded Computing} 
\label{sec:Lagrange_stable}
\bl{Lagrange coded computing requires performing an interpolation at the fusion node to recover the polynomial $f(g(x))$. Performing the interpolation by obtaining the coefficients of the polynomial in a monomial basis requires inverting a square Vandermonde matrix which is numerically  unstable.  Noting that the first $\ell$ Cheybshev polynomials also forms a basis for degree $\ell-1$ polynomials, we provide an alternative decoding procedure whose key idea is to find the coefficients of polynomial $f(g(x))$ in the basis of Chebyshev polynomials. Thereby, our decoding procedure involves inverting the Chebyshev-Vandermonde matrix\footnote{Since both systematic and non-systematic Lagrange coded computing require the inversion of the same Chebyshev-Vandermonde matrix, our numerically stable decoding procedure in Construction \ref{con:Lag} naturally extends to non-systematic Lagrange coded computing, with the only difference is in the last step of evaluating $f(g(x))$ at $x_1, \cdots, x_m$, where in the non-systematic case, $f(g(x))$ is instead evaluated at some predefined values $y_1, \cdots, y_m$ such that $g(y_i)=X_i$ for all $i\in[m].$}. Guided by Theorem \ref{thm:bound}, we choose the evaluation points to be the $P$-point Chebyshev grid $\boldsymbol{\rho}^{(P)}$ to obtain a decoding procedure that is more stable than one that uses the monomial basis.}

Our numerically stable algorithm for Lagrange coded computing is formally described in Construction \ref{con:Lag}.
\bl{
In the following, we explain the notation used in Construction \ref{con:Lag}. 
We let the  polynomial at the $i$-th entry of $f(g(x))$ be denoted $f^{(i)}{(x)}$ and written as $f^{(i)}(x)=\sum_{l=0}^{K-1} c_l^{(i)} T_{l}(x)$. Following the notation in Section \ref{sec:notation}, we use the  {Chebyshev-Vandermonde} matrices $\mathbf{G}^{(K,P)}(\boldsymbol{\rho}^{(P)})$, and $\mathbf{G}^{(K,P)}_{\mathcal{R}}(\boldsymbol{\rho}^{(P)})$, for any subset $\mathcal{R}=\{{r_1}, \cdots, {r_{K}}\}  \subset [P]$, we also define the matrix $\mathbf{G}^{(K,P)}_{[m]}(\boldsymbol{\rho}^{(P)})$.} Finally, we assume that our construction returns as output the set of evaluations $\hat{\mathcal{F}}=$ $\{\hat{f}(X_1),$   $\cdots, \hat{f}(X_m)\}$, where for each $\hat{f}(X_i), i\in [m]$, we have $\hat{f}(X_i)= (\hat{f}^{(1)}(x_i), \cdots, \hat{f}^{(v)}(x_i))$, where for every $i\in[m], j \in [v], \hat{f}^{(j)}(x_i)$ and ${f}^{(j)}(x_i)$ would be  the same if the machine had infinite precision.  

In the following, we show through numerical experiments the stability of our proposed Construction \ref{con:Lag}.
\begin{algorithm}
\caption{Numerically Stable Lagrange Coded Computing \textbf{Inputs}:$f, \mathcal{X}=\{X_1, \cdots, X_m\}$ ,~~\textbf{Output}: $\hat{\mathcal{F}}=\{\hat{f}(X_1), \cdots, \hat{f}(X_m)\}$} \label{con:Lag}
\begin{algorithmic}[1]
\Procedure{MasterNode}{$\mathcal{X}$}\Comment{The master node's procedure}
\State $r \gets 1$
\While{$r\not=P+1$}
\If{$r \in [m]$}
\State $\tilde{X}_r \gets X_i$
\Else
\State $\tilde{X}_r \gets\sum_{i=1}^{m} X_i \prod_{j\in[m]-i} \frac{\rho^{(P)}_r-\rho^{(P)}_j}{\rho^{(P)}_i-\rho^{(P)}_j}$
\EndIf
\State \textbf{send} $\tilde{X}_r$ \textbf{to worker node} $r$
\State $r \gets r+1$
\EndWhile
\EndProcedure
\State
\Procedure{WorkerNode}{$f, \tilde{X}_r$}\Comment{The procedure of worker node $r$}
\State $\Out_r \gets f(\tilde{X}_r)$
\State \textbf{send} $\Out_r$ \textbf{to the fusion node}
\EndProcedure
\State
\Procedure{FusionNode}{$\Out_{r_1}, \cdots, \Out_{r_K}$}\Comment{The fusion node's procedure, $r_i$'s are distinct}
\State $\mathbf{G}_{\text{inv}} \gets \left(\mathbf{G}_{\mathcal{R}}^{(K,P)}\right)^{-1}$ 
\For{$i\in [v]$}
\State $(c^{(i)}_0, \cdots, c^{(i)}_{K-1}) \gets (\Out^{(i)}_{{r_1}}, \cdots,\Out^{(i)}_{{r_{K}}})\mathbf{G}_{\operatorname{inv}}$
\State $(\hat{f}^{(i)}(x_1), \cdots, \hat{f}^{(i)}(x_m))\gets (c^{(i)}_0, \cdots, c^{(i)}_{K-1}) \mathbf{G}^{(K,P)}_{[m]}$ 
\EndFor
\State \textbf{return} $\hat{\mathcal{F}}$
\EndProcedure
\end{algorithmic}
\end{algorithm}
\subsection{Numerical Results}
\label{sec:Lagrange_numerical}

In our experiments, we assume that we have a distributed system of $P$ worker nodes, $m=P-2$ data points/input vectors $X_1, \cdots, X_m$, each of them is of dimension $d=10$, where each entry of every input vector is picked independently, according to the standard  Gaussian distribution $\mathcal{N}(0,1)$.  The function of interest in this system is $f(X)=Y^TX$, where $Y$ is some $d$-dimensional vector with entries picked independently according to the standard Gaussian distribution $\mathcal{N}(0,1)$. In our experiments, we compare between Construction \ref{con:Lag}, where the Chebyshev basis is used for interpolation, and the case where the monomial basis is used for interpolation instead. Let $\hat{\mathbf{f}}=(\hat{f}(X_1) \cdots \hat{f}(X_m))$ be the system's output vector, and ${\mathbf{f}}=({f}(X_1) \cdots {f}(X_m))$ be the correct output vector, we define the  relative error between $\mathbf{f}$ and $\hat{\mathbf{f}}$  to be 
\begin{align}
    E_r(\mathbf{f},\hat{\mathbf{f}})=\frac{||\mathbf{f}-\hat{\mathbf{f}}||_2}{||\mathbf{f}||_2}.
\end{align}
\begin{figure*}[!t]
    \centering
    \includegraphics[scale=0.45]{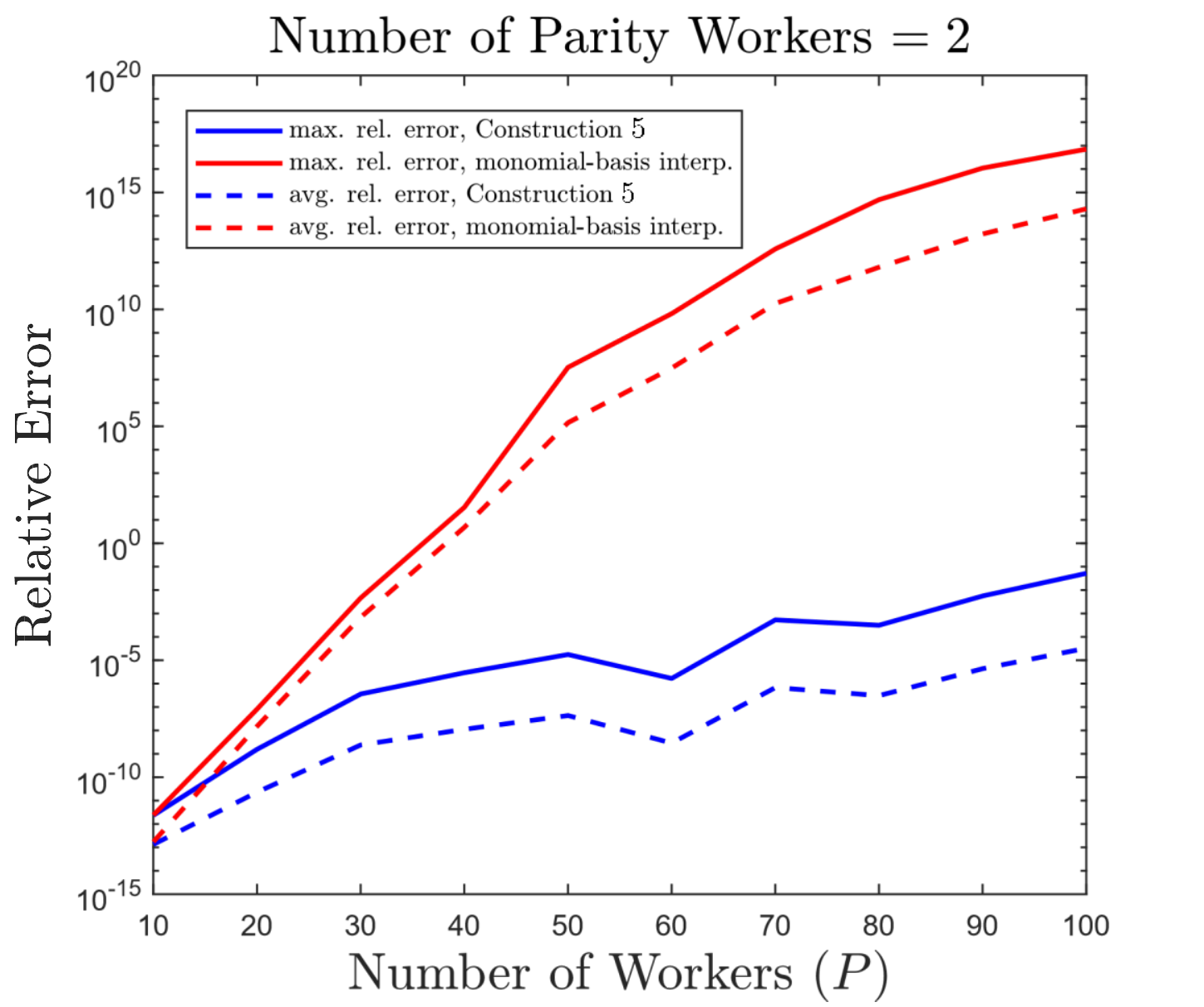}
    \caption{The growth of the relative error, for  Construction \ref{con:Lag}, using both Chebyshev basis interpolation and monomial basis interpolation, both using Chebyshev points, with the system size given a fixed number of redundant worker nodes equals 2.}
    \label{fig:lagfixpar}
  \rule{\textwidth}{0.7pt}
\end{figure*}
The results, shown in Fig. \ref{fig:lagfixpar}, illustrates that using the Chebyshev basis for interpolation provides less relative error/higher stability than the monomial basis at every system size. Fig. \ref{fig:lagfixpar} also shows that under a certain relative error constraint, Construction \ref{con:Lag} provides higher scalability than the monomial basis case. Specifically, let us assume that a relative error up to $0.1$ can be tolerated, Fig. \ref{fig:lagfixpar} shows that the monomial-basis interpolation construction can support systems with a number of worker nodes only less than $40$. However, for the same relative error constraint, Construction \ref{con:Lag} can support systems with a number of worker nodes up to $100$. 

\section{Concluding Remarks}\label{sec:conc}

In this paper, we develop numerically stable codes for matrix-matrix multiplication and Lagrange coded computing. A distinctive character of our work is the infusion of principles of numerical approximation theory into coded computing towards the end goal of numerical stability. In particular, our work is marked by the use of orthogonal polynomials for encoding, Gauss quadrature techniques for decoding and new bounds on the condition number of Chebyshev Vandermonde matrices. Notably, our constructions obtain the same recovery threshold as MatDot Codes and Polynomial Codes for matrix multiplication as well as for Lagrange Coded Computing. However, our construction in Section \ref{sec:polydot-chebyshev} obtains a weaker (higher) recovery threshold than previous constructions \cite{entPoly, genPolyDot} for the problem of coded matrix multiplication when the computation/communication cost is constrained to be lower than that of MatDot Codes. The search of numerically stable codes for this application with the same recovery threshold as \cite{entPoly, genPolyDot} remains open.  

While our paper focuses on applications where polynomial based encoding are particularly useful, our results might be useful for other applications as well. For instance, for the simple matrix-multiplication problem $\mathbf{A}\mathbf{x}$ performed in a distributed setting over $P$ worker nodes, where the goal is to encode $\mathbf{A}$ such that  each worker stores a partition $1/m$ of matrix $\mathbf{A},$ it is well known that MDS type codes can be used \cite{Huang_TC_84, lee2016speeding}. Specifically, let $\mathbf{A}=\begin{bmatrix}\mathbf{A}_{1} \\ \mathbf{A}_{2} \\ \vdots \\ \mathbf{A}_{m}\end{bmatrix}$ and let $\mathbf{H}=(h_{ij})$ be an $m \times P$ matrix where every $m \times m$ submatrix of $\mathbf{H}$ has a full rank of $m$. Then the $p$-th worker for $p \in \{1,2,\ldots,P\}$ can compute $\left(\sum_{i=1}^{m} {h}_{ip} \mathbf{A}_{i}\right)\mathbf{x};$ the product $\mathbf{A}\mathbf{x}$ can be recovered from any $m$ of the $P$ nodes. The instinctual, Reed-Solomon inspired solution of choosing $\mathbf{H}$ to be a Vandermode matrix  is ill-conditioned over real numbers. Note however that, unlike the matrix multiplication problem, the matrix $\mathbf{H}$ does not need to have a polynomial structure. Indeed, choosing $\mathbf{H}$ to be a random Gaussian matrix leads to well-conditioned solutions with high probability. In particular, the following result follows from elementary arguments that build on \cite{azais2004}.

\begin{theorem}
Let $\mathbf{H}$ be an $m \times P$ matrix, $P\geq m\geq 3$, and let the entries of $\mathbf{H}$ be independent and identically distributed standard Gaussian random variables. Then,    $$\operatorname{Pr}\Big(\kappa_2^{max}(\mathbf{H}) > m P^{2(P-m)}\big)\Big) <\frac{5.6}{P^{(P-m)}}.$$
\label{thm:iidG}
\end{theorem}
The theorem which is proved in Appendix \ref{app:iidG}, formally demonstrates that for a fixed number of redundant workers $s=P-m,$ the worst case condition number grows as $O(mP^{2s})$ with high probability.   However, the random Gaussian matrix approach has two drawbacks: (i) for a given realization of the random variables, it is difficult to verify whether it is well-conditioned, and (ii) the lack of structure could lead to more complex decoding. Our result of Theorem \ref{thm:bound} also indicates that choosing $\mathbf{H}=\mathbf{G}^{(m,P)}(\boldsymbol{\rho}^{(P)}),$ i.e., to be a Chebyshev Vandermonde matrix, naturally provides a well-conditioned solution to this problem. Another solution for the matrix-vector multiplication problem is provided in \cite{ramamoorthy2019universally} via universally decodable matrices \cite{ganesan2007existence}; in this work numerical stability is demonstrated empirically. 

It is, however, important to note that the problems resolved in our paper here are more restrictive since matrix multiplication codes - where both matrices are to be encoded so that the product can be recovered - require much more structure than matrix-multiplication where only one matrix is to be encoded. For instance, random Gaussian encoding does not naturally work for matrix multiplication to get a recovery threshold of $2m-1$, and it is not clear whether the solution of  \cite{ramamoorthy2019universally} is applicable either. The utility of Chebyshev-Vandermonde matrices for a variety of coded computing problems including matrix-vector multiplication, matrix multiplication and Lagrange coded computing motivates the study of low-complexity decoding and error correction mechanisms for these systems.

\bibliographystyle{IEEEtran}
\bibliography{IEEEabrv,sample}

\appendices
\section{Proof of Claim \ref{cl:orthAB}}\label{app:ortho}
We have,
\begin{align}\label{eq:clorth}
\int_{a}^{b} p_{\mathbf{A}}(x) p_{\mathbf{B}}(x) w(x) dx 
&=\int_{a}^{b} \left(\sum_{i=0}^{m-1} \mathbf{A}_i q_{i}(x) \right)\left(\sum_{j=0}^{m-1} \mathbf{B}_j q_{j}(x)\right) w(x) dx\notag\\
&=\int_{a}^{b} \sum_{i=0}^{m-1} \sum_{j=0}^{m-1} \mathbf{A}_i \mathbf{B}_j q_{i}(x)  q_{j}(x) w(x) dx\notag\\
&= \sum_{i=0}^{m-1} \sum_{j=0}^{m-1} \mathbf{A}_i \mathbf{B}_j\int_{a}^{b} q_{i}(x)  q_{j}(x) w(x) dx\notag\\
&=\sum_{i=0}^{m-1} \sum_{j=0}^{m-1} \mathbf{A}_i \mathbf{B}_j\ \langle q_{i},  q_{j}\rangle \notag\\
&=\sum_{i=0}^{m-1} \mathbf{A}_i \mathbf{B}_i\  \notag\\
&=\mathbf{A}\mathbf{B}.
\end{align}
In addition, noting that $p_\mathbf{A}(x)p_\mathbf{B}(x)$ (i.e., $p_\mathbf{C}(x)$) is of degree $2m-2$ (less than $2m$), Theorem \ref{thm:GQ} implies that 
\begin{align}\label{eq:clorth2}
\int_{a}^{b}   p_\mathbf{A}(x)p_\mathbf{B}(x) w(x) dx &= \sum_{r=1}^{m} a_r p_{\mathbf{A}}(\eta_r) p_{\mathbf{B}}(\eta_r)    \notag\\
&=\sum_{r=1}^{m} a_r p_{\mathbf{C}}(\eta_r). 
\end{align}
Finally, combining (\ref{eq:clorth}) and  (\ref{eq:clorth2}) completes the proof.\qedw

\section{Proof of Theorem \ref{thm:bound}}\label{app:proofgen}
We use the following trigonometric identity in our proof.

\begin{lemma}
For $n \geq 0$, let $x_i$ be chosen as (\ref{eq:Chebnodes}). Then
$\prod_{j \neq i} (x_i - x_j) = (-1)^{i-1} \frac{2^{1-n}n}{\sin(\frac{(2i-1) \pi}{2n})}$
\label{lem:trig}
\end{lemma}

\begin{proof}
Note that $2^{n-1}\prod_{i=1}^{n}(x-x_i) = T_n(x) =  \cos(n\cos^{-1}(x))$. Therefore, 
$$2^{n-1}\prod_{j \neq i} (x_i - x_j) = T_n'(x_i) = \frac{n}{ \sqrt{1-x_i^2}} \sin(n \cos^{-1}(x_i)) $$
where $T_n'(x)$ denotes the derivative of $T_n(x).$ Using $x_i=\cos(\frac{(2i-1) \pi}{2n})$ above we get the desired result.
\end{proof}
\begin{proof}[Proof of Theorem \ref{thm:bound}:]
 {We  show that any square sub-matrix of $\mathbf{G}^{(n-s,n)}(\boldsymbol{\rho}^{(n)})$ formed by any $n-s$ columns of $\mathbf{G}^{(n-s,n)}(\boldsymbol{\rho}^{(n)})$ satisfies the bound stated in the theorem.} 
 {Let $\mathcal{S}$ be a subset of $[n]$ such that $|\mathcal{S}|=s$, for some $s<n$, and define $\mathbf{G}^{(n-s,n)}_{[n]-\mathcal{S}}(\boldsymbol{\rho}^{(n)})$ to be  the square $n-s \times n-s$ submatrix of $\mathbf{G}^{(n-s,n)}(\boldsymbol{\rho}^{(n)})$ after removing the columns with indices in $\mathcal{S}$. Recalling the structure of $\mathbf{G}^{(n-s,n)}(\boldsymbol{\rho}^{(n)})$ from (\ref{eq:GT}), we can write it as 
\begin{align}
  \mathbf{G}^{(n-s,n)}(\boldsymbol{\rho}^{(n)})= \left(\hspace{-2mm}\begin{array}{ccc} T_0(\rho^{(n)}_1)&  \cdots & T_{0}(\rho^{(n)}_n)\\\vdots & \ddots &\vdots \\ 
   T_{n-s-1}(\rho^{(n)}_{1})&\cdots& T_{n-s-1}(\rho^{(n)}_{n})\end{array}\hspace{-2mm}\right).\notag
\end{align}
Moreover, for any $\mathcal{S}\subset[n]$ such that $|\mathcal{S}|=s$, we can write 
\begin{align}
 \mathbf{G}_{[n]-\mathcal{S}}^{(n-s,n)}(\boldsymbol{\rho}^{(n)}):= \mathbf{G}_{\Gamma}^{(n-s,n)}:= \left(\hspace{-2mm}\begin{array}{ccc} T_0(\gamma_1)&  \cdots & T_{0}(\gamma_{n-s})\\\vdots & \ddots &\vdots \\ 
   T_{n-s-1}(\gamma_{1})&\cdots& T_{n-s-1}(\gamma_{n-s})\end{array}\hspace{-2mm}\right),\notag
\end{align}
 where ${\Gamma}=(\gamma_1, \gamma_2, \cdots, \gamma_{n-s})=(\rho^{(n)}_{g_1}, \rho^{(n)}_{g_2}, \cdots, \rho^{(n)}_{g_{n-s}})$, where $\{g_i\}_{i\in[n-s]}=[n]-\mathcal{S}$ and $g_1 < g_2 < \cdots < g_{n-s}$.}
 {Now, notice that $||\mathbf{G}_{\Gamma}^{(n-s,n)}||^2_{F}=\sum_{i=1}^{n-s}\sum_{j=1}^{n-s} |T_{i-1}(\gamma_j)|^2$, and $|T_i(\gamma_j)| \leq 1$ for any $i,j \in [n-s]$.  Therefore, we have 
\begin{align}\label{eq:GT2}
||\mathbf{G}_{\Gamma}^{(n-s,n)}||^2_{F} \leq (n-s)^2.
\end{align}
In the following, we obtain an upper bound on $||(\mathbf{G}_{\Gamma}^{(n-s,n)})^{-1}||_{F}$. Let $L_{\Gamma,k}$ be the $k$-th Lagrange polynomial associated with $\Gamma$, that is, 
\begin{align}\label{eq:kLag}
L_{\Gamma,k}(x)= \prod_{i\in[n-s]-\{k\}} \frac{x-\gamma_i}{\gamma_k-\gamma_i}     
\end{align}
Since $L_{\Gamma,k}(x)$ has a degree of $n-s-1$, it can be written in terms of the Chebyshev basis $T_{0}(x), \cdots, T_{n-s-1}(x)$ as 
\begin{align}\label{eq:kLagCh}
L_{\Gamma,k}(x)=\sum_{i=0}^{n-s-1} a_{i,k} T_{i}(x),    
\end{align}
for some real coefficients $a_{0,k}, \cdots, a_{n-s-1,k}$. Now, from (\ref{eq:kLag}), note the following property regarding $L_{\Gamma,k}(x)$:
\begin{align}
L_{\Gamma,k}(x)=\left\{\begin{array}{cl} 1, &\text{if } x=\gamma_k\\0, &\text{if } x\in\{\gamma_i\}_{i\in[n-s]-{k}}.\end{array}\right. \notag    
\end{align}
Using this property and observing (\ref{eq:kLagCh}), we conclude that, for any $j \in [n-s]$, $\sum_{i=0}^{n-s-1} a_{i,k} T_{i}(\gamma_j)=\delta(k-j)$. Therefore,
\begin{align}\left(
\begin{array}{ccc}
a_{0,1}    &\cdots & a_{n-s-1,1}  \\
\vdots    &\ddots & \vdots \\
a_{0,n-s}    & \cdots & a_{n-s-1,n-s}
\end{array}\right)    \mathbf{G}_{\Gamma}^{(n-s,n)}=\mathbf{I}_{n-s},\notag
\end{align}
where $\mathbf{I}_{n-s}$ is the $n-s\times n-s$ identity matrix. That is, 
\begin{align}
\left(\mathbf{G}_{\Gamma}^{(n-s,n)}\right)^{-1}=\left(
\begin{array}{ccc}
a_{0,1}    &\cdots & a_{n-s-1,1}  \\
\vdots    &\ddots & \vdots \\
a_{0,n-s}    & \cdots & a_{n-s-1,n-s}
\end{array}\right),
\end{align}
Therefore, 
\begin{align}\label{eq:inv1}
\Big| \Big|\left(\mathbf{G}_{\Gamma}^{(n-s,n)}\right)^{-1}\Big|\Big|_F^2=\sum_{i=1}^{n-s}\sum_{j=1}^{n-s} |a_{i-1,j}|^2.
\end{align}
In addition, we have that
\begin{align}\label{eq:inv2}
\sum_{k=1}^{n-s} \int_{-1}^{1} L^2_{\Gamma,k}(x) w(x )dx&=  \sum_{k=1}^{n-s} \int_{-1}^{1} \sum_{i=0}^{n-s-1}\sum_{j=0}^{n-s-1} a_{i,k}a_{j,k} T_{i}(x)T_{j}(x) w(x)dx\notag\\
&=\sum_{k=1}^{n-s}  \sum_{i=0}^{n-s-1}\sum_{j=0}^{n-s-1} a_{i,k}a_{j,k}\int_{-1}^{1} T_{i}(x)T_{j}(x) w(x)dx\notag\\
&=\sum_{k=1}^{n-s}  \sum_{i=0}^{n-s-1}\sum_{j=0}^{n-s-1} a_{i,k}a_{j,k} \langle T_{i}, T_{j}\rangle \notag\\
&=\sum_{k=1}^{n-s}  \sum_{i=0}^{n-s-1} |a_{i,k}|^2.
\end{align}
From (\ref{eq:inv1}) and (\ref{eq:inv2}), we conclude that 
$||(\mathbf{G}_{\Gamma}^{(n-s,n)})^{-1}||_F^2=\sum_{k=1}^{n-s} \int_{-1}^{1} L^2_{\Gamma,k}(x) w(x )dx$.}

 Now, we express the integral $\int_{-1}^{1} L^2_{\Gamma,k}(x) w(x )dx$ in the Gauss quadrature form using  the $n$ roots of $T_n(x):$ $\rho_1^{(n)}, \cdots, \rho_n^{(n)}$. Note that this is a ``trick'' we use in the proof - it is possible to use the Gauss quadrature formula over $n-s$ nodes to express the integral of the degree $2(n-s-1)$ polynomial $L^2_{\Gamma,k}(x)$. However, the use of $n$ nodes instead of $n-s$ nodes leads to simple tractable bound for $||(\mathbf{G}_{\Gamma}^{(n-s,n)})^{-1}||^2_F.$
Now, we can write
\begin{align}
    \int_{-1}^{1} L^2_{\Gamma,k}(x) w(x )dx=\sum_{i=1}^{n} c_i L_{\Gamma,k}^2(\rho_i^{(n)}),
\end{align}
for some constants $c_1, \cdots, c_n$. Moreover, $c_1, \cdots, c_n$ for the Chebyshev polynomials of the first kind are, in fact, all equal to $\pi /n$. Therefore, we have   
\begin{align}\label{eq:kthLSum}
    \int_{-1}^{1} L^2_{\Gamma,k}(x) w(x )dx=\frac{\pi}{n}\sum_{i=1}^{n}  L_{\Gamma,k}^2(\rho_i^{(n)}),
\end{align}
and, consequently, 
\begin{align}\label{eq:kthLSum2}
    \Big|\Big|\left(\mathbf{G}_{\Gamma}^{(n-s,n)}\right)^{-1}\Big|\Big|_F^2=\frac{\pi}{n}\sum_{k=1}^{n-s} \sum_{i=1}^{n} L^2_{\Gamma,k}(\rho_i^{(n)}).
\end{align}
Now, from (\ref{eq:kLag}), note that  $L_{\Gamma,k}(x)$ has the following evaluations
\begin{align}
L_{\Gamma,k}(\rho_{i}^{(n)})=\left\{\begin{array}{cl} 1, &\text{if } i=g_k\\0, &\text{if } i\in\{g_i\}_{i\in[n-s]-{k}}\\\prod_{j\in[n-s]-\{k\}} \frac{\rho_{i}^{(n)}-\gamma_j}{\gamma_k-\gamma_j}   , &\text{if } i\in \mathcal{S}.\end{array}\right. 
\end{align}
Therefore, (\ref{eq:kthLSum2}) can be written as 
\begin{align}\label{eq:kthLSum3}
    \Big|\Big|\left(\mathbf{G}_{\Gamma}^{(n-s,n)}\right)^{-1}\Big|\Big|_F^2&=\frac{\pi}{n}\sum_{k=1}^{n-s}\left( 1+\sum_{i\in \mathcal{S}} \prod_{j\in[n-s]-\{k\}} \left(\frac{\rho_{i}^{(n)}-\gamma_j}{\gamma_k-\gamma_j}\right)^2\right)\notag\\
    &=\frac{\pi(n-s)}{n}+\frac{\pi}{n}\sum_{k=1}^{n-s}\sum_{i\in \mathcal{S}} \prod_{j\in[n-s]-\{k\}} \left(\frac{\rho_{i}^{(n)}-\gamma_j}{\gamma_k-\gamma_j}\right)^2
\end{align}
In order to obtain our upper bound on $||(\mathbf{G}_{\Gamma}^{(n-s,n)})^{-1}||_F^2$, in the following, we get an upper bound on the term $\prod_{j\in[n-s]-\{k\}} \left(\frac{\rho_{i}^{(n)}-\gamma_j}{\gamma_k-\gamma_j}\right)^2$ in (\ref{eq:kthLSum3}). Notice that $\prod_{j\in[n-s]-\{k\}} \left(\frac{\rho_{i}^{(n)}-\gamma_j}{\gamma_k-\gamma_j}\right)^2$ can be written as 
\begin{align}\label{eq:product}
    \prod_{j\in[n-s]-\{k\}} \left(\frac{\rho_{i}^{(n)}-\gamma_j}{\gamma_k-\gamma_j}\right)^2&= \prod_{j\in[n-s]-\{k\}} \left(\frac{\rho_{i}^{(n)}-\rho_{g_j}^{(n)}}{\rho_{g_k}^{(n)}-\rho^{(n)}_{g_j}}\right)^2\notag\\
    &=\left[\prod_{j\in[n-s]-\{k\}} \left(\frac{\rho_{i}^{(n)}-\rho_{g_j}^{(n)}}{\rho_{g_k}^{(n)}-\rho^{(n)}_{g_j}}\right)^2\right] \frac{\prod_{j\in \mathcal{S}\cup \{g_k\}-\{i\}}\left(\rho_{i}^{(n)}-\rho_{j}^{(n)}\right)^2}{\prod_{j\in \mathcal{S}\cup \{g_k\}-\{i\}}\left(\rho_{i}^{(n)}-\rho^{(n)}_{j}\right)^2}\notag\\&=\frac{\prod_{j\in[n]-\{i\}}\left(\rho_{i}^{(n)}-\rho_{j}^{(n)}\right)^2}{\prod_{j\in[n-s]-\{k\}}\left(\rho_{g_k}^{(n)}-\rho^{(n)}_{g_j}\right)^2\prod_{j\in \mathcal{S}\cup\{g_k\}-\{i\}}\left(\rho_{i}^{(n)}-\rho^{(n)}_{j}\right)^2}\notag\\
    &\stackrel{}{=}\frac{ \left(2^{1-n}n / {\sin(\frac{(2i-1) \pi}{2n})}\right)^2}{\prod_{j\in[n-s]-\{k\}}\left(\rho_{g_k}^{(n)}-\rho^{(n)}_{g_j}\right)^2\prod_{j\in \mathcal{S}\cup\{g_k\}-\{i\}}\left(\rho_{i}^{(n)}-\rho^{(n)}_{j}\right)^2},
\end{align}
where the last equality follows from Lemma \ref{lem:trig}. Moreover, the product  $\prod_{j\in[n-s]-\{k\}}\left(\rho_{g_k}^{(n)}-\rho^{(n)}_{g_j}\right)^2$ in (\ref{eq:product}) can be written as 
\begin{align}\label{eq:product2}
    \prod_{j\in[n-s]-\{k\}}\left(\rho_{g_k}^{(n)}-\rho^{(n)}_{g_j}\right)^2&=
   \frac{ \prod_{j\in[n-s]-\{k\}}\left(\rho_{g_k}^{(n)}-\rho^{(n)}_{g_j}\right)^2\prod_{j\in \mathcal{S}} \left(\rho^{(n)}_{g_k}-\rho^{(n)}_{j}\right)^2}{\prod_{j\in \mathcal{S}} \left(\rho^{(n)}_{g_k}-\rho^{(n)}_{j}\right)^2}\notag\\
   &= \frac{ \left(2^{1-n}n / {\sin(\frac{(2g_k-1) \pi}{2n})}\right)^2}{\prod_{j\in \mathcal{S}} \left(\rho^{(n)}_{g_k}-\rho^{(n)}_{j}\right)^2},
\end{align}
where the last equality follows from Lemma \ref{lem:trig}. Now, substituting from (\ref{eq:product2}) in (\ref{eq:product}) yields
\begin{align}\label{eq:product3}
    \prod_{j\in[n-s]-\{k\}} \left(\frac{\rho_{i}^{(n)}-\gamma_j}{\gamma_k-\gamma_j}\right)^2&= \frac{ \left({\sin(\frac{(2g_k-1) \pi}{2n})}\right)^2}{\left({\sin(\frac{(2i-1) \pi}{2n})}\right)^2}\frac{\prod_{j\in \mathcal{S}} \left(\rho^{(n)}_{g_k}-\rho^{(n)}_{j}\right)^2}{\prod_{j\in \mathcal{S}\cup\{g_k\}-\{i\}}\left(\rho_{i}^{(n)}-\rho^{(n)}_{j}\right)^2}\notag\\
    &= \frac{ \left({\sin(\frac{(2g_k-1) \pi}{2n})}\right)^2}{\left({\sin(\frac{(2i-1) \pi}{2n})}\right)^2}\frac{\prod_{j\in \mathcal{S}-\{i\}} \left(\rho^{(n)}_{g_k}-\rho^{(n)}_{j}\right)^2}{\prod_{j\in \mathcal{S}-\{i\}}\left(\rho_{i}^{(n)}-\rho^{(n)}_{j}\right)^2},\notag\\
    &\leq\frac{ \left({\sin(\frac{(2g_k-1) \pi}{2n})}\right)^2}{\left({\sin(\frac{(2i-1) \pi}{2n})}\right)^2}\left[\frac{\max_{j\in \mathcal{S}-\{i\}} \left(\rho^{(n)}_{g_k}-\rho^{(n)}_{j}\right)^2}{\min_{j\in \mathcal{S}-\{i\}}\left(\rho_{i}^{(n)}-\rho^{(n)}_{j}\right)^2}\right]^{s-1}\notag\\
    &= \frac{ 1}{\left({\sin(\frac{(2i-1) \pi}{2n})}\right)^2}\left[\frac{4}{ \left(\cos(\frac{\pi}{2n})-\cos(\frac{3\pi}{2n})\right)^2}\right]^{s-1}\notag\\
    &=\frac{4^{s-1}}{\left({\sin(\frac{(2i-1) \pi}{2n})}\right)^2 \left(\cos(\frac{\pi}{2n})-\cos(\frac{3\pi}{2n})\right)^{2(s-1)}}\notag\\
    &=O(4^{s-1}n^{2+4(s-1)}). 
\end{align}
Using (\ref{eq:product3}) in (\ref{eq:kthLSum3}), we conclude that 
\begin{align}\label{eq:invG}
    \Big|\Big|\left(\mathbf{G}_{\Gamma}^{(n-s,n)}\right)^{-1}\Big|\Big|_F^2=O(4^{s-1}(n-s)sn^{1+4(s-1)}).
\end{align}
Finally, combining (\ref{eq:GT2}) and (\ref{eq:invG}), we conclude that $$\kappa^{max}_F(\mathbf{G}^{(n-s,n)}(\boldsymbol{\rho}^{(n)})) = O\left( (n-s)\sqrt{ns(n-s)}\left({2}n^2\right)^{s-1}\right).$$
\end{proof}
\section{Proof of Claim \ref{cl:chebcoef}}\label{app:proof}
Let $\alpha=2m_2-1, \gamma=\alpha (2m_1-1)$. $p_{\mathbf{A}}(x)$ in (\ref{eq:genpolys}) can be written as
\begin{align}
  p_\mathbf{A}(x)&=\sum_{i=0}^{m_1-1}\sum_{j=0}^{m_2-1}\mathbf{A}_{i,j} T^{'}_{m_2-1-j+i\alpha}(x) \notag\\
  &=\notag  \sum_{j=0}^{m_2-2}\mathbf{A}_{0,j} T_{m_2-1-j}(x)+ 1/2~ \mathbf{A}_{0,m_2-1} T_{0}(x)+ \sum_{i=1}^{m_1-1}\sum_{j=0}^{m_2-1}\mathbf{A}_{i,j} T_{m_2-1-j+i\alpha}(x)
\end{align}
Similarly, $p_{\mathbf{B}}(x)$ in (\ref{eq:genpolys}) can be written as
\begin{align}
    p_\mathbf{B}(x)&=\sum_{k=0}^{m_2-1}\sum_{l=0}^{m_3-1} \mathbf{B}_{k,l} T^{'}_{k+l\gamma}(x)\notag\\
    &= 1/2~\mathbf{B}_{0,0} T_0(x) + \sum_{k=1}^{m_2-1} \mathbf{B}_{k,0} T_{k}(x)+\sum_{k=0}^{m_2-1}\sum_{l=1}^{m_3-1} \mathbf{B}_{k,l} T_{k+l\gamma}(x)
\end{align}
Now, the product $p_{\mathbf{A}}(x)p_{\mathbf{B}}(x)$ can be written as 
\begin{align}
    p_\mathbf{A}&(x)p_\mathbf{B}(x)=\frac{1}{2}\big(p_1(x)+p_2(x)\big)
\end{align}
where, 
\begin{align}
p_1(x)=& \sum_{j=0}^{m_2-2}\mathbf{A}_{0,j} \mathbf{B}_{0,0}T_{m_2-1-j}(x)+ \frac{1}{2} \mathbf{A}_{0,m_2-1}\mathbf{B}_{0,0} T_0(x)+\sum_{i=1}^{m_1-1}\sum_{j=0}^{m_2-1}\mathbf{A}_{i,j}\mathbf{B}_{0,0} T_{m_2-1-j+i\alpha}(x)\notag\\
    &+  \sum_{j=0}^{m_2-2}\sum_{k=1}^{m_2-1}\mathbf{A}_{0,j}  \mathbf{B}_{k,0} T_{m_2-1-j+k}(x)+ \sum_{k=1}^{m_2-1}\mathbf{A}_{0,m_2-1}  \mathbf{B}_{k,0} T_{k}(x)\notag\\
    &+\sum_{i=1}^{m_1-1}\sum_{j=0}^{m_2-1}\sum_{k=1}^{m_2-1}\mathbf{A}_{i,j}  \mathbf{B}_{k,0} T_{m_2-1-j+i\alpha+k}(x)+\sum_{j=0}^{m_2-2}\sum_{k=0}^{m_2-1}\sum_{l=1}^{m_3-1}\mathbf{A}_{0,j} \mathbf{B}_{k,l} T_{m_2-1-j+k+l\gamma}(x)\notag\\
    &+\ \sum_{k=0}^{m_2-1}\sum_{l=1}^{m_3-1}\mathbf{A}_{0,m_2-1} \mathbf{B}_{k,l} T_{k+l\gamma}(x)+\sum_{i=1}^{m_1-1}\sum_{j=0}^{m_2-1}\sum_{k=0}^{m_2-1}\sum_{l=1}^{m_3-1}\mathbf{A}_{i,j}  \mathbf{B}_{k,l} T_{m_2-1-j+i\alpha+k+l\gamma}(x)\notag\\
\end{align}
and,
\begin{align}\label{eq:p2}
&p_2(x)=     \sum_{j=0}^{m_2-2}\sum_{k=1}^{m_2-1}\mathbf{A}_{0,j}\mathbf{B}_{k,0} T_{|m_2-1-j-k|}(x) + \sum_{i=1}^{m_1-1}\sum_{j=0}^{m_2-1}\sum_{k=1}^{m_2-1}\mathbf{A}_{i,j}\mathbf{B}_{k,0}  T_{|m_2-1-j+i\alpha-k|}(x)\notag\\
    &+ \sum_{j=0}^{m_2-2}\sum_{k=0}^{m_2-1}\sum_{l=1}^{m_3-1}\mathbf{A}_{0,j}\mathbf{B}_{k,l} T_{|m_2-1-j-k-l\gamma|}(x) + \sum_{i=1}^{m_1-1}\sum_{j=0}^{m_2-1}\sum_{k=0}^{m_2-1}\sum_{l=1}^{m_3-1} \mathbf{A}_{i,j}\mathbf{B}_{k,l} T_{|m_2-1-j+i\alpha-k-l\gamma|}(x).
\end{align}
Now, in order to prove the claim, it suffices to prove the following two statements:
\begin{enumerate}
    \item For any $i\in\{0, \cdots, m_1-1\}, l\in \{0, \cdots, m_3-1\}$, $\mathbf{C}_{i,l}$ is the matrix coefficient of $T_{m_2-1+i\alpha+l\gamma}$ in $p_1(x)$.
    \item For any $i\in\{0, \cdots, m_1-1\}, l\in \{0, \cdots, m_3-1\}$, the matrix coefficient of $T_{m_2-1+i\alpha+l\gamma}$ in $p_{2}(x)$ is $\mathbf{0}_{{N_1}/{m_1}\times{N_3}/{m_3}}$, where  $\mathbf{0}_{{N_1}/{m_1}\times{N_3}/{m_3}}$ is the ${N_1}/{m_1}\times{N_3}/{m_3}$ all zeros matrix. 
\end{enumerate}
\vspace{2mm}
In the following, we prove that statement 1) is true. In order to find the coefficient of $T_{m_2-1+i\alpha+l\gamma}$ in $p_1(x)$, we find the set $\mathcal{S}_1=\{(i',j',k',l'): m_2-1-j'+i'\alpha+k'+l'\gamma=m_2-1+i\alpha+l\gamma\}$. Rewriting $m_2-1-j'+i'\alpha+k'+l'\gamma=m_2-1+i\alpha+l\gamma$, we have
\begin{align}\label{eqn:coef}
    (k'-j')+(i'-i) \alpha+(l'-l)\gamma=0.
\end{align}
(\ref{eqn:coef}) implies that $l'=l$. Suppose $l'\neq l$, this means  that $(k'-j')+(i'-i)\alpha= c \gamma$ for some integer $c$. However, this is a contradiction since $|(k'-j')+(i'-i)\alpha| < \gamma$, for any $i,i',j',k'$. Now, (\ref{eqn:coef}) can be written as 
\begin{align}\label{eqn:coef2}
    (k'-j')+(i'-i) \alpha=0.
\end{align}
Again, (\ref{eqn:coef2}) implies $i'=i$. Suppose $i'\neq i$, this means $k'-j'=c\alpha$, for some integer $c$. However, this is a contradiction since $|k'-j'|<\alpha$. Now, since $i'=i$, (\ref{eqn:coef2}) implies $j'=k'$. Thus, $\mathcal{S}_1=\{(i,j',j',k): j'\in\{0, \cdots, m_2-1\}\}$. That is, for any $i\in\{0, \cdots, m_1-1\}, j\in \{0, \cdots, m_3-1\}$,  the matrix coefficient of $T_{m_2-1+i\alpha+l\gamma}$ in $p_1(x)$ is $\sum_{j'=0}^{m_2-1} \mathbf{A}_{i,j'}\mathbf{B}_{j',l} =\mathbf{C}_{i,l}$.

Now, it remains to prove statement 2). That is,  for any $i\in\{0, \cdots, m_1-1\}, l\in \{0, \cdots, m_3-1\}$, the matrix coefficient of $T_{m_2-1+i\alpha+l\gamma}$ in $p_{2}(x)$ is $\mathbf{0}_{{N_1}/{m_1}\times{N_3}/{m_3}}$. In order to find the coefficient of $T_{m_2-1+i\alpha+l\gamma}$ in $p_2(x)$, we find the sets $\mathcal{S}_2^{(1)}=\{(i',j',k',l'): m_2-1-j'+i'\alpha-k'-l'\gamma=m_2-1+i\alpha+l\gamma\}$, and $\mathcal{S}_2^{(2)}=\{(i',j',k',l'): -m_2+1+j'-i'\alpha+k'+l'\gamma=m_2-1+i\alpha+l\gamma\}$.

First, for the set $\mathcal{S}_2^{(1)}$, rewriting $m_2-1-j'+i'\alpha-k'-l'\gamma=m_2-1+i\alpha+l\gamma$, we get
\begin{align}\label{eq:coef3}
(-j'-k')+(i'-i)\alpha+(l+l')\gamma=0.    
\end{align}
From (\ref{eq:coef3}), we conclude that 
$l+l'=0$. Otherwise, $(-j'-k')+(i'-i)\alpha=c \gamma$, for some integer $c$, a contradiction since $|(-j'-k')+(i'-i)\alpha| < \gamma$. Since $l+l'=0$ and both $l,l'$ are non-negative, we conclude that $l'=l=0$. Moreover, now (\ref{eq:coef3}) reduces to 
\begin{align}\label{eq:coef4}
(-j'-k')+(i'-i)\alpha=0.    
\end{align}
Again, since $|-j'-k'| < \alpha$, we conclude that $i'=i$, which implies that $j'+k'=0$.  Since $j'+k'=0$ and both $j',k'$ are non-negative, we conclude that $j'=k'=0$. Thus, $\mathcal{S}_2^{(1)}=\{(i,0,0,0)\}$. Now, noticing from (\ref{eq:p2}) that $\mathbf{A}_{i,0}\mathbf{B}_{0,0}$ does not contribute to any term in $p_2(x)$, we conclude that the matrix coefficient of  $T_{m_2-1+i\alpha+l\gamma}$ in $p_{2}(x)$ is only due to the  set $\mathcal{S}_2^{(2)}$.  Recall that $\mathcal{S}_2^{(2)}=\{(i',j',k',l'): -m_2+1+j'-i'\alpha+k'+l'\gamma=m_2-1+i\alpha+l\gamma\}$, we rewrite $ -m_2+1+j'-i'\alpha+k'+l'\gamma=m_2-1+i\alpha+l\gamma$ as 
\begin{align}\label{eq:coef5}
    (j'+k'-2m_2+2)-(i'+i)\alpha+(l'-l)\gamma=0
\end{align}
From (\ref{eq:coef5}), we conclude that $l=l'$. Otherwise, $(j'+k'-2m_2+2)-(i'+i)\alpha=c \gamma$, for some integer $c$, a contradiction since $|(j'+k'-2m_2+2)-(i'+i)\alpha| < \gamma$. Moreover, now (\ref{eq:coef5}) reduces to 
\begin{align}\label{eq:coef6}
(j'+k'-2m_2+2)+(i'+i)\alpha=0.    
\end{align}
Again, since $|j'+k'-2m_2+2| < \alpha$, we conclude that $i'+i=0$.  Since $i'+i=0$ and both $i,i'$ are non-negative, we conclude that $i'=i=0$, which implies that $j'+k'=2m_2-2$.  Since $j'+k'=2m_2-2$ and both $j',k'\leq m_2-1$, we conclude that $j'=k'=m_2-1$. Thus, $\mathcal{S}_2^{(2)}=\{(0,m_2-1,m_2-1,l)\}$. Now, noticing from (\ref{eq:p2}) that $\mathbf{A}_{0,m_2-1}\mathbf{B}_{m_2-1,l}$ does not contribute to any term in $p_2(x)$, we conclude that the matrix coefficient of  $T_{m_2-1+i\alpha+l\gamma}$ in $p_{2}(x)$ is $\mathbf{0}_{{N_1}/{m_1}\times{N_3}/{m_3}}$.\qedw

\section{Upper Bound on the Condition Number of Gaussian Matrices}\label{app:iidG}
We first introduce the following theorem from \cite{azais2004}.
\begin{theorem}\label{thm:azais}
Let $\mathbf{A}$ be an $m\times m$ matrix, $m\geq 3$, and let the entries of $\mathbf{A}$ be independent and identically distributed standard Gaussian random variables. Then,  for all $\alpha > 1$, $$\operatorname{Pr}(\kappa_2(\mathbf{A})>m\alpha)<\frac{5.6}{\alpha},$$ where $\kappa_2(\mathbf{A})$ is the condition number of $\mathbf{A}$ with respect to the matrix norm induced by $\ell_2$.
\end{theorem}
As a consequence, in the following, we extend the result in Theorem \ref{thm:azais} to bound the condition number of every $m\times m$ sub-matrix of a random $m\times P$ matrix with $i.i.d$ standard Gaussian entries,  $P \geq m$.
\begin{proof}[Proof of Theorem \ref{thm:iidG}]
For any subset $\mathcal{S} \subseteq \{1,2,\ldots,P\}$, let $\mathbf{H}_{\mathcal{S}}$ denote the $|\mathcal{S}| \times m$ sub-matrix of $\mathbf{H}$ containing the columns  $\mathbf{H}$ corresponding to $\mathcal{S},$  and let $s=P-m$. Then we have
\begin{align}
\operatorname{Pr}\Big(\kappa^{max}_2(\mathbf{H}) > mP^{2s}\Big) &=  \operatorname{Pr}\Big(\underset{\substack{\mathcal{S}'\subset [P], |\mathcal{S}'|=m}}{\bigcup}\big(\kappa_2(\mathbf{H}_{\mathcal{S}'})> m P^{2s}\big)\Big)\\ &\stackrel{(1)}{\leq}\underset{\substack{\mathcal{S}'\subset [P], |\mathcal{S}'|=m}}{\sum}\operatorname{Pr}\big(\kappa_2(\mathbf{H}_{\mathcal{S}'})> m P^{2s}\big)   \notag \\
&\stackrel{}{=}\left(\begin{array}{cc}
     P\\s \end{array}\right)  \operatorname{Pr}\big(\kappa_2(\mathbf{H}_{\mathcal{S}'})> m P^{2s}\big), \text{~for any } \mathcal{S}'\subset [P] \text{~such that  } |\mathcal{S}'|= m\notag\\
&\stackrel{(2)}{<} P^s  \frac{5.6}{P^{2s}}    \notag\\
&=\frac{5.6}{P^s},\notag
\end{align}
where $(1)$ follows from the union bound, and  $(2)$ follows from the fact that   $\left(\hspace{-2mm}\begin{array}{cc}
     P\\s \end{array}\hspace{-2mm}\right) \leq P^s$ and Theorem \ref{thm:azais}.
\end{proof}
\end{document}